\documentclass[a4paper,UKenglish,cleveref, autoref, thm-restate]{lipics-v2021}

\nolinenumbers
\newtheorem*{imptremark*}{Important remark}

\usepackage{xspace}
\usepackage{bbm}
\usepackage{thmtools,thm-restate}
\usepackage{csquotes}
\usepackage{tikz}
\usetikzlibrary{decorations.pathreplacing,hobby,backgrounds,trees,arrows,automata,shapes,positioning,fit,patterns,calc,matrix}

\newtheorem{fact}[theorem]{Fact}

\title{New Separation Algorithms for Level 1 in Group-based Concatenation Hierarchies}

\title{A generic polynomial time approach to separation by first-order logic without quantifier alternation}
\titlerunning{A generic polynomial time approach to separation by alternation-free first-order logic} 

\author{Thomas Place}{LaBRI, Bordeaux University, France \and \url{http://www.labri.fr/perso/tplace}}{tplace@labri.fr}{}{}

\author{Marc Zeitoun}{LaBRI, Bordeaux University, France \and \url{http://www.labri.fr/perso/zeitoun}}{mz@labri.fr}{}{}

\authorrunning{T. Place and M. Zeitoun} 

\Copyright{T. Place and M. Zeitoun} 



\category{} 

\relatedversion{} 

\EventEditors{John Q. Open and Joan R. Access}
\EventNoEds{2}
\EventLongTitle{42nd Conference on Very Important Topics (CVIT 2016)}
\EventShortTitle{CVIT 2016}
\EventAcronym{CVIT}
\EventYear{2016}
\EventDate{December 24--27, 2016}
\EventLocation{Little Whinging, United Kingdom}
\EventLogo{}
\SeriesVolume{42}
\ArticleNo{23}

\begin{document}

\newcommand{\veps}{\ensuremath{\varepsilon}\xspace}
\newcommand{\inv}{^{-1}}
\newcommand{\nat}{\ensuremath{\mathbb{N}}\xspace}
\newcommand{\frB}{\ensuremath{\mathbb{B}}\xspace}
\newcommand{\frP}{\ensuremath{\mathbb{P}}\xspace}
\newcommand{\fri}{\ensuremath{\mathbbm{i}}\xspace}
\newcommand{\frS}{\ensuremath{\mathbbm{S}}\xspace}

\newcommand{\prefsig}[1]{\ensuremath{\mathbb{P}_{#1}}\xspace}


\newcommand{\As}{\ensuremath{\mathcal{A}}\xspace}
\newcommand{\Bs}{\ensuremath{\mathcal{B}}\xspace}
\newcommand{\Cs}{\ensuremath{\mathcal{C}}\xspace}
\newcommand{\Ds}{\ensuremath{\mathcal{D}}\xspace}
\newcommand{\Es}{\ensuremath{\mathcal{E}}\xspace}
\newcommand{\Fs}{\ensuremath{\mathcal{F}}\xspace}
\newcommand{\Gs}{\ensuremath{\mathcal{G}}\xspace}
\newcommand{\Hs}{\ensuremath{\mathcal{H}}\xspace}
\newcommand{\Is}{\ensuremath{\mathcal{I}}\xspace}
\newcommand{\Js}{\ensuremath{\mathcal{J}}\xspace}
\newcommand{\Ks}{\ensuremath{\mathcal{K}}\xspace}
\newcommand{\Ls}{\ensuremath{\mathcal{L}}\xspace}
\newcommand{\Ms}{\ensuremath{\mathcal{M}}\xspace}
\newcommand{\Ns}{\ensuremath{\mathcal{N}}\xspace}
\newcommand{\Os}{\ensuremath{\mathcal{O}}\xspace}
\newcommand{\Ps}{\ensuremath{\mathcal{P}}\xspace}
\newcommand{\Qs}{\ensuremath{\mathcal{Q}}\xspace}
\newcommand{\Rs}{\ensuremath{\mathcal{R}}\xspace}
\newcommand{\Ss}{\ensuremath{\mathcal{S}}\xspace}
\newcommand{\Ts}{\ensuremath{\mathcal{T}}\xspace}
\newcommand{\Us}{\ensuremath{\mathcal{U}}\xspace}
\newcommand{\Vs}{\ensuremath{\mathcal{V}}\xspace}
\newcommand{\Ws}{\ensuremath{\mathcal{W}}\xspace}
\newcommand{\Xs}{\ensuremath{\mathcal{X}}\xspace}
\newcommand{\Ys}{\ensuremath{\mathcal{Y}}\xspace}
\newcommand{\Zs}{\ensuremath{\mathcal{Z}}\xspace}

\newcommand{\Ab}{\ensuremath{\mathbf{A}}\xspace}
\newcommand{\Bb}{\ensuremath{\mathbf{B}}\xspace}
\newcommand{\Cb}{\ensuremath{\mathbf{C}}\xspace}
\newcommand{\Db}{\ensuremath{\mathbf{D}}\xspace}
\newcommand{\Eb}{\ensuremath{\mathbf{E}}\xspace}
\newcommand{\Fb}{\ensuremath{\mathbf{F}}\xspace}
\newcommand{\Gb}{\ensuremath{\mathbf{G}}\xspace}
\newcommand{\Hb}{\ensuremath{\mathbf{H}}\xspace}
\newcommand{\Ib}{\ensuremath{\mathbf{I}}\xspace}
\newcommand{\Jb}{\ensuremath{\mathbf{J}}\xspace}
\newcommand{\Kb}{\ensuremath{\mathbf{K}}\xspace}
\newcommand{\Lb}{\ensuremath{\mathbf{L}}\xspace}
\newcommand{\Mb}{\ensuremath{\mathbf{M}}\xspace}
\newcommand{\Nb}{\ensuremath{\mathbf{N}}\xspace}
\newcommand{\Ob}{\ensuremath{\mathbf{O}}\xspace}
\newcommand{\Pb}{\ensuremath{\mathbf{P}}\xspace}
\newcommand{\Qb}{\ensuremath{\mathbf{Q}}\xspace}
\newcommand{\Rb}{\ensuremath{\mathbf{R}}\xspace}
\newcommand{\Sb}{\ensuremath{\mathbf{S}}\xspace}
\newcommand{\Tb}{\ensuremath{\mathbf{T}}\xspace}
\newcommand{\Ub}{\ensuremath{\mathbf{U}}\xspace}
\newcommand{\Vb}{\ensuremath{\mathbf{V}}\xspace}
\newcommand{\Wb}{\ensuremath{\mathbf{W}}\xspace}
\newcommand{\Xb}{\ensuremath{\mathbf{X}}\xspace}
\newcommand{\Yb}{\ensuremath{\mathbf{Y}}\xspace}
\newcommand{\Zb}{\ensuremath{\mathbf{Z}}\xspace}


\newcommand{\wsuit}{well-suited\xspace}
\newcommand{\Wsuit}{Well-suited\xspace}

\newcommand{\vari}{prevariety\xspace}
\newcommand{\Vari}{Prevariety\xspace}
\newcommand{\varis}{prevarieties\xspace}
\newcommand{\Varis}{Prevarieties\xspace}

\newcommand{\pvari}{positive prevariety\xspace}
\newcommand{\pVari}{Positive prevariety\xspace}
\newcommand{\pvaris}{positive prevarieties\xspace}
\newcommand{\pVaris}{Positive prevarieties\xspace}



\newcommand{\ptime}{\textup{\sc P}\xspace}
\newcommand{\nptime}{\textup{\sc NP}\xspace}
\newcommand{\conptime}{\textup{co-{\sc NP}}\xspace}
\newcommand{\poracle}{\ensuremath{{\text{\sc P}}^{{\text{\sc NP}}}}\xspace}
\newcommand{\exptime}{\textup{\sc EXPTIME}\xspace}
\newcommand{\pexptime}[1]{\textup{#1-}{\sc EXPTIME}\xspace}
\newcommand{\nexptime}{\textup{\sc NEXPTIME}\xspace}
\newcommand{\conexptime}{\textup{co-{\sc NEXPTIME}}\xspace}
\newcommand{\pnexptime}[1]{\textup{#1-}{\sc NEXPTIME}\xspace}
\newcommand{\exporacle}{\ensuremath{{\text{\sc EXPTIME}}^{{\text{\sc NEXPTIME}}}}\xspace}
\newcommand{\pexporacle}[1]{\ensuremath{{\textup{#1-}\text{\sc EXPTIME}}^{\textup{#1-}{\text{\sc NEXPTIME}}}}\xspace}

\newcommand{\logspace}{\textup{\sc LOGSPACE}\xspace}
\newcommand{\nl}{\textup{{\sc NL}}\xspace}
\newcommand{\conl}{\textup{co-{\sc NL}}\xspace}
\newcommand{\pspace}{\textup{\sc PSPACE}\xspace}
\newcommand{\npspace}{\textup{\sc NPSPACE}\xspace}
\newcommand{\expspace}{\textup{\sc EXPSPACE}\xspace}
\newcommand{\pexpspace}[1]{\textup{#1-}{\sc EXPSPACE}\xspace}


\newcommand{\nfa}{{{NFA}}\xspace}
\newcommand{\pnfa}{pre-NFA\xspace}
\newcommand{\dfa}{{{DFA}}\xspace}

\newcommand{\nfas}{{\textup{NFAs}}\xspace}
\newcommand{\pnfas}{{\textup{pre-NFAs}}\xspace}
\newcommand{\dfas}{{\textup{DFAs}}\xspace}

\newcommand{\nfaps}{{\textup{NFA(s)}}\xspace}
\newcommand{\pnfaps}{{\textup{pre-NFA(s)}}\xspace}
\newcommand{\dfaps}{{\textup{DFA(s)}}\xspace}

\newcommand{\lauto}[3]{\ensuremath{L_{#1}(#2,#3)}\xspace}
\newcommand{\alauto}[2]{\lauto{\As}{#1}{#2}}


\newcommand{\dnsep}[2]{\ensuremath{\Is_{#1}[\Ds,#2]}\xspace}
\newcommand{\dnsepca}{\dnsep{\Cs}{\As}}
\newcommand{\gnsep}[2]{\ensuremath{\Is_{#1}[\Gs,#2]}\xspace}
\newcommand{\gnsepca}{\gnsep{\Cs}{\As}}
\newcommand{\gnsepda}{\gnsep{\Ds}{\As}}
\newcommand{\gnsepbpa}{\gnsep{\bpol{\Ds}}{\As}}

\newcommand{\dclos}{\ensuremath{\mathord{\downarrow}}\xspace}
\newcommand{\uclos}{\ensuremath{\mathord{\uparrow}}\xspace}

\newcommand{\dclosp}[1]{\ensuremath{\mathord{\downarrow_{#1}}}\xspace}
\newcommand{\dclosr}{\dclosp{R}}
\newcommand{\dclosq}{\dclosp{Q}}

\newcommand{\imprint}{imprint\xspace}
\newcommand{\imprints}{imprints\xspace}
\newcommand{\Imprint}{Imprint\xspace}
\newcommand{\Imprints}{Imprints\xspace}

\newcommand{\iden}{\veps-approximation\xspace}
\newcommand{\idens}{\veps-approximations\xspace}
\newcommand{\Iden}{\veps-Approximation\xspace}
\newcommand{\Idens}{\veps-Approximations\xspace}

\newcommand{\tame}{multiplicative\xspace}
\newcommand{\Tame}{Multiplicative\xspace}
\newcommand{\wtame}{weak multiplicative\xspace}

\newcommand{\Ratms}{Rating maps\xspace}
\newcommand{\Ratm}{Rating map\xspace}
\newcommand{\ratms}{rating maps\xspace}
\newcommand{\ratm}{rating map\xspace}
\newcommand{\Rata}{Rating algebra\xspace}
\newcommand{\Ratas}{Rating algebras\xspace}
\newcommand{\rata}{rating algebra\xspace}
\newcommand{\ratas}{rating algebras\xspace}

\newcommand{\Nice}{Nice\xspace}
\newcommand{\nice}{nice\xspace}
\newcommand{\effective}{effective\xspace}

\newcommand{\mratm}{multiplicative rating map\xspace}
\newcommand{\mratms}{multiplicative rating maps\xspace}
\newcommand{\Mratm}{Multiplicative rating map\xspace}
\newcommand{\Mratms}{Multiplicative rating maps\xspace}

\newcommand{\pratm}{powerset rating map\xspace}
\newcommand{\pratms}{powerset rating map\xspace}
\newcommand{\Pratm}{Powerset rating map\xspace}
\newcommand{\Pratms}{Powerset rating map\xspace}

\newcommand{\wmratm}{weak multiplicative rating map\xspace}
\newcommand{\wmratms}{weak multiplicative rating maps\xspace}
\newcommand{\Wmratm}{Weak multiplicative rating map\xspace}
\newcommand{\Wmratms}{Weak multiplicative rating maps\xspace}

\newcommand{\gfil}[2]{\ensuremath{\tau^{#2}_{#1}}\xspace}
\newcommand{\gfiln}{\gfil{\Lb,\Ds}{n}}
\newcommand{\gfilp}{\gfil{\Lb}{n}\xspace}
\newcommand{\gfilpp}[1]{\gfil{\Lb}{#1}\xspace}

\newcommand{\fixg}[3]{\ensuremath{Red_{#1}^{#2}(#3)}\xspace}
\newcommand{\fixgln}[1]{\fixg{\alpha}{n}{#1}}

\newcommand{\sdrat}{\ensuremath{\zeta_{\alpha,\rho}}\xspace}
\newcommand{\drat}[1]{\ensuremath{\zeta^{#1}_{\alpha,\rho}}\xspace}
\newcommand{\dratp}[2]{\ensuremath{\zeta^{#1}_{#2}}\xspace}

\newcommand{\sgrat}{\ensuremath{\xi_{\rho}}\xspace}
\newcommand{\grat}[1]{\ensuremath{\xi^{#1}_{\rho}}\xspace}

\newcommand{\quasi}[1]{\ensuremath{\mu_{#1}}\xspace}
\newcommand{\quasir}{\quasi{\rho}}
\newcommand{\quasit}{\quasi{\tau}}
\newcommand{\quasig}{\quasi{\gamma}}

\newcommand{\brataux}[2]{\ensuremath{\xi_{#1}[#2]}\xspace}
\newcommand{\bratauxc}{\brataux{\Cs}{\rho}}
\newcommand{\bratauxd}{\brataux{\Ds}{\rho}}
\newcommand{\bratauxbd}{\brataux{\bool{\Ds}}{\rho}}
\newcommand{\bratauxbc}{\brataux{\bpol{\Cs}}{\rho}}
\newcommand{\bratauxsfc}{\brataux{\sfp{\Cs}}{\rho}}

\newcommand{\bratauxbg}{\brataux{\bpol{\Gs}}{\rho}}
\newcommand{\bratauxsfg}{\brataux{\sfp{\Gs}}{\rho}}

\newcommand{\bratauxbw}{\brataux{\bpol{\Gs^+}}{\rho}}

\newcommand{\lrataux}[3]{\ensuremath{\zeta_{#1}[#2,#3]}\xspace}
\newcommand{\lratauxc}{\lrataux{\Cs}{\alpha}{\rho}}
\newcommand{\lratauxd}{\lrataux{\Ds}{\alpha}{\rho}}
\newcommand{\lratauxpc}{\lrataux{\pol{\Cs}}{\alpha}{\rho}}
\newcommand{\lratauxppc}{\lrataux{\pbpol{\Cs}}{\alpha}{\rho}}
\newcommand{\lratauxpnc}{\lrataux{\pbpol{\Cs}}{\nu}{\rho}}
\newcommand{\lratauxpd}{\lrataux{\pol{\Ds}}{\nu}{\rho}}

\newcommand{\lratrp}[2]{\lrataux{#1}{#2}{\rho}}
\newcommand{\lratdrp}[1]{\lrataux{\Ds}{#1}{\rho}}
\newcommand{\lratdrn}{\lratdrp{n}}

\newcommand{\lratauxpng}{\lrataux{\pbpol{\Gs}}{\nu}{\rho}}

\newcommand{\lratauxpg}{\lrataux{\pol{\Gs}}{\alpha}{\rho}}
\newcommand{\lratauxppg}{\lrataux{\pbpol{\Gs}}{\alpha}{\rho}}

\newcommand{\lratauxppw}{\lrataux{\pbpol{\Gs^+}}{\alpha}{\rho}}


\newcommand{\prin}[2]{\ensuremath{\Is[#1](#2)}\xspace}

\newcommand{\itriv}[1]{\ensuremath{\Is_{\mathit{triv}}[#1]}\xspace}
\newcommand{\ptriv}[1]{\ensuremath{\Ps_{\mathit{triv}}[#1]}\xspace}

\newcommand{\opti}[2]{\ensuremath{\Is_{#1}\left[#2\right]}\xspace}
\newcommand{\copti}[1]{\opti{\Cs}{#1}}
\newcommand{\dopti}[1]{\opti{\Ds}{#1}}
\newcommand{\gopti}[1]{\opti{\Gs}{#1}}
\newcommand{\popti}[3]{\ensuremath{\Ps_{#1}[#2,#3]}\xspace}
\newcommand{\pcopti}[2]{\popti{\Cs}{#1}{#2}}
\newcommand{\pdopti}[2]{\popti{\Ds}{#1}{#2}}

\newcommand{\iopti}[2]{\ensuremath{\fri_{#1}[#2]}\xspace}
\newcommand{\ioptic}[1]{\iopti{\Cs}{#1}}
\newcommand{\ioptid}[1]{\iopti{\Ds}{#1}}
\newcommand{\ioptig}[1]{\iopti{\Gs}{#1}}

\newcommand{\nsep}[2]{\ensuremath{\Is_{#1}[#2]}\xspace}
\newcommand{\nsepc}[1]{\nsep{\Cs}{#1}}
\newcommand{\nsepcm}{\nsepc{\alpha}}
\newcommand{\nsepca}{\nsepc{\As}}
\newcommand{\nsepd}[1]{\nsep{\Ds}{#1}}
\newcommand{\nsepdm}{\nsepd{\alpha}}
\newcommand{\nsepda}{\nsepd{\As}}

\newcommand{\nsepbca}{\nsep{\bpol{\Cs}}{\As}}
\newcommand{\gnsepbca}{\nsep{\bpol{\Cs}}{\Gs,\As}}


\newcommand{\rbpol}[1]{\ensuremath{\Rs^\rho_{#1}}\xspace}
\newcommand{\rbpols}{\ensuremath{\rbpol{S}}\xspace}
\newcommand{\rbpolt}{\ensuremath{\rbpol{T}}\xspace}

\newcommand{\tbpol}[1]{\ensuremath{\Rs_\tau(#1)}\xspace}
\newcommand{\tbpols}{\ensuremath{\tbpol{S}}\xspace}

\newcommand{\redstone}[1]{\ensuremath{\Ss\Ts^{1,\rho}_{#1}}\xspace}
\newcommand{\redstones}{\redstone{S}}

\newcommand{\reddotone}[1]{\ensuremath{\Ds\Ds^{1,\rho}_{#1}}\xspace}
\newcommand{\reddotones}{\reddotone{S}}

\newcommand{\rstd}[1]{\ensuremath{\Ss\Ts^{2,\rho}_{#1}}\xspace}
\newcommand{\rstds}{\rstd{S}}


\newcommand{\ttup}[3]{\ensuremath{\Ts_{#1}^{#2}[{#3}]}\xspace}
\newcommand{\polttup}[1]{\ttup{\pol{\Cs}}{#1}{\Lb}}
\newcommand{\polttupn}{\polttup{n}}

\newcommand{\stup}[1]{\ensuremath{\Ts^{#1}[\Lb]}\xspace}
\newcommand{\stupn}{\stup{n}}
\newcommand{\stupnm}{\stup{n-1}}

\newcommand{\alte}[3]{\ensuremath{\As_{#1}^{#2}[#3]}\xspace}
\newcommand{\polalte}[1]{\alte{\pol{\Cs}}{#1}{\Lb}}

\newcommand{\salte}{\ensuremath{\As[\Lb]}\xspace}

\newcommand{\ite}[2]{\ensuremath{it_{#1}(#2)}\xspace}
\newcommand{\iteh}[1]{\ite{h}{#1}}

\newcommand{\extra}[2]{\ensuremath{ex_{#1}(#2)}\xspace}

\newcommand{\gfix}[1]{\ensuremath{Res(#1)}\xspace}

\newcommand{\igrpopti}{\iopti{\grp}{\rho}}
\newcommand{\grpopti}{\opti{\grp}{\rho}}
\newcommand{\mdopti}{\opti{\md}{\rho}}

\newcommand{\atopti}{\opti{\at}{\rho}}
\newcommand{\attopti}{\opti{\att}{\rho}}
\newcommand{\ltopti}{\opti{\lt}{\rho}}
\newcommand{\lttopti}{\opti{\ltt}{\rho}}


\newcommand{\bsuopti}{\opti{\bscu}{\rho}}
\newcommand{\ptopti}{\opti{\pt}{\rho}}
\newcommand{\pptopti}{\popti{\ppt}{\alpha}{\rho}}


\newcommand{\fodopti}{\opti{\ul}{\rho}}
\newcommand{\fodpopti}{\popti{\ul}{\etaat}{\rho}}


\newcommand{\foopti}{\opti{\fo}{\rho}}
\newcommand{\sfopti}{\opti{\sfr}{\rho}}

\newcommand{\sfsatag}{\sata{G}{\rho}}
\newcommand{\sfsatah}{\sata{H}{\rho}}


\newcommand{\powgopti}{\popti{\pol{\Gs^+}}{\alpha}{\rho}}
\newcommand{\bpwgopti}{\opti{\bpol{\Gs^+}}{\rho}}
\newcommand{\pbpwgopti}{\popti{\pbpol{\Gs^+}}{\alpha}{\rho}}

\newcommand{\pogopti}{\popti{\pol{\Gs}}{\alpha}{\rho}}
\newcommand{\bpgopti}{\opti{\bpol{\Gs}}{\rho}}
\newcommand{\pbpgopti}{\popti{\pbpol{\Gs}}{\alpha}{\rho}}

\newcommand{\pocopti}{\popti{\pol{\Cs}}{\alpha}{\rho}}
\newcommand{\cpocopti}{\popti{\pol{\Cs}}{\etac,\alpha}{\rho}}
\newcommand{\bpopti}{\opti{\bpol{\Cs}}{\rho}}
\newcommand{\cbpopti}{\popti{\bpol{\Cs}}{\etac}{\rho}}
\newcommand{\pbpopti}{\popti{\pbpol{\Cs}}{\alpha}{\rho}}
\newcommand{\cpbpopti}{\popti{\pbpol{\Cs}}{\etac,\alpha}{\rho}}
\newcommand{\sfcopti}{\opti{\sfp{\Cs}}{\rho}}
\newcommand{\csfcopti}{\popti{\sfp{\Cs}}{\etac}{\rho}}

\newcommand{\stau}{\ensuremath{\overline{\tau}}\xspace}
\newcommand{\sgamma}{\ensuremath{\overline{\gamma}}\xspace}
\newcommand{\gpocopti}{\popti{\pol{\Cs}}{\etac,\rho_*}{\gamma}}

\newcommand{\sthoneopti}{\popti{\sthone}{\alpha}{\rho}}
\newcommand{\sthtwoopti}{\popti{\sthtwo}{\alpha}{\rho}}
\newcommand{\asthtwoopti}{\popti{\sthtwo}{\etaat,\alpha}{\rho}}
\newcommand{\sththropti}{\popti{\sththree}{\alpha}{\rho}}
\newcommand{\asththropti}{\popti{\sththree}{\etaat,\alpha}{\rho}}

\newcommand{\stoneopti}{\opti{\stone}{\rho}}
\newcommand{\dotoneopti}{\opti{\dotone}{\rho}}
\newcommand{\sttwoopti}{\opti{\sttwo}{\rho}}
\newcommand{\asttwoopti}{\popti{\sttwo}{\etaat}{\rho}}
\newcommand{\dothoneopti}{\popti{\dothone}{\alpha}{\rho}}
\newcommand{\dothtwoopti}{\popti{\dothtwo}{\alpha}{\rho}}

\newcommand{\upolopti}{\opti{\upol{\Cs}}{\rho}}
\newcommand{\ldetopti}{\opti{\ldet{\Cs}}{\rho}}
\newcommand{\rdetopti}{\opti{\rdet{\Cs}}{\rho}}
\newcommand{\ldetnopti}{\opti{\ldetn{\Cs}}{\rho}}
\newcommand{\rdetnopti}{\opti{\rdetn{\Cs}}{\rho}}

\newcommand{\pupolopti}{\popti{\upol{\Cs}}{\etac}{\rho}}
\newcommand{\pldetopti}{\popti{\ldet{\Cs}}{\etac}{\rho}}
\newcommand{\pldetoptid}{\popti{\ldet{\Ds}}{\etac}{\rho}}
\newcommand{\prdetopti}{\popti{\rdet{\Cs}}{\etac}{\rho}}
\newcommand{\prdetoptid}{\popti{\rdet{\Ds}}{\etac}{\rho}}
\newcommand{\pmdetopti}{\popti{\mdet{\Cs}}{\etac}{\rho}}
\newcommand{\pmdetoptid}{\popti{\mdet{\Ds}}{\etac}{\rho}}
\newcommand{\pldetnopti}{\popti{\ldetn{\Cs}}{\etac}{\rho}}
\newcommand{\pldetnmopti}{\popti{\ldetp{n-1}{\Cs}}{\etac}{\rho}}
\newcommand{\prdetnopti}{\popti{\rdetn{\Cs}}{\etac}{\rho}}
\newcommand{\prdetnmopti}{\popti{\rdetp{n-1}{\Cs}}{\etac}{\rho}}
\newcommand{\pmdetnopti}{\popti{\mdetn{\Cs}}{\etac}{\rho}}
\newcommand{\pmdetnmopti}{\popti{\mdetp{n-1}{\Cs}}{\etac}{\rho}}


\newcommand{\typ}[2]{\ensuremath{[#1]_{#2}}\xspace}
\newcommand{\ctype}[1]{\typ{#1}{\Cs}}
\newcommand{\pctype}[1]{\typ{#1}{\pol{\Cs}}}
\newcommand{\cmult}{\ensuremath{\mathbin{\scriptscriptstyle\bullet}}}
\newcommand{\dtype}[1]{\typ{#1}{\Ds}}
\newcommand{\gtype}[1]{\typ{#1}{\Gs}}

\newcommand{\rtyp}[2]{\ensuremath{\llbracket #1 \rrbracket_{#2}}\xspace}
\newcommand{\crtype}[1]{\rtyp{#1}{\Cs}}
\newcommand{\ertype}[1]{\rtyp{#1}{\varepsilon}}

\newcommand{\sztyp}[1]{\typ{#1}{\stzer}}
\newcommand{\szrtyp}[1]{\rtyp{#1}{\stzer}}

\newcommand{\dztyp}[1]{\typ{#1}{\dotzer}}
\newcommand{\dzrtyp}[1]{\rtyp{#1}{\dotzer}}

\newcommand{\compset}[2]{\ensuremath{\fil{\alpha}{\typ{#2}{#1}}}\xspace}
\newcommand{\compsetc}[1]{\compset{\Cs}{#1}}


\newcommand{\alrat}{\ensuremath{\nu^{\alpha}_{\rho}}\xspace}
\newcommand{\abrat}{\ensuremath{\nu_{\rho}}\xspace}

\newcommand{\glrat}[1]{\ensuremath{\nu^{\alpha}_{\rho,#1}}\xspace}
\newcommand{\glrats}{\glrat{S}}

\newcommand{\bwmrat}[1]{\ensuremath{\eta_{\rho,#1}}\xspace}
\newcommand{\bwmrats}{\bwmrat{S}}

\newcommand{\bpwmrat}[1]{\ensuremath{\lambda_{\rho,#1}}\xspace}
\newcommand{\bpwmrats}{\bpwmrat{S}}
\newcommand{\idpr}{\ensuremath{E_+(\rho)}\xspace}

\newcommand{\pbwmrat}[1]{\ensuremath{\delta^{\alpha}_{\rho,#1}}\xspace}
\newcommand{\pbwmrats}{\pbwmrat{S}}

\newcommand{\pwbwmrat}[1]{\ensuremath{\mu^{\alpha}_{\rho,#1}}\xspace}
\newcommand{\pwbwmrats}{\pwbwmrat{S}}
\newcommand{\idpar}{\ensuremath{E_+(\alpha,\rho)}\xspace}


\tikzset{every state/.style={draw=blue!50!green,very thick,fill=blue!50!green!20}}
\tikzset{statesub/.style={state,minimum size=1.3cm,inner sep=1pt}}
\tikzset{pattstate/.style={state,draw=red!50!yellow,line width=2pt,fill=red!50!yellow!20}}
\tikzset{pdotstate/.style={state,minimum size=0.75cm,inner sep=0.5pt,draw=red!50!yellow,line
		width=2pt,dashed,fill=red!50!yellow!20}}
\tikzset{lstate/.style={state,minimum size=0.65cm,inner sep=0.5pt}}
\tikzstyle{trans}=[shorten >= 1pt,thick,->]

\makeatletter
\tikzstyle{initial by arrow}=   [after node path=
{
	{
		[to path=
		{
			[->,double=none,shorten >= 1pt,thick,every initial by arrow]
			([shift=(\tikz@initial@angle:\tikz@initial@distance)]\tikztostart.\tikz@initial@angle)
			node [shape=coordinate,anchor=\tikz@initial@anchor,draw] {\tikz@initial@text}
			-- (\tikztostart)}]
		edge ()
	}
}]
\makeatother

\makeatletter
\tikzstyle{accepting by arrow}=   [after node path=
{
	{
		[to path=
		{
			[->,double=none,shorten >= 1pt,thick,every accepting by arrow]
			(\tikztostart) -- 
			([shift=(\tikz@accepting@angle:\tikz@accepting@distance)]\tikztostart.\tikz@accepting@angle)
			node [shape=coordinate,anchor=\tikz@accepting@anchor,draw] {\tikz@accepting@text}
		}
		]
		edge ()
	}
}]
\makeatother

\newcommand{\mvline}[3][]{%
	\pgfmathtruncatemacro\hc{#3-1}
	\draw[#1]({$(#2-1-#3.west)!.5!(#2-1-\hc.east)$} |- #2.north) -- ({$(#2-1-#3.west)!.5!(#2-1-\hc.east)$} |- #2.south);
}
\newcommand{\mhline}[3][]{%
	\pgfmathtruncatemacro\hc{#3-1}
	\node[fit=(#2-#3-1),inner sep=0pt](R){};
	\node[fit=(#2-\hc-1),inner sep=0pt](L){};
	\node (K) at ($(R)!0.5!(L)$) {};
	\draw[#1] (K -| #2.west) -- (K -| #2.east);
}

\newcommand{\mhlinel}[3][]{%
	\pgfmathtruncatemacro\hc{#3-1}
	\node[fit=(#2-#3-2),inner sep=0pt](R){};
	\node[fit=(#2-\hc-2),inner sep=0pt](L){};
	\node (K) at ($(R)!0.5!(L)$) {};
	\draw[#1] (K -| #2.west) -- (K -| #2.east);
}


\newcommand{\pstr}[2]{\ensuremath{#1\textup{-}{#2}}\xspace}
\newcommand{\pstc}[1]{\pstr{#1}{\Cs}}
\newcommand{\kstc}{\pstc{k}}


\newcommand{\hiep}[2]{\ensuremath{#1_{#2}}\xspace}
\newcommand{\hiec}[1]{\hiep{\Cs}{#1}}
\newcommand{\hied}[1]{\hiep{\Ds}{#1}}
\newcommand{\hieg}[1]{\hiep{\Gs}{#1}}

\newcommand{\hiesp}[2]{\ensuremath{(#1^+)_{#2}}\xspace}
\newcommand{\hiesc}[1]{\hiesp{\Cs}{#1}}
\newcommand{\hiesg}[1]{\hiesp{\Gs}{#1}}

\renewcommand{\min}{\ensuremath{\text{\scriptsize min}}\xspace}
\renewcommand{\max}{\ensuremath{\text{\scriptsize max}}\xspace}

\newcommand{\mods}{\ensuremath{\mathit{MOD}}\xspace}
\newcommand{\amods}{\ensuremath{\mathit{AMOD}}\xspace}

\newcommand{\amodp}[3]{\ensuremath{\mathit{MOD}^{#1}_{#2,#3}}\xspace}
\newcommand{\amodamd}{\amodp{a}{m}{d}}

\newcommand{\modp}[2]{\ensuremath{\mathit{MOD}_{#1,#2}}\xspace}
\newcommand{\modmd}{\modp{m}{d}}

\newcommand{\sigenr}{\ensuremath{<,+1,\min,\max,\varepsilon}\xspace}
\newcommand{\addsucc}{\ensuremath{+1,\min,\max,\varepsilon}\xspace}


\newcommand{\fo}{\ensuremath{\textup{FO}}\xspace}
\newcommand{\fow}{\mbox{\ensuremath{\fo(<)}}\xspace}
\newcommand{\fows}{\mbox{\ensuremath{\fo(<,+1)}}\xspace}
\newcommand{\fowss}{\mbox{\ensuremath{\fo(\sigenr)}}\xspace}
\newcommand{\foeqp}{\mbox{\ensuremath{\fo(+1)}}\xspace}
\newcommand{\foeq}{\mbox{\ensuremath{\fo(\emptyset)}}\xspace}
\newcommand{\fop}{\mbox{\ensuremath{\fo(+1)}}\xspace}

\newcommand{\fowam}{\mbox{\ensuremath{\fo(<,\amods)}}\xspace}
\newcommand{\fowsam}{\mbox{\ensuremath{\fo(<,+1,\amods)}}\xspace}

\newcommand{\fowm}{\mbox{\ensuremath{\fo(<,\mods)}}\xspace}
\newcommand{\fowsm}{\mbox{\ensuremath{\fo(<,+1,\mods)}}\xspace}
\newcommand{\foeqpm}{\mbox{\ensuremath{\fo(=,+1,\mods)}}\xspace}
\newcommand{\fopm}{\mbox{\ensuremath{\fo(+1,\mods)}}\xspace}
\newcommand{\foeqm}{\mbox{\ensuremath{\fo(=,\mods)}}\xspace}

\newcommand{\rec}{\ensuremath{\textup{REC}}\xspace}
\newcommand{\reg}{\ensuremath{\textup{REG}}\xspace}
\newcommand{\mso}{\ensuremath{\textup{MSO}}\xspace}
\newcommand{\msow}{\ensuremath{\mso(<)}\xspace}
\newcommand{\msows}{\ensuremath{\mso(<,+1,\mods)}\xspace}
\newcommand{\msowsa}{\ensuremath{\mso(<,+1,\amods)}\xspace}
\newcommand{\msop}{\ensuremath{\mso(+1)}\xspace}
\newcommand{\msono}{\mso without first-order variables\xspace}

\newcommand{\fou}{\ensuremath{\fo^1}\xspace}
\newcommand{\fod}{\ensuremath{\fo^2}\xspace}
\newcommand{\fot}{\ensuremath{\fo^3}\xspace}
\newcommand{\foi}{\ensuremath{\fo^i}\xspace}

\newcommand{\folab}{\mbox{\ensuremath{\fou(\emptyset)}}\xspace}
\newcommand{\folabm}{\mbox{\ensuremath{\fou(\mods)}}\xspace}
\newcommand{\folabam}{\mbox{\ensuremath{\fou(\amods)}}\xspace}
\newcommand{\fouw}{\ensuremath{\fou(<)}\xspace}
\newcommand{\foueq}{\ensuremath{\fou(=)}\xspace}
\newcommand{\foup}{\ensuremath{\fou(+1)}\xspace}

\newcommand{\fodw}{\ensuremath{\fod(<)}\xspace}
\newcommand{\fodws}{\ensuremath{\fod(<,+1)}\xspace}
\newcommand{\fodeq}{\ensuremath{\fod(=)}\xspace}
\newcommand{\fodp}{\ensuremath{\fod(+1)}\xspace}
\newcommand{\fodwm}{\ensuremath{\fod(<,\mods)}\xspace}
\newcommand{\fodwam}{\ensuremath{\fod(<,\amods)}\xspace}
\newcommand{\fodwsm}{\ensuremath{\fod(<,+1,\mods)}\xspace}
\newcommand{\fodwsam}{\ensuremath{\fod(<,+1,\amods)}\xspace}
\newcommand{\fodreg}{\ensuremath{\fod(REG)}\xspace}

\newcommand{\fodsi}[1]{\ensuremath{\sic{#1}^2}\xspace}
\newcommand{\fodsn}{\fodsi{n}}

\newcommand{\fodpi}[1]{\ensuremath{\pic{#1}^2}\xspace}
\newcommand{\fodpn}{\fodpi{n}}

\newcommand{\fodb}[1]{\ensuremath{\bsc{#1}^2}\xspace}
\newcommand{\fodbn}{\fodb{n}}

\newcommand{\foieq}{\ensuremath{\foi(=)}\xspace}
\newcommand{\foip}{\ensuremath{\foi(+1)}\xspace}

\newcommand{\fotw}{\ensuremath{\fot(<)}\xspace}
\newcommand{\foiw}{\ensuremath{\foi(<)}\xspace}

\newcommand{\fodiw}{\ensuremath{\fo^2_i(<)}\xspace}
\newcommand{\foduw}{\ensuremath{\fo^2_1(<)}\xspace}
\newcommand{\foddw}{\ensuremath{\fo^2_2(<)}\xspace}
\newcommand{\fodtw}{\ensuremath{\fo^2_3(<)}\xspace}

\newcommand{\fodiws}{\ensuremath{\fo^2_i(<,+1)}\xspace}
\newcommand{\foduws}{\ensuremath{\fo^2_1(<,+1,\min,\max)}\xspace}
\newcommand{\foddws}{\ensuremath{\fo^2_2(<,+1)}\xspace}
\newcommand{\fodtws}{\ensuremath{\fo^2_3(<,+1)}\xspace}


\newcommand{\sic}[1]{\ensuremath{\Sigma_{#1}}\xspace}
\newcommand{\siw}[1]{\ensuremath{\Sigma_{#1}(<)}\xspace}
\newcommand{\siws}[1]{\ensuremath{\Sigma_{#1}(<,+1)}\xspace}
\newcommand{\siwm}[1]{\ensuremath{\Sigma_{#1}(<,\mods)}\xspace}
\newcommand{\siwam}[1]{\ensuremath{\Sigma_{#1}(<,\amods)}\xspace}
\newcommand{\siwsm}[1]{\ensuremath{\Sigma_{#1}(<,+1,\mods)}\xspace}
\newcommand{\siwsam}[1]{\ensuremath{\Sigma_{#1}(<,+1,\amods)}\xspace}
\newcommand{\siwss}[1]{\ensuremath{\Sigma_{#1}(<,+1,\min,\max,\varepsilon)}\xspace}
\newcommand{\siwssm}[1]{\ensuremath{\Sigma_{#1}(<,+1,\min,\max,\varepsilon,\mods)}\xspace}

\newcommand{\sici}{\sic{i}}
\newcommand{\siwi}{\siw{i}}
\newcommand{\siwsi}{\siws{i}}
\newcommand{\siwu}{\siw{1}}
\newcommand{\sicu}{\sic{1}}
\newcommand{\siwsu}{\siws{1}}
\newcommand{\sicd}{\sic{2}}
\newcommand{\siwd}{\siw{2}}
\newcommand{\siwsd}{\siws{2}}
\newcommand{\sict}{\sic{3}}
\newcommand{\siwt}{\siw{3}}
\newcommand{\siwst}{\siws{3}}
\newcommand{\sicq}{\sic{4}}
\newcommand{\siwq}{\siw{4}}
\newcommand{\siwsq}{\siws{4}}

\newcommand{\pic}[1]{\ensuremath{\Pi_{#1}}\xspace}
\newcommand{\piw}[1]{\ensuremath{\Pi_{#1}(<)}\xspace}
\newcommand{\piws}[1]{\ensuremath{\Pi_{#1}(<,+1)}\xspace}
\newcommand{\piwm}[1]{\ensuremath{\Pi_{#1}(<,\mods)}\xspace}
\newcommand{\piwam}[1]{\ensuremath{\Pi_{#1}(<,\amods)}\xspace}
\newcommand{\piwsm}[1]{\ensuremath{\Pi_{#1}(<,+1,\mods)}\xspace}
\newcommand{\piwsam}[1]{\ensuremath{\Pi_{#1}(<,+1,\amods)}\xspace}
\newcommand{\piwss}[1]{\ensuremath{\Pi_{#1}(<,+1,\min,\max,\varepsilon)}\xspace}

\newcommand{\pici}{\pic{i}}
\newcommand{\piwi}{\piw{i}}
\newcommand{\piwsi}{\piws{i}}
\newcommand{\piwu}{\piw{1}}
\newcommand{\picu}{\pic{1}}
\newcommand{\piwsu}{\piws{1}}
\newcommand{\picd}{\pic{2}}
\newcommand{\piwd}{\piw{2}}
\newcommand{\piwsd}{\piws{2}}
\newcommand{\pict}{\pic{3}}
\newcommand{\piwt}{\piw{3}}
\newcommand{\piwst}{\piws{3}}
\newcommand{\picq}{\pic{4}}
\newcommand{\piwq}{\piw{4}}
\newcommand{\piwsq}{\piws{4}}

\newcommand{\dec}[1]{\ensuremath{\Delta_{#1}}\xspace}
\newcommand{\dew}[1]{\ensuremath{\Delta_{#1}(<)}\xspace}
\newcommand{\dews}[1]{\ensuremath{\Delta_{#1}(<,+1)}\xspace}
\newcommand{\dewm}[1]{\ensuremath{\Delta_{#1}(<,\mods)}\xspace}
\newcommand{\dewam}[1]{\ensuremath{\Delta_{#1}(<,\amods)}\xspace}
\newcommand{\dewsm}[1]{\ensuremath{\Delta_{#1}(<,+1,\mods)}\xspace}
\newcommand{\dewsam}[1]{\ensuremath{\Delta_{#1}(<,+1,\amods)}\xspace}
\newcommand{\dewss}[1]{\ensuremath{\Delta_{#1}(<,+1,\min,\max,\varepsilon)}\xspace}

\newcommand{\deci}{\dec{i}}
\newcommand{\dewi}{\dew{i}}
\newcommand{\dewsi}{\dews{i}}
\newcommand{\dewu}{\dew{1}}
\newcommand{\decu}{\dec{1}}
\newcommand{\dewsu}{\dews{1}}
\newcommand{\decd}{\dec{2}}
\newcommand{\dewd}{\dew{2}}
\newcommand{\dewsd}{\dews{2}}
\newcommand{\dewmd}{\dewm{2}}
\newcommand{\dewamd}{\dewam{2}}
\newcommand{\dewsmd}{\dewsm{2}}
\newcommand{\dewsamd}{\dewsam{2}}
\newcommand{\dect}{\dec{3}}
\newcommand{\dewt}{\dew{3}}
\newcommand{\dewst}{\dews{3}}
\newcommand{\decq}{\dec{4}}
\newcommand{\dewq}{\dew{4}}
\newcommand{\dewsq}{\dews{4}}

\newcommand{\bsc}[1]{\ensuremath{\Bs\Sigma_{#1}}\xspace}
\newcommand{\bsw}[1]{\ensuremath{\Bs\Sigma_{#1}(<)}\xspace}
\newcommand{\bswm}[1]{\ensuremath{\Bs\Sigma_{#1}(<,\mods)}\xspace}
\newcommand{\bswam}[1]{\ensuremath{\Bs\Sigma_{#1}(<,\amods)}\xspace}
\newcommand{\bsws}[1]{\ensuremath{\Bs\Sigma_{#1}(<,+1)}\xspace}
\newcommand{\bswsm}[1]{\ensuremath{\Bs\Sigma_{#1}(<,+1,\mods)}\xspace}
\newcommand{\bswsam}[1]{\ensuremath{\Bs\Sigma_{#1}(<,+1,\amods)}\xspace}
\newcommand{\bswss}[1]{\ensuremath{\Bs\Sigma_{#1}(<,+1,\min,\max,\varepsilon)}\xspace}

\newcommand{\bsci}{\bsc{i}}
\newcommand{\bswi}{\bsw{i}}
\newcommand{\bswsi}{\bsws{i}}
\newcommand{\bswu}{\bsw{1}}
\newcommand{\bscu}{\bsc{1}}
\newcommand{\bswsu}{\bsws{1}}
\newcommand{\bscd}{\bsc{2}}
\newcommand{\bswd}{\bsw{2}}
\newcommand{\bswsd}{\bsws{2}}
\newcommand{\bsct}{\bsc{3}}
\newcommand{\bswt}{\bsw{3}}
\newcommand{\bswst}{\bsws{3}}
\newcommand{\bscq}{\bsc{4}}
\newcommand{\bswq}{\bsw{4}}
\newcommand{\bswsq}{\bsws{4}}


\newcommand{\ppt}{\ensuremath{\textup{PPT}}\xspace}
\newcommand{\pptp}[1]{\ensuremath{#1\textup{-}\textup{PPT}}\xspace}
\newcommand{\kppt}{\pptp{k}}
\newcommand{\pt}{\ensuremath{\textup{PT}}\xspace}
\newcommand{\ptp}[1]{\ensuremath{#1\textup{-}\textup{PT}}\xspace}
\newcommand{\kpt}{\ptp{k}}
\newcommand{\ld}{\ensuremath{\textup{LD}}\xspace}
\newcommand{\rd}{\ensuremath{\textup{RD}}\xspace}

\newcommand{\wat}{\ensuremath{\textup{WAT}}\xspace}
\newcommand{\at}{\ensuremath{\textup{AT}}\xspace}

\newcommand{\cgrp}{\ensuremath{\textup{Ab}}\xspace}
\newcommand{\grp}{\ensuremath{\textup{GR}}\xspace}
\newcommand{\sgrp}{\ensuremath{\textup{SGR}}\xspace}

\newcommand{\abg}{\ensuremath{\textup{AMT}}\xspace}
\newcommand{\sabg}{\ensuremath{\textup{SAMT}}\xspace}
\newcommand{\pabg}[1]{\ensuremath{#1\textup{-}\textup{AMT}}\xspace}
\newcommand{\dabg}{\pabg{d}}

\newcommand{\md}{\ensuremath{\textup{MOD}}\xspace}
\newcommand{\smd}{\ensuremath{\textup{SMOD}}\xspace}
\newcommand{\pmd}[1]{\ensuremath{#1\textup{-}\textup{MOD}}\xspace}
\newcommand{\dmd}{\pmd{d}}

\newcommand{\pr}{\ensuremath{\textup{PR}}\xspace}
\newcommand{\wpr}{\ensuremath{\textup{PPR}}\xspace}
\newcommand{\prp}[1]{\ensuremath{#1\textup{-}\textup{PR}}\xspace}
\newcommand{\kpr}{\prp{k}}

\newcommand{\wsu}{\ensuremath{\textup{PSU}}\xspace}
\newcommand{\su}{\ensuremath{\textup{SU}}\xspace}
\newcommand{\psu}[1]{\ensuremath{#1\textup{-}\textup{SU}}\xspace}
\newcommand{\ksu}{\psu{k}}

\newcommand{\sfr}{\ensuremath{\textup{SF}}\xspace}
\newcommand{\sfrp}[1]{\ensuremath{#1\textup{-}\textup{SF}}\xspace}
\newcommand{\ksfr}{\sfrp{k}}
\newcommand{\bsd}{\ensuremath{\textup{SD}}\xspace}

\newcommand{\patt}[1]{\ensuremath{\ensuremath{#1\textup{-ATT}}}\xspace}
\newcommand{\datt}{\patt{d}}
\newcommand{\att}{\ensuremath{\textup{ATT}}\xspace}

\newcommand{\plt}[1]{\ensuremath{#1\textup{-LT}}\xspace}
\newcommand{\pltt}[2]{\ensuremath{(#1,#2)\textup{-LTT}}\xspace}
\newcommand{\klt}{\plt{k}}
\newcommand{\kdltt}{\pltt{k}{d}}
\newcommand{\lt}{\ensuremath{\textup{LT}}\xspace}
\newcommand{\ltt}{\ensuremath{\textup{LTT}}\xspace}
\newcommand{\lmt}{\ensuremath{\textup{LMT}}\xspace}

\newcommand{\stzer}{\textup{ST}\xspace}
\newcommand{\dotzer}{\textup{DD}\xspace}

\newcommand{\sttp}[1]{\hiep{\stzer}{#1}}
\newcommand{\dotdp}[1]{\hiep{\dotzer}{#1}}
\newcommand{\mdp}[1]{\hiep{\md}{#1}}
\newcommand{\abgp}[1]{\hiep{\abg}{#1}}
\newcommand{\grpp}[1]{\hiep{\grp}{#1}}
\newcommand{\mdwp}[1]{\hiep{\smd}{#1}}
\newcommand{\abgwp}[1]{\hiep{\sabg}{#1}}
\newcommand{\grpwp}[1]{\hiep{\sgrp}{#1}}

\newcommand{\sthone}{\sttp{{1}/{2}}}
\newcommand{\stone}{\sttp{1}}
\newcommand{\sthtwo}{\sttp{{3}/{2}}}
\newcommand{\sttwo}{\sttp{2}}
\newcommand{\sththree}{\sttp{{5}/{2}}}
\newcommand{\stthree}{\sttp{3}}
\newcommand{\sthfour}{\sttp{{7}/{2}}}
\newcommand{\stfour}{\sttp{4}}

\newcommand{\dothone}{\dotdp{{1}/{2}}}
\newcommand{\dotone}{\dotdp{1}}
\newcommand{\dothtwo}{\dotdp{{3}/{2}}}
\newcommand{\dottwo}{\dotdp{2}}
\newcommand{\doththree}{\dotdp{{5}/{2}}}
\newcommand{\dotthree}{\dotdp{3}}
\newcommand{\dothfour}{\dotdp{{7}/{2}}}
\newcommand{\dotfour}{\dotdp{4}}


\newcommand{\cocl}[1]{\ensuremath{\mathit{co\textup{-}}\!#1}\xspace}
\newcommand{\bool}[1]{\ensuremath{\mathit{Bool}(#1)}\xspace}
\newcommand{\pol}[1]{\ensuremath{\mathit{Pol}(#1)}\xspace}
\newcommand{\bpol}[1]{\ensuremath{\mathit{BPol}(#1)}\xspace}
\newcommand{\pbpol}[1]{\ensuremath{\mathit{PBPol}(#1)}\xspace}
\newcommand{\upol}[1]{\ensuremath{\mathit{UPol}(#1)}\xspace}
\newcommand{\spol}[1]{\ensuremath{\mathit{SPol}(#1)}\xspace}
\newcommand{\copol}[1]{\ensuremath{\cocl{Pol(#1)}}\xspace}
\newcommand{\capol}[1]{\ensuremath{\pol{#1} \cap \copol{#1}}\xspace}
\newcommand{\sfp}[1]{\ensuremath{\mathit{\textup{SF}}(#1)}\xspace}
\newcommand{\bsdp}[1]{\ensuremath{\mathit{\textup{SD}}(#1)}\xspace}
\newcommand{\dsdp}[1]{\ensuremath{\mathit{\textup{DSD}}(#1)}\xspace}
\newcommand{\upolo}{\ensuremath{\mathit{UPol}}\xspace}
\newcommand{\bpolo}{\ensuremath{\mathit{BPol}}\xspace}

\newcommand{\booln}{\ensuremath{\mathit{Bool}}\xspace}
\newcommand{\poln}{\ensuremath{\mathit{Pol}}\xspace}
\newcommand{\bpoln}{\ensuremath{\mathit{BPol}}\xspace}
\newcommand{\pbpoln}{\ensuremath{\mathit{PBPol}}\xspace}
\newcommand{\sfpn}{\ensuremath{\textup{SF}}\xspace}

\newcommand{\ldet}[1]{\ensuremath{\mathit{LPol}(#1)}\xspace}
\newcommand{\rdet}[1]{\ensuremath{\mathit{RPol}(#1)}\xspace}
\newcommand{\adet}[1]{\ensuremath{\mathit{APol}(#1)}\xspace}
\newcommand{\wadet}[1]{\ensuremath{\mathit{WAPol}(#1)}\xspace}

\newcommand{\idet}[1]{\ensuremath{\ldet{#1} \cap \rdet{#1}}\xspace}
\newcommand{\jdet}[1]{\ensuremath{\ldet{#1} \vee \rdet{#1}}\xspace}
\newcommand{\mdet}[1]{\ensuremath{\mathit{MPol}(#1)}\xspace}
\newcommand{\ldeto}{\ensuremath{\mathit{LPol}}\xspace}
\newcommand{\rdeto}{\ensuremath{\mathit{RPol}}\xspace}
\newcommand{\adeto}{\ensuremath{\mathit{APol}}\xspace}
\newcommand{\wadeto}{\ensuremath{\mathit{WAPol}}\xspace}
\newcommand{\ideto}{\ensuremath{\mathit{IPol}}\xspace}
\newcommand{\jdeto}{\ensuremath{\mathit{JPol}}\xspace}
\newcommand{\mdeto}{\ensuremath{\mathit{MPol}}\xspace}

\newcommand{\ldetp}[2]{\ensuremath{\ldeto_{#1}(#2)}\xspace}
\newcommand{\rdetp}[2]{\ensuremath{\rdeto_{#1}(#2)}\xspace}
\newcommand{\mdetp}[2]{\ensuremath{\mdeto_{#1}(#2)}\xspace}
\newcommand{\idetp}[2]{\ensuremath{\ldetp{#1}{#2} \cap \rdetp{#1}{#2}}\xspace}
\newcommand{\jdetp}[2]{\ensuremath{\ldetp{#1}{#2} \vee \rdetp{#1}{#2}}\xspace}
\newcommand{\ldetn}[1]{\ldetp{n}{#1}}
\newcommand{\rdetn}[1]{\rdetp{n}{#1}}
\newcommand{\mdetn}[1]{\mdetp{n}{#1}}
\newcommand{\idetn}[1]{\idetp{n}{#1}}
\newcommand{\jdetn}[1]{\jdetp{n}{#1}}

\newcommand{\polp}[2]{\ensuremath{Pol_{#2}(#1)}\xspace}
\newcommand{\polk}[1]{\polp{#1}{k}}

\newcommand{\bpolp}[2]{\ensuremath{BPol_{#2}(#1)}\xspace}
\newcommand{\bpolk}[1]{\bpolp{#1}{k}}


\newcommand{\davar}{\ensuremath{\mathbf{DA}}\xspace}
\newcommand{\ldavar}{\ensuremath{\mathbf{LDA}}\xspace}
\newcommand{\jvar}{\ensuremath{\mathbf{J}}\xspace}

\maketitle

\begin{abstract}
  We look at classes of languages associated to the fragment of first-order logic \bsc{1}, in which quantifier alternations are disallowed. Each class is defined by choosing the set of predicates on positions that may be used. Two key such fragments are those equipped with the linear ordering and possibly the successor relation. Simon and Knast proved that these two variants have decidable \emph{membership}: ``does an input regular language belong to the class ?''. We rely on a characterization of \bsc{1} by the operator \bpoln: given an input class \Cs, it outputs a class \bpol{\Cs} that corresponds to a variant of \bsc{1} equipped with special predicates associated to \Cs. We extend the above results in two orthogonal directions. First, we use two kinds of inputs: classes \Gs of \emph{group languages} (\emph{i.e.}, recognized by a DFA in which each letter induces a permutation of the states) and extensions thereof, written $\Gs^+$. The classes \bpol{\Gs} and \bpol{\Gs^+} capture many natural variants of \bsc{1} which use predicates such as the linear ordering, the successor, the modular predicates or the alphabetic modular predicates.
  
  Second, instead of membership, we explore the more general separation problem: decide if two regular languages can be separated by a language from the class under study. We show that separation is decidable for \bpol{\Gs} and \bpol{\Gs^+} when this is the case for \Gs. This was known for \bpol{\Gs} and for two particular classes of the form \bpol{\Gs^+}. Yet, the algorithms were indirect and relied on involved frameworks, yielding poor upper complexity bounds. In contrast, the approach of the paper is direct. We work only with elementary concepts (mainly, finite automata). Our main contribution consists in polynomial time Turing reductions from both \bpol{\Gs}- and \bpol{\Gs^+}-separation to \Gs-separation. This yields polynomial algorithms for many key variants of \bsc{1}, including those equipped with the linear ordering and possibly the successor and/or the modular predicates.
\end{abstract}

\section{Introduction}
\label{sec:intro}
An important question in automata theory is to precisely understand the prominent classes of regular languages of finite words. We are interested in the classes associated to a piece of syntax (such as regular expressions or logic), whose purpose is to specify the languages of such classes. In the paper, we formalize the goal of ``understanding a given class \Cs'' by looking at a decision problem: \Cs-separation. It takes two regular languages $L_1,L_2$ as input and asks whether there exists $K \in \Cs$ such that $L_1 \subseteq K$ and $K \cap L_2=\emptyset$. The key idea is that obtaining an algorithm for \Cs-separation requires a solid understanding of~\Cs.

We investigate a family of classes associated to a fragment of first-order logic written~\bsc{1}. The sentences of \bsc{1} are Boolean combinations of \emph{existential} formulas, \emph{i.e.}, whose prenex normal form has the shape $\exists x_1\exists x_2\cdots\exists x_k\varphi$, with~$\varphi$ quantifier-free. Several classes are associated to \bsc{1}, each determined by the predicates on positions that we allow. In the literature, standard examples of predicates include the linear order ``$<$''~\cite{simonthm}, the successor relation ``$+1$''~\cite{knast83} or modular predicates ``$\mathit{MOD}$''~\cite{ChaubardPS06}. Thus, a generic approach is desirable.

We tackle languages associated to \bsc1 through the operator $\Cs \mapsto \bpol{\Cs}$ defined on classes of languages. It is the composition of the polynomial closure $\Cs \mapsto \pol{\Cs}$ and the Boolean closure $\Cs\mapsto\bool\Cs$ operators: $\bpol\Cs=\bool{\pol\Cs}$. Recall that the polynomial closure of a class \Cs consists of all finite unions of languages of the form $L_0a_1L_1\cdots a_nL_n$, where $n\geq0$, each $a_i$ is a letter and each $L_i$ belongs to~\Cs. Indeed, many classes associated to~\bsc{1} are of the form \bpol{\Cs}~\cite{ThomEqu,PZ:generic18}. In this paper, we look at specific input classes \Cs.

The \emph{group languages} are those recognized by a finite group, or equivalently by a permutation automaton~\cite{permauto} (\emph{i.e.}, which is complete, deterministic \emph{and} co-deterministic). We consider input classes that are either a class \Gs consisting of group languages, or a \wsuit extension  thereof, $\Gs^+$ (roughly, $\Gs^+$ is the least Boolean algebra containing~\Gs and the singleton~$\{\veps\}$). It is known~\cite{PZ:generic18} that if \Gs is a class of group languages, then $\bpol{\Gs} =  \bsc{1}(<,\prefsig{\Gs})$ and $\bpol{\Gs^+} =  \bsc{1}(<,+1,\prefsig{\Gs})$. Here, \prefsig{\Gs} is a set of predicates associated to \Gs: each language~$L$ in \Gs gives rise to a predicate $P_L(x)$, which selects all positions $x$ in a word $w$ such that the prefix of $w$ up to position $x$ (excluded) belongs to $L$. This captures most of the natural examples. In particular, we get signatures including the aforementioned predicates, such as $\{<\}$, $\{<,+1\}$, $\{<,\mathit{MOD}\}$ and $\{<,+1,\mathit{MOD}\}$ (we provide  some more examples in the~paper).

\smallskip\noindent
{\bf State of the art.} Historically, \bpol{\Gs} and \bpol{\Gs^+} were first investigated for  particular input classes. A prominent example is the class of piecewise testable languages~\cite{simonthm}, \emph{i.e.}, the class $\bpol{\stzer} = \bsw{1}$ where $\stzer = \{\emptyset,A^*\}$. It was shown that \bpol{\stzer}-separation is decidable in~\cite{azptsep} using technical algebraic arguments. Simpler polynomial time algorithms were discovered later~\cite{pvzptsep,cmmptsep}. There also exists an involved specialized separation algorithm~\cite{Zetzsche18} for $\bpol{\md} = \bswm{1}$, where \md is the class of modulo languages. Decidability can be lifted to $\bpol{\stzer^+}=\bsws{1}$ (the languages of dot-depth one~\cite{knast83}) and to $\bpol{\md^+} = \bswsm{1}$ via transfer results~\cite{PlaceZ20,prwmodulo}. Unfortunately, this approach yields an exponential complexity blow-up. Recently, a generic approach was developed for \bpol{\Gs}. It is proved in~\cite{pzconcagroup} that if \Gs is a class of group languages with mild hypotheses, \bpol{\Gs}-separation is decidable when \Gs-separation is decidable. Yet, this generic approach is indirect and considers a more general problem: \emph{covering}. Because of this, the algorithms and their proofs are complex and rely on an intricate framework~\cite{pzcovering2}, yielding poor upper complexity bounds. This contrasts with the simple polynomial time procedures presented in~\cite{pvzptsep,cmmptsep} for \bpol{\stzer}. No generic result of this kind is known for the classes \bpol{\Gs^+}.

\smallskip
\noindent
{\bf Contributions.} We give \emph{generic polynomial time Turing reductions} from \bpol{\Gs}- and \bpol{\Gs^+}-separation  to \Gs-separation, where \Gs is a class of group languages with mild properties. We present them as greatest fixpoint procedures which use an oracle for \mbox{\Gs-separation} at each step and run in \emph{polynomial time} (for input languages represented by nondeterministic finite automata). While the proofs are involved, they are self-contained and based exclusively on elementary concepts from automata theory. No particular knowledge on group theory is required to follow them: we only use immediate consequences of the definition of a group.

For \bpol{\Gs}, this new approach is a significant improvement on the results of~\cite{pzconcagroup}. While we do reuse some ideas of~\cite{pzconcagroup}, we complement them with new ones and the presentation is independent. We get a simpler algorithm, which requires only basic notions from automata theory. In particular, one direction of the proof describes a generic construction for building separators in \bpol{\Gs} (when they exist). This serves our main objective: understanding classes of languages. In addition, we obtain much better complexity upper bounds on \bpol{\Gs}-separation. Finally, our techniques can handle \bpol{\Gs^+} as well. This was not the case in~\cite{pzconcagroup}: the generic reduction from \bpol{\Gs^+}-separation to \Gs-separation is a new result.

These results apply to several key classes. Separation is decidable in polynomial time for $\stzer =\{\emptyset, A^*\}$, for the class \md of modulo languages and for the class \grp of \emph{all} group languages~\cite{pzgroup}. Hence, the problem is also decidable in polynomial time for \bpol{\stzer} (\emph{i.e.}, \bsw{1}), \bpol{\stzer^+} (\emph{i.e.}, \bsws{1}),  \bpol{\md} (\emph{i.e.}, \bswm{1}), \bpol{\md^+} (\emph{i.e.}, \bswsm{1}), \bpol{\grp} and \bpol{\grp^+} (the logical characterization of the last two classes is not standard, yet they are quite prominent as well~\cite{margpin85,henckell:hal-00019815}). This reproves a known result for \bpol{\stzer} (in fact, we essentially reprove the algorithm of~\cite{cmmptsep}). The polynomial time upper bounds are new for all other classes. Another application is the class \abg of alphabet modulo testable languages (which are recognized by commutative groups): \bpol{\abg} and \bpol{\abg^+} correspond to \bswam{1} and \bswsam{1} where ``$\mathit{AMOD}$'' is the set of \emph{alphabetic modular predicates}. We obtain the decidability of separation for these classes (this is a new result for \bpol{\abg^+}). However, we do not get a polynomial time upper bound: this is because \abg-separation is \conptime-complete (see~\cite{pzgroup}).

\smallskip
\noindent
{\bf Important remark.}
  Eilenberg's theorem~\cite{EilenbergB} connects some classes of regular languages  (the ``varieties of languages'') with \emph{varieties of finite monoids}. It~raised the hope to solve decision problems on languages (such as membership) by translating them in terms of monoids and solving the resulting purely algebraic questions---without referring to languages anymore.
  In particular, Margolis and Pin~\cite{margpin85,pin:hal-00112635} characterized the algebraic counterpart of \bpol\Gs in Eilenberg's correspondence (when \Gs is a variety) as the ``\emph{semidirect product}'' $\mathsf{J} * \mathsf{G}$, where $\mathsf{J}$ is the variety of monoids corresponding to $\bsw1$ and $\mathsf{G}$ is the one corresponding to~\Gs. The new purely algebraic question is then: ``decide membership of a monoid in $\mathsf{J}*\mathsf{G}$''. Tilson~\cite{tilson} developed an involved framework to reformulate membership in semidirect products in terms of categories, which was successfully exploited to handle $(\mathsf{J}*\mathsf{G})$-membership~\cite{henckell:hal-00019815,Steinberg01}.
  
  Our results are completely independent from this algebraic approach. To clarify, we do use  combinatorics on monoids. Yet, our motivations and techniques are disconnected from the theory of varieties of monoids, which is a distinct field. We avoid it by choice: while the above approach highlights an interesting connection between two fields, it is not necessarily desirable when looking back at our primary goal, understanding \emph{classes of languages}. Indeed, a detour via varieties of monoids would obfuscate the intuition at the language level. Fortunately, this paper shows that this detour can be bypassed, while getting \emph{stronger} results. %
  First, our results are more general: they apply to \emph{separation}, and not only membership. It is not clear at all that this can be obtained in the context of monoid varieties, as we rely strongly on the definition of \bpolo: we work with languages of the form $L_0a_1L_1\cdots a_nL_n$, for~$L_i\in \Gs$.
  Second, we can handle  $\bpol{\Gs^+}$, thus capturing the successor relation on the logical side. As far as we know, the only class of this kind captured by the above framework is \bpol{\stzer^+} (these are the well-known dot-depth one languages~\cite{StrauVD}).
  Third, using the above approach requires \emph{varieties} of languages as input classes. This, for example, excludes the class $\bpol{\md}$. This does not mean that this class cannot be handled by algebraic techniques: this was actually done by  Straubing~\cite{straubing2002,pscvar}, who rebuilt the whole theory to be able to handle such classes. In contrast, our result applies \emph{uniformly} to \md.

\smallskip
\noindent
{\bf Organization of the paper.} We present the objects that we investigate and the required terminology in Section~\ref{sec:prelims}. We introduce separation and the techniques that we use to handle it in Section~\ref{sec:separ}. Finally, we present our results for \bpol{\Gs}- and \bpol{\Gs^+}-separation in Section~\ref{sec:bpolg}. Due to space limitations, some proofs are presented in the appendix only.

\section{Preliminaries}
\label{sec:prelims}
\subsection{Words, regular languages and classes}

We fix a finite \emph{alphabet} $A$ for the paper. As usual, $A^*$ denotes the set of all finite words over $A$, including the empty word \veps. We let $A^+ = A^* \setminus \{\veps\}$. For $u,v \in A^*$, we let $uv$ be the word obtained by concatenating $u$ and $v$. A \emph{language} is a subset of $A^*$. We denote the singleton language $\{u\}$ by $u$. We lift concatenation to languages: for $K,L \subseteq A^*$, we let $KL = \{uv \mid u \in K \text{ and } v \in L\}$. We shall consider \emph{marked products}: given languages $L_0,\dots,L_n \subseteq A^*$, a marked product of $L_0,\dots,L_n$ is a product of the form $L_0a_1L_1 \cdots a_nL_n$ where $a_1,\dots,a_n \in A$ (note that ``$L_0$'' is a marked product: this is the case $n =0$).

\smallskip
\noindent
{\bf Regular languages.} In the paper, we consider \emph{regular} languages. A nondeterministic finite automaton (\nfa) is a pair $\As= (Q,\delta)$ where $Q$ is a finite set of states, and $\delta \subseteq Q \times A \times Q$ is a set of transitions. We now define the languages recognized by \As. Given $q,r \in Q$ and $w \in A^*$, we say that there exists a \emph{run labeled by $w$ from $q$ to $r$} (in \As) if there exist $q_0,\dots,q_n \in Q$ and $a_1,\dots,a_n \in A$ such that $w = a_1 \cdots a_n$, $q_0 = q$, $q_n = r$ and $(q_{i-1},a_i,q_i) \in \delta$ for every $1 \leq i \leq n$. Given two sets $I,F\subseteq Q$, we write $\alauto{I}{F} \subseteq A^*$ for the language of all words $w \in A^*$ such that there exist $q \in I$, $r \in F$, and a run labeled by $w$ from $q$ to $r$ in \As. We say that a language $L \subseteq A^*$ is \emph{recognized} by \As if and only if there exist $I,F \subseteq Q$ such that $L = L_\As(I,F)$. The regular languages are those which can be recognized by an \nfa.

We also use \nfas with \emph{\veps-transitions}. In such an \nfa $\As = (Q,\delta)$, a transition may also be labeled by the empty word ``\veps'' (that is, $\delta\subseteq Q \times (A \cup \{\veps\} )\times Q$). We use the standard semantics: an \veps-transition can be taken without consuming an input letter. Note that unless otherwise specified, the \nfas that we consider are assumed to be  \emph{without}~\veps-transitions.


\smallskip
\noindent
{\bf Classes.}  A \emph{class} of languages is a set of languages. A \emph{lattice} is a class containing $\emptyset$ and $A^*$ and closed under both union and intersection. Moreover, a \emph{Boolean algebra} is a lattice closed under complement. Finally, a class \Cs is \emph{quotient-closed} when for all $L \in \Cs$ and all $v \in A^*$, the languages $v\inv L = \{w \in A^* \mid vw \in L\}$ and $Lv\inv =  \{w \in A^* \mid wv \in L\}$ both belong to \Cs as well.  A \emph{\pvari} (resp. a~\emph{\vari}) is a quotient-closed lattice (resp.~a quotient-closed Boolean algebra) containing \emph{regular languages only}.


\smallskip
\noindent
{\bf Group languages.} A \emph{monoid} is a set $M$ equipped with a multiplication \mbox{$s,t \mapsto st$}, which is associative and has a neutral element denoted by ``$1_M$''. Observe that $A^*$ endowed with concatenation is a monoid (\veps is the neutral element). It is well-known that a language $L$ is regular if and only if it is \emph{recognized} by a morphism $\alpha: A^* \to M$ into a \emph{finite} monoid $M$, \emph{i.e.}, there exists $F\subseteq M$ such that $L=\alpha\inv(F)$. We now restrict this definition: a monoid $G$~is a \emph{group} if every element $g \in G$ has an inverse $g\inv \in G$, \emph{i.e.}, such that \hbox{$gg\inv = g\inv g = 1_G$}. A ``\emph{group language}'' is a language recognized by a morphism into a \emph{finite~group}.

%
%

We consider classes \Gs that are group \varis (\emph{i.e.}, containing group languages only). We let \grp be the class of \emph{all} group languages. Another important example is the class \abg of \emph{alphabet modulo testable languages}. For every $w\in A^*$ and every $a \in A$, we write $\#_a(w) \in \nat$ for the number of occurrences of ``$a$'' in $w$. The class \abg consists in all finite Boolean combinations of languages $\{w \in A^* \mid \#_a(w) \equiv k \bmod m\}$ where $a \in A$ and $k,m \in \nat$ are such that $k < m$. One may verify that these are exactly the languages recognized by commutative groups. We also consider the class \md, which consists in all finite Boolean combinations of languages $\{w \in A^* \mid |w| \equiv k \bmod m\}$ with $k,m \in \nat$ such that $k < m$. Finally, we write \stzer for the trivial class $\stzer = \{\emptyset,A^*\}$. One may verify that \grp, \abg, \md and \stzer are all group \varis.

One may verify that $\{\veps\}$ and $A^+$ are \emph{not} group languages. This motivates the next definition: the \emph{\wsuit extension of a class \Cs}, denoted by $\Cs^+$, consists of all languages of the form $L \cap A^+$ or $L \cup \{\veps\}$ where $L \in \Cs$. The next lemma follows from the~definition.

\begin{lemma}\label{lem:wsuit}
	Let \Cs be a \vari. Then, $\Cs^+$ is a \vari containing $\{\veps\}$ and $A^+$.
\end{lemma}

\subsection{Polynomial and Boolean closure}

We investigate two operators that one may apply to a class \Cs. The \emph{Boolean closure} of \Cs, written \bool{\Cs}, is the least Boolean algebra containing \Cs. The \emph{polynomial closure} of \Cs, denoted by \pol{\Cs}, consists of all finite unions of marked products $L_0a_1L_1 \cdots a_nL_n$ where $L_0,\dots,L_n  \in \Cs$ and $a_1,\dots,a_n\in A$. Finally, we write \bpol{\Cs} for \bool{\pol{\Cs}}. If \Cs is a \vari, then \pol{\Cs} is a \pvari and \bpol{\Cs} is a \vari. Proving that \pol{\Cs} is closed under intersection is not immediate. It was shown by Arfi~\cite{arfi87} (see~also~\cite{jep-intersectPOL,PZ:generic18}).

\begin{restatable}{theorem}{bpolvar}\label{thm:bpolvar}
	If\/ \Cs is a \vari, \pol{\Cs} is a \pvari and \bpol{\Cs} is a~\vari.
\end{restatable}

The two operators \poln and \booln induce standard classifications called concatenation hierarchies: for a \vari \Cs, the \emph{concatenation hierarchy of basis \Cs} is built from \Cs by alternatively applying the operators \poln and \booln. We are interested in \bpol{\Cs}, which is level \emph{one} in the concatenation hierarchy of basis~\Cs. We look at bases that are either a group \vari \Gs or its \wsuit extension~$\Gs^+$. Most of the prominent concatenation hierarchies in the literature use such bases. This is in part motivated by the logical characterization of concatenation hierarchies, due to Thomas~\cite{ThomEqu}. We briefly recall it for the level one.

Consider a word $w = a_1 \cdots a_{|w|} \in A^*$. We view $w$ as a linearly ordered set of $|w|+2$ positions $\{0,1,\dots,|w|,|w|+1\}$ such that each position $1 \leq i \leq |w|$ carries the label $a_i \in A$ (on the other hand, $0$ and $|w|+1$ are artificial unlabeled leftmost and rightmost positions). We use first-order logic to describe properties of words: a sentence can quantify over the positions of a word and use a predetermined set of predicates to test properties of these positions. We also allow two constants ``$min$'' and ``$max$'' interpreted as the artificial unlabeled positions $0$ and $|w|+1$ in a given word $w$. A first-order sentence $\varphi$ defines the language of all words satisfying the property stated by $\varphi$. We use several kinds of predicates. For each $a \in A$, we associate a unary predicate (also denoted by~$a$), which selects the positions labeled by ``$a$''. We also use two binary predicates: the (strict) linear order ``$<$'' and the successor relation~``$+1$''. Finally, we associate a set of predicates $\frP_{\Gs}$ to each group \vari \Gs. Every $L \in \Gs$ yields a unary predicate $P_L$ in $\frP_{\Gs}$, which is interpreted as follows. Let $w = a_1 \cdots a_{|w|} \in A^*$. The unary predicate $P_L$ selects all positions $i \in \{0,\dots,|w|+1\}$ such that $i \neq 0$ and $a_1 \cdots a_{i-1} \in L$.

\begin{example}
  The sentence ``$\exists x \exists y\ (x < y) \wedge a(x) \wedge b(y)$'' defines the language $A^*aA^*bA^*$. The sentence ``$\exists x \exists y\ a(x) \wedge c(y) \wedge (y+1 = max)$'' defines $A^*aA^*c$. Finally, if $L = (AA)^* \in \md$ (the words of even length), the sentence ``$\exists x\ a(x) \wedge P_L(x)$'' defines the language $(AA)^*aA^*$.
\end{example}

The fragment of first-order logic containing exactly the Boolean combinations of existential first-order sentences is denoted by ``\bsc{1}''. Let \Gs be a group \vari. We write $\bsc{1}(<,\frP_{\Gs})$ for the class of all languages  defined by a sentence of \bsc{1} using only the label predicates, the linear order ``$<$'' and those in $\frP_{\Gs}$. Moreover, we write $\bsc{1}(<,+1,\frP_{\Gs})$ for the class of all languages defined by a sentence of \bsc{1}, which additionally allows the successor predicate ``$+1$''. The following proposition follows from the results of~\cite{PZ:generic18,place2022characterizing}.

\begin{proposition}\label{prop:logcar}
	Let \Gs be a group \vari. We have $\bpol{\Gs} = \bsc{1}(<,\frP_{\Gs})$ and $\bpol{\Gs^+} =  \bsc{1}(<,+1,\frP_{\Gs})$.
\end{proposition}

\noindent
{\bf Key examples.} The basis $\stzer = \{\emptyset,A^*\}$ yields the \emph{Straubing-Thérien hierarchy}~\cite{StrauConcat,TheConcat} (hence the notation of this basis). Its level one is the class of piecewise testable languages~\cite{simonthm}. Its \wsuit extension $\stzer^+$ induces the \emph{dot-depth hierarchy}~\cite{BrzoDot}. In particular, \bpol{\stzer} and \bpol{\stzer^+} correspond to \bsw{1} and \bsws{1}, as all predicates in $\frP_{\stzer}$ are trivial. The hierarchies of bases \md and $\md^+$ are also prominent (see for example~\cite{ChaubardPS06,MACIEL2000135,Zetzsche18}). The classes \bpol{\md} and \bpol{\md^+} correspond to \bswm{1} and \bswsm{1} where ``$\mathit{MOD}$'' is the set of \emph{modular predicates} (for all $r,q \in \nat$ such that $r < q$, it contains a unary predicate $M_{r,q}$ selecting the positions $i$ such that $i \equiv r \bmod q$). Similarly, \bpol{\abg} and \bpol{\abg^+} correspond to \bswam{1} and \bswsam{1} where ``$\mathit{AMOD}$'' is the set of \emph{alphabetic modular predicates} (for all $a \in A$ and $r,q \in \nat$ such that $r < q$, it contains a unary predicate $M^a_{r,q}$ selecting the positions $i$ such the that number of positions $j < i$ with label $a$ is congruent to $r$ modulo $q$). Finally, the group hierarchy, whose basis is \grp is also prominent~\cite{margpin85,henckell:hal-00019815}, though its logical characterization is not standard.

\smallskip
\noindent
{\bf Properties.} We present a key ingredient~\cite[Lemma~3.6]{pzbpolcj} (we provide a proof in Appendix~\ref{app:prelims}). It describes a concatenation principle for the classes \bpol{\Cs} based on the notion of ``cover''. Given a language $L$, a cover of $L$ is a \emph{finite} set \Kb of languages satisfying $L \subseteq \bigcup_{K \in \Kb} K$. If \Ds is a class, a \Ds-cover of $L$ is a cover \Kb of $L$ such that $\Kb \subseteq \Ds$.

\begin{restatable}{proposition}{bconcat}\label{prop:bconcat}
  Let \Cs be a \vari, $n \in \nat$, $L_0,\dots,L_n \in \pol{\Cs}$ and $a_1,\dots,a_n \in A$. If $\Hb_i$ is a \bpol{\Cs}-cover of $L_i$ for all $i \leq n$, then there is a \bpol{\Cs}-cover \Kb of $L_0a_1L_1 \cdots a_nL_n$ such that for all $K \in \Kb$, there exists $H_i \in \Hb_i$ for each $i \leq n$ satisfying $K \subseteq H_0a_1H_1 \cdots a_nH_n$.
\end{restatable}

For applying Proposition~\ref{prop:bconcat}, we need a language $L_0a_1L_1 \cdots a_nL_n$ with $L_0,\dots,L_n \in \pol{\Cs}$. The next tailored statements build such languages when $\Cs = \Gs$ or $\Gs^+$ for a group \vari~\Gs (see App.~\ref{app:prelims} for proofs). While simple, these results are central: this is the unique place where we use the fact that \Gs contains only \emph{group languages}. Let $L \subseteq A^*$. With every word $w=a_1\cdots a_n \in A^*$, we associate the language $\uclos_L w = La_1L \cdots a_nL \subseteq A^*$ (we let $\uclos_L \veps = L$). We first present the statement for the case $\Cs = \Gs$, which can also be found in~\cite[Prop.~3.11]{cano:hal-01247172}.

\begin{restatable}{proposition}{pgcov}\label{prop:pgcov}
  Let $H \subseteq A^*$ be a language and $L \subseteq A^*$ be a group language containing~$\veps$. There exists a cover \Kb of $H$ such that every $K \in \Kb$ is of the form $K = \uclos_L w$ for some $w \in H$.
\end{restatable}

The next statement, useful for the case $\Cs = \Gs^+$, is a corollary of Proposition~\ref{prop:pgcov}. Let $\As = (Q,\delta)$ be an \nfa. Moreover, let $w,z \in A^*$. We say that $z$ is a \emph{left \As-loop} for $w$ if for every $q,r \in Q$ such that $w \in \alauto{q}{r}$, there exists $s \in Q$ such that $z \in \alauto{q}{s} \cap \alauto{s}{s}$ and $zw \in \alauto{s}{r}$ (in particular, $zz^*zw \subseteq \alauto{q}{r}$). Symmetrically, we say that $z$ is a \emph{right \As-loop} for $w$ if for every $q,r \in Q$ such that $w \in \alauto{q}{r}$, there exists $s \in Q$ such that $wz \in \alauto{q}{s}$ and $z \in \alauto{s}{s} \cap \alauto{s}{r}$ (in particular, $wzz^*z \subseteq \alauto{q}{r}$).

Now, given an arbitrary word $w \in A^*$, an \emph{\As-guarded decomposition of $w$} is a tuple $(w_1,\dots, w_{n+1})$ for some $n \in \nat$ where $w_1 \in A^*$ and $w_i \in A^+$ for $2 \leq i \leq n+1$, and such that $w = w_1 \cdots w_{n+1}$ and, if $n \geq 1$, then for every $i$ satisfying $1 \leq i \leq n$, there exists a \emph{nonempty} word $z_i \in A^+$ which is a right \As-loop for $w_i$ and a left \As-loop for $w_{i+1}$.

\begin{restatable}{proposition}{pgpcov}\label{prop:pgpcov}
  Let $H \subseteq A^*$ be a language, \As be an \nfa and $L \subseteq A^*$ be a group language containing~$\veps$. There exists a cover \Kb of $H$ such that for each $K \in \Kb$, there exist a word $w \in H$ and an \As-guarded decomposition $(w_1,\dots, w_{n+1})$ of $w$ for some $n \in \nat$ such that  $K = w_1 L \cdots w_nLw_{n+1}$ (if $n=0$, then $K=\{w_1\}$).
\end{restatable}



\section{Separation framework}
\label{sec:separ}
In order to investigate a given class \Cs, we rely on a generic decision problem that one may associate to it: \emph{\Cs-separation}. We first define it and then present a variant, ``tuple separation'', that we shall require as a proof ingredient. The missing proofs are presented in Appendix~\ref{app:separ}. 

\subsection{The separation problem}

Consider two languages $L_0,L_1 \subseteq A^*$. We say that a third language $K \subseteq A^*$ \emph{separates} $L_0$ from $L_1$ when $L_0 \subseteq K$ and $K \cap L_1 = \emptyset$. Then, given an arbitrary class \Cs, we say that $L_0$ is \emph{\Cs-separable} from $L_1$ when there exists $K \in \Cs$ that separates $L_0$ from $L_1$. For every class \Cs, the \emph{\Cs-separation problem} takes \emph{two} regular languages $L_0$ and $L_1$ as input (in the paper, they are represented by \nfas) and asks whether $L_0$ is \Cs-separable from $L_1$. We complete the definition with a useful result, which holds when \Cs is a \pvari.

\begin{restatable}{lemma}{sepconcat} \label{lem:sepconcat}
  Let \Cs be a \pvari and $L_0,L_1,H_0,H_1 \subseteq A^*$. If $L_0$ is not \Cs-separable from $L_1$ and $H_0$ is not \Cs-separable from $H_1$ then $L_0H_0$ is not \Cs-separable from~$L_1H_1$.
\end{restatable}

In the paper, we look at \Cs-separation when $\Cs = \bpol{\Gs}$ or $\bpol{\Gs^+}$ for a group \vari~\Gs. We prove that in these two cases, there are polynomial time (Turing) reductions to \Gs-separation. We now introduce terminology that we shall use to present the algorithms.

\smallskip
\noindent
{\bf Framework.} Consider a class \Cs and an \nfa $\As = (Q,\delta)$. We associate a set $\nsepca \subseteq Q^4$: the \emph{inseparable \Cs-quadruples} associated to \As. We define,
\[
  \nsepca = \big\{(q,r,s,t)\in Q^4 \mid \alauto{q}{r} \text{ is \underbar{not} \Cs-separable from } \alauto{s}{t}\big\}.
\]
The next easy result connects \Cs-separation to this set, for input languages given by \nfas.

\begin{restatable}{proposition}{autosep} \label{prop:autosep}
  Let \Cs be a lattice. Consider an \nfa $\As = (Q,\delta)$ and four sets of states $I_1,F_1,I_2,F_2 \subseteq Q$. The two following conditions are equivalent:
  \begin{enumerate}
    \item $\alauto{I_1}{F_1}$ is \Cs-separable from $\alauto{I_2}{F_2}$.
    \item $\left(I_1 \times F_1 \times I_2 \times F_2\right) \cap \nsepca = \emptyset$.
  \end{enumerate}
\end{restatable}

\medskip

Clearly, given as input two regular languages recognized by \nfas, one may compute in polynomial time a single \nfa recognizing both languages. Hence, Proposition~\ref{prop:autosep}  yields a polynomial time reduction from \Cs-separation to the problem of computing $\nsepca$ from an input \nfa. Naturally, this does not necessarily mean that there exists a polynomial time algorithm for \Cs-separation: depending on \Cs, computing \nsepca may or may not be costly.

\smallskip
We introduce a key definition for manipulating \nsepca, for an \nfa $\As = (Q,\delta)$.  Let~\mbox{$S \subseteq Q^4$} and \Kb be  a finite set of languages. We say that \emph{\Kb is separating for $S$} when for every $(q,r,s,t) \in Q^4$ and every $K \in \Kb$, if $K$ intersects both $\alauto{q}{r}$ and $\alauto{s}{t}$, then \mbox{$(q,r,s,t) \in S$}. Then, \nsepca is the smallest set of 4-tuples admitting a \Cs-cover of $A^*$ which is separating for it.

\begin{restatable}{lemma}{lautosep} \label{lem:autosep}
  Let \Cs be a Boolean algebra and $\As = (Q,\delta)$ be an \nfa. Then the following~holds:
  \begin{itemize}
    \item There exists a \Cs-cover \Kb of $A^*$ which is separating for \nsepca.
    \item Let $S \subseteq Q^4$. If there exists a \Cs-cover \Kb of $A^*$ which is separating for $S$, then $\nsepca \subseteq S$.
  \end{itemize}
\end{restatable}

\smallskip
\noindent
{\bf Controlled separation.} We present additional terminology tailored to the classes built from a group \vari. Consider two classes \Cs and \Ds  (in practice, \Ds will be a group \vari \Gs and \Cs will be either \bpol{\Gs} or \bpol{\Gs^+}). Let $L_0,L_1 \subseteq A^*$. We say that $L_0$ is \Cs-separable from $L_1$ \emph{under \Ds-control} if there exists $H\in \Ds$ such that $\veps \in H$ and $L_0 \cap H$ is \Cs-separable from $L_1 \cap H$. Given an \nfa $\As = (Q,\delta)$, we associate a set $\dnsepca \subseteq Q^4$: 
\[
  \dnsepca = \big\{(q,r,s,t) \in Q^4 \mid \alauto{q}{r} \text{ is \underbar{not} \Cs-separable from } \alauto{s}{t} \text{ under \Ds-control}\big\}.
\] 
Clearly, we have $\dnsepca \subseteq \nsepca$. Let us connect this new definition to the notion of separating cover presented above. In this case as well, this will be useful in proof arguments.


\begin{restatable}{lemma}{gautosep} \label{lem:gautosep}
  Let $\Cs$ and $\Ds$ be Boolean algebras such that $\Ds \subseteq \Cs$ and let $\As = (Q,\delta)$ be an \nfa. The following properties hold:
  \begin{itemize}
    \item There exists $L \in \Ds$ with $\veps \in L$, and a \Cs-cover \Kb of $L$ which is separating for \dnsepca.
    \item Let $S \subseteq Q^4$. If there exist $L \in \Ds$ with $\veps \in L$, and a \Cs-cover \Kb of $L$ which is separating for $S$, then $\dnsepca \subseteq S$.
  \end{itemize}
\end{restatable}

This notion is only useful if $\{\veps\} \not\in \Ds$. If $\{\veps\} \in \Ds$, then $L_0$ is \Cs-separable from $L_1$ under \Ds-control if and only if either $\veps \not\in L_0$ or $\veps \not\in L_1$. This is why the notion is designed for group \varis: if \Gs is such a class, then $\{\veps\} \not\in \Gs$. In this case, if $\Cs  \in \{\Gs,\Gs^+\}$, then the set \gnsepbca carries more information than \nsepbca. This is useful for the computation: rather than computing \nsepbca directly, our procedures first compute \gnsepbca. The proof is based on Propositions~\ref{prop:bconcat} and~\ref{prop:pgcov} (the latter requires \Gs to consist of group languages).

\begin{restatable}{proposition}{compgnsep} \label{prop:compgnsep}
  Let \Gs be a group \vari, let \Cs be a \vari such that $\Gs \subseteq \Cs$ and let $\As = (Q,\delta)$ be an \nfa. Then, \nsepbca is the least set $S \subseteq Q^4$ that contains \gnsepbca and satisfies the two following conditions:
  \begin{enumerate}
    \item\label{c:1} For all $q,r,s,t\in Q$ and $a \in A$, if $(q,a,r),(s,a,t) \in \delta$, then $(q,r,s,t) \in S$.
    \item\label{c:2} For all $(q_1,r_1,s_1,t_1), (q_2,r_2,s_2,t_2) \in S$, if $r_1 = q_2$ and $t_1 = s_2$, then $(q_1,r_2,s_1,t_2) \in S$.
  \end{enumerate}
\end{restatable}

\begin{proof}
  Let $S \subseteq Q^4$ be the least set containing \gnsepbca and satisfying both conditions. We prove that $S = \nsepbca$. For $S \subseteq \nsepbca$, since $\gnsepbca \subseteq \nsepbca$ by definition, it suffices to prove that $\nsepbca$ satisfies both conditions in the proposition. First, consider $a \in A$ and $q,r,s,t \in Q$ such that $(q,a,r),(s,a,t) \in \delta$. We have $a \in \alauto{q}{r}$  and $a \in \alauto{s}{t}$. Hence, they are not \bpol{\Cs}-separable and $(q,r,s,t) \in \nsepbca$. Now, let $(q_1,r_1,s_1,t_1), (q_2,r_2,s_2,t_2) \in \nsepbca$ such that $r_1 = q_2$ and $t_1 = s_2$. For $i \in \{1,2\}$, we know that \alauto{q_i}{r_i} is not \bpol{\Cs}-separable from $\alauto{s_i}{t_i}$. Since \bpol{\Cs} is a \vari by Theorem~\ref{thm:bpolvar}, it follows from Lemma~\ref{lem:sepconcat} that $\alauto{q_1}{r_1}\alauto{q_2}{r_2}$ is  not \bpol{\Cs} separable from $\alauto{s_1}{t_1}\alauto{s_2}{t_2}$. Since $r_1 = q_2$ and $t_1 = s_2$, it is immediate that $\alauto{q_1}{r_1}\alauto{q_2}{r_2} \subseteq \alauto{q_1}{r_2}$ and $\alauto{s_1}{t_1}\alauto{s_2}{t_2} \subseteq \alauto{s_1}{t_2}$. Hence,  \alauto{q_1}{r_2} is not \bpol{\Cs}-separable from \alauto{s_1}{t_2}  and we get $(q_1,r_2,s_1,t_2) \in \nsepbca$ as desired.

  We turn to the inclusion $\nsepbca \subseteq S$. By Lemma~\ref{lem:gautosep}, there exists $L \in \Gs$ such that $\veps \in L$ and a \bpol{\Cs}-cover \Vb of $L$ which is separating for \gnsepbca. By hypothesis, $L$ is a group language and $\veps \in L$. Hence, Proposition~\ref{prop:pgcov} yields a cover \Pb of $A^*$ such that every $P \in \Pb$ is of the form $P = \uclos_L w_P$ for some word $w_P \in A^*$. Let $P \in \Pb$ and $a_1,\dots, a_n \in A$ be the letters such that $w_P = a_1 \cdots a_n$. We have $P = La_1L \cdots a_nL$ by definition (if $w_P = \veps$, then $P = L$). By definition, $L \in \Gs \subseteq \pol{\Cs}$. Hence, since \Vb is a \bpol{\Cs}-cover of $L$, Proposition~\ref{prop:bconcat} yields a \bpol{\Cs}-cover $\Kb_P$ of $P$ such that for every $K \in \Kb_P$, there are $V_0, \dots,V_n \in \Vb$ such that $K\subseteq V_0a_1V_1 \cdots a_nV_n$. We let $\Kb = \bigcup_{P \in \Pb} \Kb_P$. Since \Pb is a cover of $A^*$ and $\Kb_P$ is a \bpol{\Cs}-cover of $P$ for each $P \in \Pb$, $\Kb$ is a \bpol{\Cs}-cover of $A^*$. We show that \Kb is separating for $S$ which implies that $\nsepbca \subseteq S$ by Lemma~\ref{lem:autosep}.

  Let $(q,r,s,t) \in Q^4$ and $K \in \Kb$ such that we have $x \in K \cap \alauto{q}{r}$ and $y \in K \cap \alauto{s}{t}$. We show that $(q,r,s,t) \in S$. We have $K \in \Kb_P$ for some $P \in \Pb$. Let $a_1,\dots,a_n\in A$ such that $w_P = a_1 \cdots a_n$. By definition, there are $V_0, \dots,V_n \in \Vb$ such that $K \subseteq V_0a_1V_1 \cdots a_nV_n$. Since $x,y \in K$, we get $x_i,y_i \in V_i$ for $0 \leq i \leq n$ such that $x = x_0a_1x_1 \cdots a_n x_n$ and $y = y_0a_1y_1 \cdots a_n y_n$. Since $x \in \alauto{q}{r}$, we get $q_i,r_i \in Q$ for $0 \leq i \leq n$ such that $q_0 = q$, $r_n = r$, $x_i \in \alauto{q_i}{r_i}$ for $0 \leq i \leq n$ and $(r_{i-1},a_i,q_i) \in \delta$ for $1 \leq i \leq n$. Finally, since $y \in \alauto{s}{t}$, we get $s_i,t_i \in Q$ for $0 \leq i \leq n$ such that $s_0 = s$, $t_n = t$, $y_i \in \alauto{s_i}{t_i}$ for $0 \leq i \leq n$ and $(t_{i-1},a_i,s_i) \in \delta$ for $1 \leq i \leq n$. Since $S$ satisfies Condition~\ref{c:1} in the proposition, we get $(r_{i-1},q_i,t_{i-1},s_i) \in S$ for $1 \leq i \leq n$. Since $V_i \in \Vb$ which is separating for \gnsepbca and $x_i,y_i \in V_i$, we also get $(q_i,r_i,q_i,t_i) \in \gnsepbca$ for $0 \leq i \leq n$. Thus, Condition~\ref{c:2} in the proposition yields $(q_0,r_0,s_n,t_n) \in S$, \emph{i.e.} $(q,r,s,t) \in S$ as desired.
\end{proof}

Proposition~\ref{prop:compgnsep} provides a least fixpoint algorithm for computing the set \nsepbca from \gnsepbca. Combined with Proposition~\ref{prop:autosep}, this yields a polynomial time reduction from \bpol{\Cs}-separation to computing \gnsepbca from an \nfa. We shall prove that when $\Cs \in \{\Gs,\Gs^+\}$, there are polynomial time reductions of the latter problem to \Gs-separation.  

\subsection{Tuple separation}

This generalized variant of separation is taken from~\cite{pzbpol}. We shall use it as a proof ingredient: for every lattice~\Cs, it is connected to the classical separation problem for \bool{\Cs}. For every $n \geq 1$, we call ``\emph{$n$-tuple}'' a tuple of $n$ languages $(L_1,\dots,L_n)$. In the sequel, given another language $K$, we shall write $(L_1,\dots,L_n) \cap K$ for the $n$-tuple $(L_1 \cap K, \dots, L_n \cap K)$. Let \Cs be a lattice, we use induction on $n$ to define the \emph{\Cs-separable $n$-tuples}:
\begin{itemize}
  \item If $n = 1$, a $1$-tuple $(L_1)$ is \Cs-separable when $L_1 = \emptyset$.
  \item If $n \geq 2$, an $n$-tuple $(L_1,\dots,L_n)$ is \Cs-separable when there exists $K \in \Cs$ such that $L_1 \subseteq K$ and $(L_2,\dots,L_n) \cap K$ is \Cs-separable. We call $K$ a \emph{separator} of $(L_1,\dots,L_n)$.
\end{itemize}
One may verify that classical separation is the special case $n = 2$. We generalize \Ds-controlled separation to this setting. For a class \Ds, we say that an $n$-tuple $(L_1,\dots,L_n)$ is \Cs-separable under \Ds-control if there exists $H\in\Ds$ such that $\veps \in H$ and $(L_1,\dots,L_n) \cap H$ is \Cs-separable.

We complete the definition with two simple properties of tuple separation (see Appendix~\ref{app:separ} for the proofs). The second one is based on closure under quotients and generalizes Lemma~\ref{lem:sepconcat}.

\begin{restatable}{lemma}{tuptriv} \label{lem:tuptriv}
  Let \Cs be a lattice and let $(L_1,\dots,L_n),(H_1,\dots,H_n)$ be two $n$-tuples. If $L_1 \cap \cdots \cap L_n \neq \emptyset$, then $(L_1,\dots,L_n)$ is not \Cs-separable. Moreover, if $L_i \subseteq H_i$ for every $i \leq n$ and $(L_1,\dots,L_n)$ is not \Cs-separable, then $(H_1,\dots,H_n)$ is not \Cs-separable either.
\end{restatable}

\begin{restatable}{lemma}{tupconcat} \label{lem:tupconcat}
  Let \Cs be a \pvari, $n \geq 1$ and let $(L_1,\dots,L_n),(H_1,\dots,H_n)$ be two $n$-tuples, which are not \Cs-separable. Then, $(L_1H_1,\dots,L_nH_n)$ is not \Cs-separable either.
\end{restatable}


A theorem of~\cite{pzbpol} connects tuple \Cs-separation for a lattice \Cs to \bool{\Cs}-separation: $L_0$ is \bool{\Cs}-separable from $L_1$ if and only if $(L_0,L_1)^p$ is \Cs-separable for some $p \geq 1$. Here, $(L_0,L_1)^p$ denotes the $2p$-tuple obtained by concatenating $p$ copies of $(L_0,L_1)$. For example, $(L_0,L_1)^3 = (L_0,L_1,L_0,L_1,L_0,L_1)$. We use a corollary applying to \Ds-controlled separation. Proofs for both the original theorem of~\cite{pzbpol} and the corollary are available in Appendix~\ref{app:separ}.

\begin{restatable}{corollary}{ccovtsep} \label{cor:covtsep}
  Let $\Cs$ and $\Ds$ be two lattices such that $\Ds \subseteq \Cs$ and let $L_0,L_1 \subseteq A^*$. The following properties are equivalent:
  \begin{enumerate}
    \item $L_0$ is \bool{\Cs}-separable from $L_1$ under \Ds-control.
    \item There exists $p \geq 1$ such that $(L_0,L_1)^p$ is \Cs-separable under \Ds-control.
  \end{enumerate}
\end{restatable}

\smallskip

We only use the contrapositive of $1) \Rightarrow 2)$ in Corollary~\ref{cor:covtsep}. We complete the presentation with two important lemmas about tuple separation for \pol{\Ds} and \pol{\Ds^+}. We use them to prove that tuples are not separable (see Appendix~\ref{app:separ} for the proofs). Note that in practice, \Ds will be a group \vari \Gs. Yet, the results are true regardless of this hypothesis.

\begin{restatable}{lemma}{pgsound} \label{lem:pgsound}
  Let \Ds be a \vari and $(L_1,\dots,L_n)$ an $n$-tuple which is not \pol{\Ds}-separable under \Ds-control. Then, $(\{\veps\},L_1,\dots,L_n)$ is not \pol{\Ds}-separable.
\end{restatable}
\begin{restatable}{lemma}{pgpsound} \label{lem:pgpsound}
  Let \Ds be a \vari and $w \in A^+$. If $(L_1,\dots,L_n)$ is not \pol{\Ds^+}-separable under \Ds-control, then $(w^+,w^+L_1w^+,\dots,w^+L_nw^+)$ is not \pol{\Ds^+}-separable.
\end{restatable}

\section{\texorpdfstring{Separation Algorithms for \bpol{\Gs} and \bpol{\Gs^+}}{Separation Algorithms for BPol(G) and BPol(G+)}}
\label{sec:bpolg}
\newcommand{\gnsbg}{\gnsep{\bpol{\Gs}}{\As}}
\newcommand{\tautg}{\ensuremath{\tau_{\As,\Gs}}\xspace}
\newcommand{\gnsbgp}{\gnsep{\bpol{\Gs^+}}{\As}}
\newcommand{\betg}{\ensuremath{\tau^+_{\As,\Gs}}\xspace}

For a group \vari \Gs, we now consider \bpol{\Gs}- and \bpol{\Gs^+}-separation. We rely on the notions of Section~\ref{sec:separ}: given an arbitrary \nfa $\As = (Q,\delta)$, we~present a generic characterization of the inseparable \bpol{\Gs}- and \bpol{\Gs^+}-quadruples under~\Gs control associated to \As, \emph{i.e.}, of the subsets $\gnsbg$ and $\gnsbgp$ of $Q^4$. Thanks to Proposition~\ref{prop:compgnsep}, this also yields characterizations of \nsep{\bpol{\Gs}}{\As}  and of \nsep{\bpol{\Gs^+}}{\As}, which in turn, in view of Proposition~\ref{prop:autosep}, yield reductions from both \bpol{\Gs}- and \bpol{\Gs^+}-separation to \Gs-separation. These polynomial time reductions are therefore \emph{effective} when \Gs-separation~is~decidable.

\subsection{Statements}

Let \Gs be a group \vari and let $\As = (Q,\delta)$ be an \nfa. We present characterizations of  $\gnsbg$ and $\gnsbgp$. They follow the same pattern, but each of them depends on a specific function from $2^{{Q^{4}}}$ to $2^{{Q^{4}}}$, which we first~describe.

\smallskip\noindent{\textbf{Characterization of \gnsbg.}} We use a function  $\tautg:2^{{Q^{4}}}\to2^{{Q^{4}}}$. For $S \subseteq Q^4$, we define the set $\tautg(S) \subseteq Q^4$. The definition is based on an auxiliary \nfa $\Bs_{S} = (Q^3,\gamma_{S})$ \emph{\underbar{with \veps-transitions}}, which depends on $S$. Its states are triples in~$Q^3$. The set $\gamma_{S} \subseteq Q^3 \times  (A \cup \{\veps\}) \times Q^3$ includes two kinds of transitions. First, given $a \in A$ and $s_1,s_2,s_3,\,t_1,t_2,t_3 \in Q$, we let $\big((s_1,s_2,s_3),a,(t_1,t_2,t_3)\big) \in \gamma_S$ if and only if $(s_1,a,t_1) \in \delta$, $(s_2,a,t_2) \in \delta$ and $(s_3,a,t_3) \in \delta$. Second, for every state $q_1 \in Q$ and every $(q_2,r_2,q_3,r_3) \in S$, we add the following \veps-transition: $((q_1,q_2,q_3),\veps,(q_1,r_2,r_3)) \in \gamma_S$. We represent this construction process graphically in Figure~\ref{fig:gauto}.

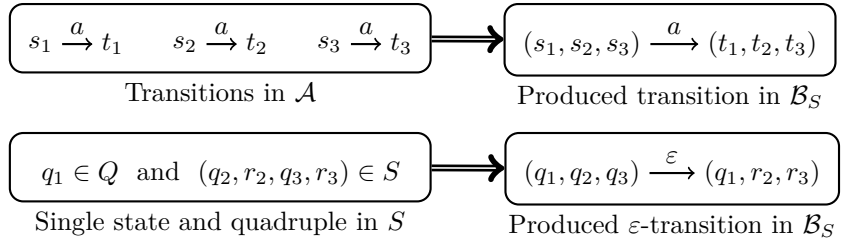
\begin{figure}[htb!]
    \centering
    \begin{tikzpicture}[scale=0.95]

        \node (s1) at (0.0,0) {$s_1$};
        \node (t1) at (1.0,0) {$t_1$};
        
        \node (q2) at (2.0,0.0) {$s_2$};
        \node (r2) at (3.0,0.0) {$t_2$};

        \node (q3) at (4.0,0.0) {$s_3$};
        \node (r3) at (5.0,0.0) {$t_3$};

        \draw[thick,->] (s1) to node[above] (l1) {$a$} (t1);
        
        \draw[thick,->] (q2) to node[above] {$a$} (r2);
        \draw[thick,->] (q3) to node[above] {$a$} (r3);
        
        \node[fit=(s1)(r3)(l1),draw,thick,rounded corners=5pt,label=below:Transitions in \As] (box1) {};

        \node (b) at (7.5,0.0) {$(s_1,s_2,s_3)$};
        \node (c) at (10.0,0.0) {$(t_1,t_2,t_3)$};
        \draw[thick,->] (b) to node[above] (l2) {$a$} (c);
        \node[fit=(l2)(b)(c),draw,thick,rounded corners=5pt,label=below:Produced transition in $\Bs_S$] (box2) {};
        \draw[line width=1pt,double,->] (box1) to (box2);

        \node (h1) at (2.5,-1.8) {$q_1 \in Q$ and $(q_2,r_2,q_3,r_3) \in S$};

            \node (as1) at (0.0,-1.8) {$s_1$};
        \node (at1) at (1.0,-1.8) {$t_1$};
        
        \node (aq2) at (2.0,-1.8) {$s_2$};
        \node (ar2) at (3.0,-1.8) {$t_2$};

        \node (aq3) at (4.0,-1.8) {$s_3$};
        \node (ar3) at (5.0,-1.8) {$t_3$};

        \draw[thick,->] (as1) to node[above] (al1) {$a$} (at1);
        
        \draw[thick,->] (aq2) to node[above] {$a$} (ar2);
        \draw[thick,->] (aq3) to node[above] {$a$} (ar3);
        
        \node[fit=(as1)(ar3)(al1),fill=white,draw,thick,rounded corners=5pt,label=below:Single state and quadruple in $S$] (box3) {};
        
        \node (h1) at (2.5,-1.8) {$q_1 \in Q\ $ and $\ (q_2,r_2,q_3,r_3) \in S$};

            \node (e) at (7.5,-1.8) {$(q_1,q_2,q_3)$};
        \node (f) at (10.0,-1.8) {$(q_1,r_2,r_3)$};
        \draw[thick,->] (e) to node[above] (l3) {$\veps$} (f);
        \node[fit=(l3)(e)(f),draw,thick,rounded corners=5pt,label=below:Produced \veps-transition in $\Bs_S$] (box4) {};
        \draw[line width=1pt,double,->] (box3) to (box4);

    \end{tikzpicture}
    \caption{Construction of the transitions in the auxiliary automaton $\Bs_S$}
    \label{fig:gauto}
\end{figure}

\begin{remark}
    The \nfa $\Bs_S$ and its counterpart $\Bs_S^+$ (which we define below as a means to handle \bpol{\Gs^+}) are the \emph{only} \nfas with \veps-transitions considered in the paper. In particular, the original input \nfa \As is assumed to be \emph{without} \veps-transitions.
\end{remark}

%
%
%
%
%
%
%
%
%
%
%
%
%
%

We are ready to define $\tautg(S) \subseteq Q^4$. For every $(q,r,s,t) \in Q^4$, we let $(q,r,s,t) \in \tautg(S)$ if and only if the two following conditions hold:
\begin{equation} \label{eq:gauto}
    \begin{array}{l}
        \text{$\{\veps\}$ is \emph{not} \Gs-separable from \lauto{\Bs_S}{(s,q,s)}{(t,r,t)}, and}\\
        \text{$\{\veps\}$ is \emph{not} \Gs-separable from \lauto{\Bs_S}{(q,s,q)}{(r,t,r)}}.
    \end{array}
\end{equation}

A set $S\subseteq Q^4$ is \emph{$(\bpoln,*)$-sound} for \Gs and $\As$ if it is a fixpoint for \tautg, \emph{i.e.} $\tautg(S) = S$. We have the following simple lemma which can be verified from the definition (see Appendix~\ref{app:bpolgp} for the proof). It states that $\tautg:2^{{Q^{4}}}\to2^{{Q^{4}}}$ is \emph{increasing} (for inclusion). In particular, this implies that it has a \emph{greatest fixpoint}, \emph{i.e.}, there is a \emph{greatest $(\bpoln,*)$-sound set}.

\begin{restatable}{lemma}{ginc} \label{lem:ginc}
    Let \Gs be a group \vari and let $\As = (Q,\delta)$ be an \nfa. For every $S,S' \subseteq Q^4$, we have $S \subseteq S' \Rightarrow \tautg(S) \subseteq \tautg(S')$.
\end{restatable}

\smallskip
We may now state the first key theorem of the paper. It applies to \bpol{\Gs}-separation.

\begin{restatable}{theorem}{gauto}\label{thm:gauto}
  Let \Gs be a group \vari and $\As = (Q,\delta)$ an \nfa. Then, \gnsbg is the greatest $(\bpoln,*)$-sound subset of $Q^4$ for \Gs and \As.
\end{restatable}

\medskip\noindent{\textbf{Characterization of \gnsbgp.}} The characterization of \gnsbgp is analogous. Roughly, the only difference is that we modify the definition of the auxiliary automaton $\Bs_S$. Let \Gs be a group \vari and $\As = (Q,\delta)$ be an \nfa. We define a new function \mbox{$\betg: 2^{Q^4} \to 2^{Q^4}$}. For $S \subseteq Q^4$, we define $\betg(S) \subseteq Q^4$ using another auxiliary \nfa $\Bs^+_{S} = (Q^3,\gamma^+_{S})$ \emph{with \veps-transitions}. Its states are triples in $Q^3$ and $\gamma^+_{S} \subseteq Q^3 \times  (A \cup \{\veps\}) \times Q^3$ contains two kinds of transitions. First, for $a \in A$ and $s_1,s_2,s_3,t_1,t_2,t_3 \in Q$, we let $\big((s_1,s_2,s_3),a,(t_1,t_2,t_3)\big) \in \gamma^+_S$ if and only if $(s_1,a,t_1) \in \delta$, $(s_2,a,t_2) \in \delta$ and $(s_3,a,t_3) \in \delta$. Second, for all $q_1 \in Q$ and all $(q_2,r_2,q_3,r_3) \in S$, if  $A^+ \cap \alauto{q_1}{q_1} \cap \alauto{q_2}{q_2} \cap \alauto{q_3}{q_3} \cap \alauto{r_2}{r_2} \cap \alauto{r_3}{r_3} \neq \emptyset$, then we add the following \veps-transition: $((q_1,q_2,q_3),\veps,(q_1,r_2,r_3)) \in \gamma^+_S$. We represent this construction in Figure~\ref{fig:pgauto}.

\begin{figure}[htb!]
    \centering
    \begin{tikzpicture}[scale=0.95]

	  \node (s1) at (-0.8,0) {$s_1$};
	  \node (t1) at (0.4,0) {$t_1$};

	  \node (q2) at (1.6,0.0) {$s_2$};
	  \node (r2) at (2.8,0.0) {$t_2$};

	  \node (q3) at (4.0,0.0) {$s_3$};
	  \node (r3) at (5.2,0.0) {$t_3$};

	  \draw[thick,->] (s1) to node[above] (l1) {$a$} (t1);

	  \draw[thick,->] (q2) to node[above] {$a$} (r2);
	  \draw[thick,->] (q3) to node[above] {$a$} (r3);

	  \node[fit=(s1)(r3)(l1),draw,thick,rounded corners=5pt,label=below:Transitions in \As] (box1) {};

	  \node (b) at (7.5,0.0) {$(s_1,s_2,s_3)$};
	  \node (c) at (10.0,0.0) {$(t_1,t_2,t_3)$};
	  \draw[thick,->] (b) to node[above] (l2) {$a$} (c);
	  \node[fit=(l2)(b)(c),draw,thick,rounded corners=5pt,label=below:Produced transition in $\Bs^+_S$] (box2) {};
	  \draw[line width=1pt,double,->] (box1) to (box2);

	  \node[text=white] (as1) at (-0.8,-1.8) {$s_1$};
	  \node[text=white] (at1) at (0.4,-1.8) {$t_1$};

	  \node[text=white] (aq2) at (1.6,-1.8) {$s_2$};
	  \node[text=white] (ar2) at (2.8,-1.8) {$t_2$};

	  \node[text=white] (aq3) at (4.0,-1.8) {$s_3$};
	  \node[text=white] (ar3) at (5.2,-1.8) {$t_3$};

	  \draw[thick,white,->] (as1) to node[text=white,above] (al1) {$a$} (at1);

	  \draw[thick,white,->] (aq2) to node[text=white,above] {$a$} (ar2);
	  \draw[thick,white,->] (aq3) to node[text=white,above] {$a$} (ar3);

	  \node (h1) at (2.2,-1.8) {$q_1 \in Q, (q_2,r_2,q_3,r_3) \in S$ and $z \in A^+$};

	  \begin{scope}[xshift=0.2cm]
		\node[anchor=east] (h6) at (1.0,-3.0) {such that};
        
        \node (lq1) at (1.3,-3.05) {$q_1$};
        \draw[thick,->] (lq1) to [loop above] node[above] {$z$} (lq1);
        \node (lq2) at (2.0,-3.05) {$q_2$};
        \draw[thick,->] (lq2) to [loop above] node[above] {$z$} (lq2);
        \node (lq3) at (2.7,-3.05) {$r_2$};
        \draw[thick,->] (lq3) to [loop above] node[above] {$z$} (lq3);
        \node (lr2) at (3.4,-3.05) {$q_3$};
        \draw[thick,->] (lr2) to [loop above] node[above] {$z$} (lr2);
        \node (lr3) at (4.1,-3.05) {$r_3$};
        \draw[thick,->] (lr3) to [loop above] node[above] {$z$} (lr3);
	  \end{scope}

	  \node[fit=(as1)(ar3)(al1)(lr3),draw,thick,rounded corners=5pt,label=below:Single state and quadruple in $S$] (box3) {};

	  \node (e) at (7.5,-1.8) {$(q_1,q_2,q_3)$};
	  \node (f) at (10.0,-1.8) {$(q_1,r_2,r_3)$};
	  \draw[thick,->] (e) to node[above] (l3) {$\veps$} (f);
	  \node[fit=(l3)(e)(f),draw,thick,rounded corners=5pt,label=below:Produced \veps-transition in $\Bs^+_S$] (box4) {};
	  \draw[line width=1pt,double,->] ($(box3.north east)!(box4.west)!(box3.south east)$) to (box4);

    \end{tikzpicture}
    \caption{Construction of the transitions in the auxiliary automaton $\Bs^+_S$}
    \label{fig:pgauto}
\end{figure}

We are ready to define $\betg(S) \subseteq Q^4$. For every $(q,r,s,t) \in Q^4$, we let $(q,r,s,t) \in \betg(S)$ if and only if the two following conditions hold:
\begin{equation} \label{eq:gpauto}
    \begin{array}{l}
        \text{$\{\veps\}$ is \emph{not} \Gs-separable from \lauto{\Bs^+_S}{(s,q,s)}{(t,r,t)}, and}\\
        \text{$\{\veps\}$ is \emph{not} \Gs-separable from \lauto{\Bs^+_S}{(q,s,q)}{(r,t,r)}}.
    \end{array}
\end{equation}
A set $S\subseteq Q^4$ is \emph{$(\bpoln,+)$-sound} for \Gs and $\As$ if it is a fixpoint for \betg, \emph{i.e.} $\betg(S) = S$. The following monotonicity lemma implies that there is a \emph{greatest $(\bpoln,+)$-sound set} (see Appendix~\ref{app:bpolgp}).

\begin{restatable}{lemma}{pginc} \label{lem:pginc}
    Let \Gs be a group \vari and $\As = (Q,\delta)$ an \nfa. For every $S,S' \subseteq Q^4$, we have $S \subseteq S' \Rightarrow \betg(S) \subseteq \betg(S')$.
\end{restatable} 

\medskip

We may now state our second key theorem. It applies to \bpol{\Gs^+}-separation.

\begin{restatable}{theorem}{pgauto}\label{thm:pgauto}
    Let \Gs be a group \vari and $\As = (Q,\delta)$ an \nfa. Then, \gnsep{\bpol{\Gs^+}}{\As} is the greatest $(\bpoln,+)$-sound subset of $Q^4$ for \Gs and \As.
\end{restatable}

\medskip

Let us discuss the consequences of  Theorems~\ref{thm:gauto} and~\ref{thm:pgauto}. Since $\Bs_S$ and $\Bs^+_S$ can be computed from \As and $S$, one can compute $\tautg(S)$ and $\betg(S)$ from $S$ provided that \Gs-separation is decidable. Hence, if \Gs-separation is decidable, Theorem~\ref{thm:gauto} (resp.~Theorem~\ref{thm:pgauto}) yields a \emph{greatest} fixpoint procedure for computing \gnsbg (resp.~\gnsep{\bpol{\Gs^+}}{\As}). Indeed, consider the sequence of subsets defined by $S_0 = Q^4$, and $S_{n} = \tautg(S_{n-1})$ for $n \geq 1$. By definition, computing $S_n$ from $S_{n-1}$ boils down to deciding \Gs-separation. Since \tautg is increasing by Lemma~\ref{lem:ginc}, we get a decreasing sequence $Q^4 = S_0 \supseteq S_1 \supseteq S_2 \cdots$. Moreover, since $Q^4$ is finite, this sequence stabilizes at some point: there exists $n \in \nat$ such that $S_n = S_{j}$ for all $j \geq n$. One may verify that $S_n$ is the greatest $(\bpoln,*)$-sound subset of $Q^4$. By Theorem~\ref{thm:gauto}, it follows that $S_n = \gnsbg$. Likewise, the sequence $T_n$ defined by $T_0=Q^4$ and $T_{n} = \betg(T_{n-1})$ is computable when \Gs-separation is decidable, and, since it is decreasing, it stabilizes. By Theorem~\ref{thm:pgauto}, its stabilization value is~\gnsep{\bpol{\Gs^+}}{\As}.

By Proposition~\ref{prop:compgnsep}, \nsep{\bpol{\Gs}}\As  (resp. \nsep{\bpol{\Gs^+}}{\As}) can be computed from \gnsbg (resp.~\gnsbgp) via a \emph{least} fixpoint procedure. Altogether, by Proposition~\ref{prop:autosep}, we get reductions from \bpol{\Gs}- and \bpol{\Gs^+}-separation to \Gs-separation. One may verify that these are polynomial time reductions (we mean ``reduction'' in the Turing sense: \bpol{\Gs}- and \bpol{\Gs^+}-separation can be decided in polynomial time using an oracle for \Gs-separation).

Now, it is known that separation can be decided in polynomial time  for the classes \stzer, \md and \grp (this is trivial for \stzer, see~\cite{pzgroup} for \md and \grp). Hence, we obtain from Theorem~\ref{thm:gauto} that separation is decidable in polynomial time for \bpol{\stzer} (\emph{i.e.}, \bsw{1}), \bpol{\md} (\emph{i.e.}, \bswm{1}) and \bpol{\grp}. This was well-know for \bpol{\stzer} (the class of piecewise testable languages, see~\cite{cmmptsep,pvzptsep}). For the other two, decidability was known~\cite{Zetzsche18,pzconcagroup} but not the polynomial time upper bound. Using Theorem~\ref{thm:pgauto}, we also obtain that separation is decidable in polynomial time for \bpol{\stzer^+} (\emph{i.e.}, the languages of dot-depth one or equivalently \bsws{1}), \bpol{\md^+} (\emph{i.e.}, \bswsm{1}) and \bpol{\grp^+}. Decidability was already known for \bpol{\stzer^+} and \bpol{\md^+}: the results can be obtained indirectly by reduction to \bpol{\stzer}-separation using transfer theorems~\cite{PlaceZ20,prwmodulo}. Yet, the polynomial time upper bounds are new as the transfer theorems have a built-in exponential blow-up. Moreover, decidability of separation is a new result for \bpol{\grp^+}.

Finally, the statement applies to \bpol{\abg} and \bpol{\abg^+} (\emph{i.e.}, \bswam{1} and \bswsam{1}). This is a new result for \bpol{\abg^+}. Yet, since \abg-separation is \conptime-complete when the alphabet is part of the input~\cite{pzgroup} (the problem being in \ptime for a fixed alphabet), the complexity analysis is not entirely immediate. However, one may verify that the procedures yield \conptime algorithms
for both \bpol{\abg}- and \bpol{\abg^+}-separation. We summarize the upper bounds in Figure~\ref{fig:bpol:gcomp}.

\begin{figure}[!htb]
    \begin{center}
        \begin{tikzpicture}
            \matrix (M) [matrix of nodes, column  sep=5mm,row  sep=4mm,draw,very thick,rounded corners=3pt,nodes={align=center,text width = 1cm,anchor=center}]
            {
                \node[text width = 6cm] (t1) {Input class \Gs}; & \stzer    & \md  & \abg &  \grp \\
                \node[text width = 6cm] (t2) {\bpol{\Gs}- and \bpol{\Gs^+}-separation}; &   \ptime & \ptime  & \conptime & \ptime\\
            };
            
            \foreach \col in {3,4,5} {%
                \mvline[thick]{M}{\col}
            }
            
            \node (L) at ($(t2.east)!0.5!(M-2-2.west)$) {};

            \draw[thick] (L |- M.north) -- (L |- M.south);
            \node (K) at ($(t1)!0.5!(t2)$) {};

            \draw[thick] (K -| M.west) -- (K -| M.east);

        \end{tikzpicture}
    \end{center}
    \caption{Complexity of separation (for input languages represented by \nfas).}
    \label{fig:bpol:gcomp}
\end{figure}

\subsection{Proof of Theorem~\ref{thm:gauto}}

We now concentrate on the proof of Theorem~\ref{thm:gauto}. The key ingredients in this argument are Proposition~\ref{prop:pgcov} and Lemma~\ref{lem:pgsound}. On the other hand, the proof of Theorem~\ref{thm:pgauto} is postponed to Appendix~\ref{app:bpolgp}. It is based on similar ideas. Roughly, we replace Proposition~\ref{prop:pgcov} and Lemma~\ref{lem:pgsound} (which are tailored to classes \bpol{\Gs}) by their counterparts for \bpol{\Gs^+}: Proposition~\ref{prop:pgpcov} and Lemma~\ref{lem:pgpsound}. However, note that proving Theorem~\ref{thm:pgauto} is technically more involved as manipulating the automaton $\Bs_{S}^+$ in the definition of \betg requires more work.

We fix a group \vari \Gs and an \nfa $\As = (Q,\delta)$. Let $S \subseteq Q^4$ be the greatest $(\bpoln,*)$-sound subset for \Gs and \As. We prove that $S = \gnsbg$.

\smallskip
\noindent
{\bf First part: $S\subseteq\gnsbg$.} We use \emph{tuple separation} and Lemma~\ref{lem:pgsound}. Let us start with some terminology. For every $n \geq 1$ and $(q_1,r_1,q_2,r_2) \in Q^4$, we associate an $n$-tuple of languages, written $T_n(q_1,r_1,q_2,r_2)$. We use induction on $n$ and tuple concatenation to present the definition. If $n = 1$ then, $T_1(q_1,r_1,q_2,r_2) = \big(\alauto{q_2}{r_2}\big)$. If $n > 1$, then,
\[
T_n(q_1,r_1,q_2,r_2) = \left\{\begin{array}{ll}
    (\alauto{q_2}{r_2}) \cdot T_{n-1}(q_1,r_1,q_2,r_2) & \text{if $n$ is odd} \\
    (\alauto{q_1}{r_1}) \cdot T_{n-1}(q_1,r_1,q_2,r_2) & \text{if $n$ is even.}
\end{array}\right.
\]
For example, we have $T_3(q_1,r_1,q_2,r_2) = (\alauto{q_2}{r_2},\alauto{q_1}{r_1},\alauto{q_2}{r_2})$.

\begin{restatable}{proposition}{gsound}\label{prop:gsound}
    For every $n \geq 1$ and $(q_1,r_1,q_2,r_2) \in S$, the $n$-tuple $T_n(q_1,r_1,q_2,r_2)$ is not \pol{\Gs}-separable under \Gs-control. 
\end{restatable}

By definition, Proposition~\ref{prop:gsound} implies that for all $p \geq 1$ and $(q_1,r_1,q_2,r_2) \in S$, the $2p$-tuple $(\alauto{q_1}{r_1},\alauto{q_2}{r_2})^p$ is not \pol{\Gs}-separable under \Gs-control. By Corollary~\ref{cor:covtsep}, it follows that $\alauto{q_1}{r_1}$ is not \bpol{\Gs}-separable from $\alauto{q_2}{r_2}$ under \Gs-control, \emph{i.e.}, that $(q_1,r_1,q_2,r_2)  \in \gnsbg$. We get $S \subseteq \gnsbg$ as desired.

We prove Proposition~\ref{prop:gsound} by induction on $n$. We fix $n \geq 1$ for the proof. In order to exploit the hypothesis that $S$ is $(\bpoln,*)$-sound, we need a property of the \nfa $\Bs_S = (Q^3,\gamma_S)$ used to define \tautg. When $n \geq 2$, this is where we use induction on $n$ and Lemma~\ref{lem:pgsound}.

\begin{restatable}{lemma}{gsinduc}\label{lem:gsinduc}
    Let $(s_1,s_2,s_3),(t_1,t_2,t_3) \in Q^3$ and $w \in \lauto{\Bs_{S}}{(s_1,s_2,s_3)}{(t_1,t_2,t_3)}$. Then, $w \in \alauto{s_1}{t_1}$ and, if $n \geq 2$, the $n$-tuple $(\{w\})  \cdot T_{n-1}(s_{2},t_{2},s_{3},t_{3})$ is not \pol{\Gs}-separable.
\end{restatable}

\begin{proof}
    Since $w \in \lauto{\Bs_{S}}{(s_1,s_2,s_3)}{(t_1,t_2,t_3)}$, there exists a run labeled by $w$ from $(s_1,s_2,s_3)$ to $(t_1,t_2,t_3)$ in $\Bs_S$. We use a sub-induction on the number of transitions involved in that run. First, assume that no transitions are used: we have $w = \veps$ and $(s_1,s_2,s_3) = (t_1,t_2,t_3)$. Clearly, $\veps \in \alauto{s_1}{s_1}$ and, if $n \geq 2$, the $n$-tuple $(\{\veps\})  \cdot T_{n-1}(s_{2},s_{2},s_{3},s_{3})$ is not \pol{\Gs}-separable by Lemma~\ref{lem:tuptriv} since $\veps \in \alauto{s_2}{s_2} \cap \alauto{s_3}{s_3}$. We now assume that at least one transition is used and consider the last one: we have $(q_1,q_2,q_3) \in Q^3$, $w' \in A^*$ and $x \in A\cup \{\veps\}$ such that $w = w'x$, $w' \in \lauto{\Bs_{S}}{(s_1,s_2,s_3)}{(q_1,q_2,q_3)}$ and $((q_1,q_2,q_3),x,(t_1,t_2,t_3)) \in \gamma_S$. By induction, we have $w' \in \alauto{s_1}{q_1}$ and, if $n \geq 2$, the $n$-tuple $(\{w'\})  \cdot T_{n-1}(s_{2},q_{2},s_{3},q_{3})$ is not \pol{\Gs}-separable. We prove that $x \in \alauto{q_1}{t_1}$ and, if $n \geq 2$, the $n$-tuple $(\{x\})  \cdot T_{n-1}(q_{2},t_{2},q_{3},t_{3})$ is not \pol{\Gs}-separable. It will then be immediate that $w = w'x  \in \alauto{s_1}{t_1}$ and, if $n \geq 2$, Lemma~\ref{lem:tupconcat} implies that $(\{w\})  \cdot T_{n-1}(s_{2},t_{2},s_{3},t_{3})$ is not \pol{\Gs}-separable.
    
    We consider two cases depending on whether $x \in A$ or $x = \veps$. First, if $x = a \in A$, then $(q_i,a,t_i) \in \delta$ for $i = \{1,2,3\}$. Clearly, this implies that $a \in \alauto{q_1}{t_1}$ and, if  $n \geq 2$, then $(\{a\})  \cdot T_{n-1}(q_{2},t_{2},q_{3},t_{3})$ is not \pol{\Gs}-separable by Lemma~\ref{lem:tuptriv} since $a \in \alauto{q_2}{t_2} \cap \alauto{q_3}{t_3}$. Assume now that $x = \veps$: we are dealing with an \veps-transition. By definition of $\gamma_S$, we have $q_1 = t_1$ and $(q_2,t_2,q_3,t_3) \in S$. The former yields $\veps \in \alauto{q_1}{t_1}$. Moreover, if $n \geq 2$, since $(q_2,t_2,q_3,t_3) \in S$, it follows from induction on $n$ in Proposition~\ref{prop:gsound} that the $(n-1)$-tuple $T_{n-1}(q_2,t_2,q_3,t_3)$ is not \pol{\Gs}-separable under \Gs-control. Combined with Lemma~\ref{lem:pgsound}, this yields that $(\{\veps\})  \cdot T_{n-1}(q_2,t_2,q_3,t_3)$ is not \pol{\Gs}-separable, as desired.
\end{proof}

We may now complete the proof of Proposition~\ref{prop:gsound}. By symmetry, we only treat the case when $n$ is odd and leave the case when it is even to the reader. Let $(q_1,r_1,q_2,r_2) \in S$, we have to prove that $T_n(q_1,r_1,q_2,r_2)$ is not \pol{\Gs}-separable under \Gs-control. Hence, we fix $H \in \Gs$ such that $\veps \in H$ and prove $H \cap T_n(q_1,r_1,q_2,r_2)$ is not \pol{\Gs}-separable. Since $S$ is $(\bpoln,*)$-sound, we have $\tautg(S) = S$, which implies that $(q_1,r_1,q_2,r_2) \in \tautg(S)$. Hence, it follows from~\eqref{eq:gauto} that $\{\veps\}$ is not \Gs-separable from $\lauto{\Bs_{S}}{(q_2,q_1,q_2)}{(r_2,r_1,r_2)}$. Since $H \in \Gs$ and $\veps \in H$, we get a word $w \in H \cap \lauto{\Bs_{S}}{(q_2,q_1,q_2)}{(r_2,r_1,r_2)}$. By Lemma~\ref{lem:gsinduc}, we have $w \in  H \cap \alauto{q_2}{r_2}$. This completes the proof when $n = 1$. Indeed, in that case we have $T_1(q_1,r_1,q_2,r_2) = (\alauto{q_2}{r_2})$ and since $H \cap \alauto{q_2}{r_2} \neq \emptyset$, it follows that $H \cap T_1(q_1,r_1,q_2,r_2)$ is not \pol{\Gs}-separable, as desired. If $n \geq 2$, then Lemma~\ref{lem:gsinduc} also implies that $(\{w\})  \cdot T_{n-1}(q_{1},r_{1},q_{2},r_{2})$ is not \pol{\Gs}-separable. Since $w \in H \cap \alauto{q_2}{r_2}$, Lemma~\ref{lem:tuptriv} yields that $(H \cap \alauto{q_2}{r_2}) \cdot T_{n-1}(q_1,r_1,q_2,r_2)$ is not \pol{\Gs}-separable. Thus, since $H \in \Gs \subseteq \pol{\Gs}$, one may verify that the $n$-tuple $(H \cap \alauto{q_2}{r_2}) \cdot (H \cap T_{n-1}(q_1,r_1,q_2,r_2))$ is not \pol{\Gs}-separable. By definition, this exactly says that $H \cap T_n(q_1,r_1,q_2,r_2)$ is not \pol{\Gs}-separable, completing the proof.

\medskip
\noindent
{\bf Second part: $\gnsbg \subseteq S$.} In the sequel, we say that an arbitrary set $R \subseteq Q^4$ is \emph{good} if there exists $L \in \Gs$ such $\veps \in L$ and a \bpol{\Gs}-cover \Kb of $L$ which is separating for $R$. 

\begin{restatable}{proposition}{bgcomp}\label{prop:bgcomp}
    Let $R \subseteq Q^4$. If $R$ is good, then $\tautg(R)$ is good as well.
\end{restatable}

We use Proposition~\ref{prop:bgcomp} to complete the proof. Let $S_0 = Q^4$ and $S_i = \tautg(S_{i-1})$ for $i \geq 1$. By Lemma~\ref{lem:ginc}, we have $S_0 \supseteq S_1 \subseteq S_2 \supseteq \cdots$ and there is $n \in \nat$ such that $S_n$ is the greatest $(\bpoln,*)$-sound subset for \Gs and \As, \emph{i.e.}, such that $S_n = S$. Since $S_0$ is good ($\{A^*\}$ is a \bpol{\Gs}-cover of $A^* \in \Gs$ which is separating for $S_0 = Q^4$), Proposition~\ref{prop:bgcomp} implies that $S_i$ is good for all $i \in \nat$.  Thus, $S = S_n$ is good. We get  $L \in \Gs$ such that $\veps \in L$ and a \bpol{\Gs}-cover \Kb of $L$ which is separating for $S$. Lemma~\ref{lem:gautosep} then yields $\gnsbg \subseteq S$ as desired.

\begin{remark}
    The proof of Proposition~\ref{prop:bgcomp} actually provides a construction for building $L \in \Gs$ such that $\veps \in L$ and a \bpol{\Gs}-cover \Kb of $L$ which is separating for $S$ (yet, this involves building separators in \Gs, see Lemma~\ref{lem:gbl}). As we have now established that $S = \gnsbg$, one may then follow the proof of Proposition~\ref{prop:compgnsep} to build a \bpol{\Gs}-cover $\Hb$ of $A^*$ which is separating for \nsep{\bpol{\Gs}}{\As}. Finally, \Hb encodes separators for all pairs of languages recognized by \As which are \bpol{\Gs}-separable (this is the proof of Lemma~\ref{lem:autosep} presented in Appendix~\ref{app:separ}). Altogether, we get a way to build separators in \bpol{\Gs}, when they exist.
\end{remark}

\medskip

We now prove Proposition~\ref{prop:bgcomp}. Let $R \subseteq Q^4$ be good. We have to build $L \in \Gs$ with $\veps \in L$ and a \bpol{\Gs}-cover \Kb of $L$ which is separating for $\tautg(R)$ (which will prove that $\tautg(R)$ is good as well). We first build $L$ (this part is independent from our hypothesis on~$R$).


\begin{restatable}{lemma}{gbl}\label{lem:gbl}
    There exists $L \in \Gs$ such that $\veps \in L$ and for every $(q,r,s,t) \in Q^4$, if $\lauto{\Bs_R}{(q,s,q)}{(r,t,r)} \cap L \neq \emptyset$ and $\lauto{\Bs_R}{(s,q,s)}{(t,r,t)} \cap L \neq \emptyset$, then $(q,r,s,t) \in \tautg(R)$.
\end{restatable}

\begin{proof}
    Let \Hb be the \emph{finite} set of all languages recognized by $\Bs_R$ such that $\{\veps\}$ is \Gs-separable from $H$. For every $H \in \Hb$, there exists $L_H \in \Gs$ such that $\veps \in L_H$ and $L_H \cap H = \emptyset$. We define $L = \bigcap_{H \in \Hb} L_H \in \Gs$. It is clear that $\veps \in L$. Moreover, given $(q,r,s,t) \in Q^4$, if $\lauto{\Bs_R}{(q,s,q)}{(r,t,r)} \cap L \neq \emptyset$ and $\lauto{\Bs_R}{(s,q,s)}{(t,r,t)} \cap L \neq \emptyset$, it follows from the definition of $L$ that $\{\veps\}$ is not \Gs-separable from both \lauto{\Bs_R}{(q,s,q)}{(r,t,r)} and \lauto{\Bs_R}{(s,q,s)}{(t,r,t)}. It follows from~\eqref{eq:gauto} in the definition of \tautg that $(q,r,s,t) \in \tautg(R)$.
\end{proof}

We fix $L \in \Gs$ as described in Lemma~\ref{lem:gbl} for the remainder of the proof. We now build the \bpol{\Gs}-cover \Kb of $L$ using the hypothesis that $R$ is good and Proposition~\ref{prop:pgcov}.

\begin{restatable}{lemma}{grun}\label{lem:grun}
    For all $(q,r) \in Q^2$, there is $H_{q,r} \in \bpol{\Gs}$ such that $\alauto{q}{r}\cap L\subseteq H_{q,r}$ and for all pairs $(s,t) \in Q^2$, if $\alauto{s}{t} \cap H_{q,r} \neq \emptyset$ then $\lauto{\Bs_R}{(q,s,q)}{(r,t,r)} \cap L \neq \emptyset$.
\end{restatable}

\begin{proof}
    Since $R$ is good, there are $U \in \Gs$ such that $\veps\in U$ and a \bpol{\Gs}-cover \Vb of $U$ which is separating for $R$. We use them to build $H_{q,r}$. Since $U$ is a group language and $\veps \in U$, Proposition~\ref{prop:pgcov} yields a cover \Pb of $\alauto{q}{r} \cap L$ such that every $P \in \Pb$ is of the form $P = \uclos_U w_P$ where $w_P \in \alauto{q}{r} \cap L$. For every $P \in \Pb$, we build a \bpol{\Gs}-cover $\Kb_P$ of $P$. Let $a_1,\dots,a_n \in A$ be the letters such that $w_P = a_1 \cdots a_n$. We have $P = Ua_1U \cdots a_nU$. Since $U \in \Gs \subseteq \pol{\Gs}$ and \Vb is a \bpol{\Gs}-cover of $U$, Proposition~\ref{prop:bconcat} yields a \bpol{\Gs}-cover $\Kb_P$ of $P$ such that for every $K \in \Kb_P$, there exist $V_0,\dots,V_n \in \Vb$ satisfying $K \subseteq V_0a_1V_1 \cdots a_nV_n$. We define $H_{q,r}$ as the union of all languages $K$ such that $K \in \Kb_P$ for some $P \in \Pb$ and $\alauto{q}{r} \cap K \neq \emptyset$. Clearly, $H_{q,r} \in \bpol{\Gs}$. Moreover, since \Pb is a cover of $\alauto{q}{r} \cap L$, and $\Kb_P$ is a cover of $P$ for each $P \in \Pb$, it is clear that $\alauto{q}{r}\cap L\subseteq H_{q,r}$. We now fix $(s,t) \in Q^2$ such that $\alauto{s}{t} \cap H_{q,r} \neq \emptyset$ and show that $\lauto{\Bs_R}{(q,s,q)}{(r,t,r)} \cap L \neq \emptyset$. By definition of $H_{q,r}$, we get $P \in \Pb$ and $K \in \Kb_P$ such that $\alauto{q}{r} \cap K\neq \emptyset$ and $\alauto{s}{t}\cap K \neq \emptyset$. By definition, $P = \uclos_U w_P$ with $w_P \in \alauto{q}{r} \cap L$. Hence, it suffices to prove that $w_P \in \lauto{\Bs_R}{(q,s,q)}{(r,t,r)}$.

    We fix $x \in \alauto{s}{t}\cap K$ and $y \in \alauto{q}{r} \cap K$. 
	Recall that $w_P = a_1 \cdots a_n$ (if $n = 0$, then $w_P = \veps$). Since $w_P \in \alauto{q}{r}$, we may consider the corresponding run in \As: we get $p_0,\dots,p_n \in Q$ such that $p_0 = q$, $p_n = r$ and $(p_{i-1},a_i,p_{i}) \in \delta$ for $1 \leq i \leq n$. Moreover, since $K \in \Kb_P$ and $w_P = a_1 \cdots a_n$, we have $K \subseteq  V_0a_1V_1 \cdots a_nV_n$ for $V_0, \dots,V_n \in \Vb$ (if $n = 0$, then $K \subseteq V_0$). Since $x,y \in K$, we get $x_i,y_i \in V_i$ for $0 \leq i \leq n$ such that $x = x_0a_1x_1 \cdots a_n x_n$ and $y = y_0a_1y_1 \cdots a_n y_n$. Since $x \in \alauto{s}{t}$, we get $s_0,t_0,\dots,s_n,t_n \in Q$ such that $s_0 = s$, $t_n = t$, $x_i \in \alauto{s_i}{t_i}$ for $0 \leq i \leq n$, and $(t_{i-1},a_i,s_i) \in \delta$ for $1 \leq i \leq n$. Symmetrically, since $y \in \alauto{q}{r}$, we get $q_0,r_0,\dots,q_n,r_n \in Q$ such that $q_0 = q$, $r_n = r$, $y_i \in \alauto{q_i}{r_i}$ for $0 \leq i \leq n$, and $(r_{i-1},a_i,q_i) \in \delta$ for $1 \leq i \leq n$. By definition of $\gamma_R$, it is immediate that $((p_{i-1},t_{i-1},r_{i-1}),a_i,(p_{i},s_i,q_i)) \in \gamma_R$ for $1 \leq i \leq n$. Since $V_i \in \Vb$ and $\Vb$ is separating for $R$, the fact that $x_i,y_i \in V_i$ implies that $(s_i,t_i,q_i,r_i) \in R$ for $0 \leq i \leq n$. Hence, $((p_{i},s_{i},q_{i}),\veps,(p_{i},t_i,r_i)) \in \gamma_R$ by definition. Thus, we get a run labeled by $w_P$ from $(p_0,s_0,q_0)$ to $(p_n,t_n,r_n)$ in $\Bs_R$, \emph{i.e.}, $w_P \in \lauto{\Bs_R}{(q,s,q)}{(r,t,r)}$ as desired.
\end{proof}

We may now build \Kb. Let $\Hb = \big\{H_{q,r} \mid (q,r) \in Q^2\big\}$. Consider the following equivalence $\sim$ defined on $L$: given $u,v\in L$, we let $u \sim v$ if and only if $u \in H_{q,r} \Leftrightarrow v \in H_{q,r}$ for every $(q,r) \in Q^2$. We let \Kb as the partition of $L$ into $\sim$-classes. Clearly, each $K \in \Kb$ is a Boolean combination involving the languages in \Hb (which belong to \bpol{\Gs}) and $L \in \Gs$. Hence, $\Kb$ is a \bpol{\Gs}-cover of $L$. We now prove that it is separating for $\tautg(R)$. Let $q,r,s,t \in Q$ and $K \in \Kb$ such that there are $u \in \alauto{q}{r} \cap K$ and $v \in \alauto{s}{t} \cap K$. We show that $(q,r,s,t) \in \tautg(R)$. By definition of \Kb, we have $u,v \in L$ and $u \sim v$. In particular, $u \in \alauto{q}{r} \cap L$ which yields $u \in H_{q,r}$ by definition in Lemma~\ref{lem:grun}. Together with $u \sim v$, this yields $v \in H_{q,r}$. Hence, $\alauto{s}{t} \cap H_{q,r} \neq \emptyset$ and Lemma~\ref{lem:grun} yields $\lauto{\Bs_R}{(q,s,q)}{(r,t,r)} \cap L \neq \emptyset$. One may now use a symmetrical argument to obtain $\lauto{\Bs_R}{(s,q,s)}{(t,r,t)} \cap L \neq \emptyset$. By definition of $L$ in Lemma~\ref{lem:gbl}, this yields $(q,r,s,t) \in \tautg(R)$, completing the proof.

\section{Conclusion}
\label{sec:conc}
In this paper, we proved that for every group \vari \Gs, there exist generic polynomial time Turing reductions from \bpol{\Gs}- and \bpol{\Gs^+}-separation to \Gs-separation, for input languages represented by \nfas. While a generic reduction from \bpol{\Gs}-separation to \Gs-separation was already developed in~\cite{pzconcagroup}, it relied on an involved machinery, which required to dig into a more general problem than \bpol{\Gs}-separation, namely ``\bpol{\Gs}-covering''. In particular, the techniques from~\cite{pzconcagroup} do not provide any way to~build separators in \bpol{\Gs} (when they exist). They also yield poor upper complexity bounds. At last, the results of~\cite{pzconcagroup} do not apply to \bpol{\Gs^+}. In this case, even the existence of a generic reduction is new. It would be interesting to unify ideas of the present paper with the techniques of~\cite{pzconcagroup}, to  lift them to the setting of \bpol{\Gs}- and \bpol{\Gs^+}-covering. We leave this for further~work.

Our results imply that separation is decidable in \emph{polynomial time} for a number of standard classes: the piecewise testable languages (\emph{i.e.}, \bpol{\stzer} or equivalently \bsw{1}), the languages of dot-depth one (\emph{i.e.}, \bpol{\stzer^+} or equivalently \bsws{1}), the classes \bpol{\md} and \bpol{\md^+} (\emph{i.e.}, $\bsc{1}(<,MOD)$ and $\bsc{1}(<,+1,MOD)$) and the classes \bpol{\grp} and \bpol{\grp^+}. While this was well-known for the piecewise testable languages~\cite{pvzptsep,cmmptsep}, all other results are new---not only regarding the complexity, but even regarding the decidability. Actually, it is shown in~\cite{Masopust18} that \bpol{\stzer}-separation is \ptime-complete. It turns out that the reduction of~\cite{Masopust18}, from the circuit value problem, adapts to prove the \ptime-completeness of separation for all of the above classes (we leave the details for further work). Finally, our results also apply to the classes \bpol{\abg} and \bpol{\abg^+}  (\emph{i.e.}, $\bsc{1}(<,AMOD)$ and $\bsc{1}(<,+1,AMOD)$): we obtain that separation is in \conptime. While this is currently unknown, we conjecture that this is a \emph{tight} upper bound. Indeed, it is known that \abg-separation is \conptime-complete~\cite{pzgroup}.



\bibliography{main}

\appendix

\section{Appendix to Section~\ref{sec:prelims}}
\label{app:prelims}
We start with the proof of Proposition~\ref{prop:bconcat}. Let us first recall the statement.

\bconcat*

\begin{proof}
	We first handle the case when $n = 1$ (\emph{i.e.}, there are two languages $L_0,L_1 \in \pol{\Cs}$) and then lift the result to the general case using a simple induction. For the sake of avoiding clutter, we write $\Ds = \pol{\Cs}$ in the proof.

	\smallskip
	\noindent
	{\bf Case $n = 1$.} Consider two languages $L_0,L_1\in \Ds$ and $a \in A$. Moreover, let $\Hb_0$ and $\Hb_1$ be \bool{\Ds}-covers of $L_0$ and $L_1$. We need to build an appropriate \bool{\Ds}-cover \Kb of $L_0aL_1$. By hypothesis, every language in $\Hb_0 \cup \Hb_1$ is a Boolean combination of languages in \Ds. Moreover, $L_0,L_1\in \Ds$. Hence, there exists a \emph{finite} set of languages $\Ub \subseteq \Ds$ containing $L_0,L_1$ and  such that every language $H \in \Hb_0 \cup \Hb_1$ is a Boolean combination of languages in \Ub. We define \Vb as the set containing all finite intersections of languages in \Ub. Clearly, \Vb remains finite and since $\Ds = \pol{\Cs}$ is a lattice by Theorem~\ref{thm:bpolvar}, we have $\Ub \subseteq \Vb \subseteq \Ds$. We let $\Pb = \{V_0aV_1 \mid V_0,V_1 \in \Vb\}$. We have $L_0aL_1 \in \Pb$ by definition and $\Pb \subseteq \Ds$ since $\Ds = \pol{\Cs}$ is closed under marked product . We now use \Pb to define an equivalence on~$A^*$. Given $w,w' \in A^*$, we write $w\sim w'$ if and only if $w\in P\Leftrightarrow w' \in P$ for every $P \in \Pb$. Since \Pb is finite, we know that $\sim$ has finite index. Moreover, by definition, every $\sim$-class is a Boolean combination of languages in $\Pb \subseteq\Ds$,  which means that it belongs to \bool{\Ds}.
	
	Since $L_0aL_1 \in \Pb$, the definition implies that $L_0aL_1$ is a finite union of $\sim$-classes. We define \Kb as the set containing all $\sim$-classes in this union. This is a \bool{\Ds}-cover of $L_0aL_1$ by definition. It remains to prove that for every $K \in \Kb$, there exist $H_0 \in \Hb_0$ and $H_1 \in \Hb_1$ such that $K \subseteq H_0aH_1$. We fix $K \in \Kb$ for the proof and use the following lemma.
	
	\begin{lemma} \label{lem:hintro:boolc}
		Let $G \subseteq K$ be a \emph{finite} language. There exists $H_0 \in \Hb_0$ and $H_1 \in \Hb_1$ such that $G \subseteq H_0aH_1$.
	\end{lemma}
	
	We first apply Lemma~\ref{lem:hintro:boolc} to complete the main argument. For each $n \in \nat$, we let $G_n \subseteq K$ as the (finite) language containing all words in $K$ of length at most $n$. Clearly, we have,
	\[
	K = \bigcup_{n \in \nat} G_n \quad \text{and} \quad G_n \subseteq G_{n+1} \text{ for all $n \in \nat$}.
	\]
	For every $n \in \nat$, Lemma~\ref{lem:hintro:boolc} yields $H_{0,n} \in \Hb_0$ and $H_{1,n} \in \Hb_1$ such that $G_n \subseteq H_{0,n}aH_{1,n}$. Since $\Hb_0$ and $\Hb_1$ are finite sets, there exist $H_0\in\Hb_0$ and $H_1\in\Hb_1$ such that $H_{0,n}=H_0$ and $H_{1,n}=H_1$ for infinitely many $n$. Since $G_n \subseteq G_{n+1}$ for every $n \in \nat$, it then follows that $G_n \subseteq H_0aH_1$ for every $n \in \nat$. Finally, since $K = \bigcup_{n \in \nat} G_n$, this implies $K \subseteq H_0aH_1$, completing the main proof.
	
	\smallskip
	
	We turn to the proof of Lemma~\ref{lem:hintro:boolc}. We fix a finite language $G \subseteq K$ for the proof. We exhibit $H_0 \in \Hb_0$ and $H_1 \in \Hb_1$ such that $G \subseteq H_0aH_1$. Let $w_1,\dots,w_n \in A^*$ be the words contained in $G$, \emph{i.e.} $G=\{w_1,\dots,w_n\}$. By definition, we know that $K$ is a $\sim$-class included in $L_0aL_1$. Consequently, we have $w_1,\dots,w_n \in L_0aL_1$ and $w_1 \sim \cdots \sim w_n$. We use the latter property to prove an intermediary fact. Given two words $w,w' \in A^*$, we write $w \preceq w'$ if and only if $w \in V \Rightarrow w' \in V$ for every $V \in \Vb$. Clearly, ``$\preceq$'' is a preorder.
	
	\begin{claim} \label{fct:hintro:bcf}
		For every $u,v \in A^*$ such that $w_n = uav$, there exist $u_1,\dots,u_n,v_1,\dots,v_n \in A^*$ such that $w_i = u_iav_i$ for every $i \leq n$, $u \preceq u_1 \preceq \cdots \preceq u_n$ and $v \preceq v_1 \preceq \cdots \preceq v_n$.
	\end{claim}
	
	\begin{proof}
		We prove the existence of $u_1,v_1 \in A^*$ such that $w_1\! =\! u_1av_1$, $u \preceq u_1$ and $v \preceq v_1$ using the fact that $w_n = uav$ and $w_n \sim w_1$, one may then iterate the argument to build $u_2,\dots,u_n \in A^*$ and $v_2,\dots,v_n \in A^*$. Consider the languages $L_u = \bigcap_{\{V \in \Vb \mid u \in V\}} V$ and $L_v = \bigcap_{\{V \in \Vb \mid v \in V\}} V$. Since \Vb is finite and closed under intersection by definition, we have $L_u,L_v \in  \Vb$. Hence, $L_uL_v \in \Pb$ by definition of \Pb. Moreover, it is clear that $uav \in L_uaL_v$. Therefore, since $uav=w_n \sim w_1$, the definition of $\sim$ implies that $w_1 \in L_uaL_v$. This yields $u_1 \in L_u$ and $v_1 \in L_v$ such that $w_1 = u_1av_1$. Finally, the definitions of $L_u$ and $L_v$ imply that $u \preceq u_1$ and $v \preceq v_1$, completing the proof.
	\end{proof}
	
	Since $w_n \in L_0aL_1$, it can be decomposed as $w_n = uav$ with $u \in L_0$ and $v \in L_1$. Since $w_1$ is a finite word, it admits finitely many decompositions $w_1 = u_1av_1$ with $u_1,v_1 \in A^*$. Therefore, a repeated application of the claim together with the pigeon-hole principle yield $u_1,\dots,u_n,v_1,\dots,v_n \in A^*$ such that $w_i = u_iav_i$ for every $i \leq n$ and,
	\[
	u \preceq u_1 \preceq \cdots \preceq u_n \preceq u_1 \qquad \text{and} \qquad v \preceq v_1 \preceq \cdots \preceq v_n \preceq v_1.
	\]
	Since $u \in L_0$, $v \in L_1$ and $L_0,L_1 \in \Vb$ by definition of \Vb, we get that $u_1 \in L_0$ and $v_1 \in L_1$ by definition of $\preceq$. Therefore, since $\Hb_0$ and $\Hb_1$ are covers of $L_0$ and $L_1$ respectively, there exist $H_0 \in \Hb_0$ and $H_1 \in \Hb_1$ such that $u_1 \in H_0$ and $v_1 \in H_1$. Moreover, we have $u_1 \preceq u_i \preceq u_1$ and $v_1 \preceq v_i \preceq v_1$ for every $i \leq n$. By definition of $\preceq$, this implies that for every language $V \in \Vb$, we have $u_1\in V \Leftrightarrow u_i \in V$ and $v_1\in V\Leftrightarrow v_i \in V$. Since the languages in $\Hb_0 \cup \Hb_1$ are Boolean combinations of those in \Vb, it follows that $u_1\in H \Leftrightarrow u_i \in H$ and $v_1\in H \Leftrightarrow v_i \in H$ for all $H \in \Hb_0 \cup \Hb_1$ and $i \leq n$.  Hence, since $u_1 \in H_0$, $H_0 \in \Hb_0$, $v_1 \in H_1$ and $H_1 \in \Hb_1$, we obtain $u_1,\dots,u_n \in H_0$ and $v_1,\dots,v_n \in H_1$. Altogether, it follows that $G = \{w_1,\dots,w_n\} = \{u_1av_1,\dots,u_nav_n\} \subseteq H_0aH_1$. This concludes the proof of Lemma~\ref{lem:hintro:boolc}.
	
	\medskip
	\noindent
	{\bf General case.} We now use induction on $n \in \nat$ to prove the general case in Proposition~\ref{prop:bconcat}. We fix $L_0,\dots,L_n \in \Ds$,  $a_1,\dots,a_n \in A$ and $\Hb_i$ a \bpol{\Cs}-cover of $L_i$ for all $i \leq n$. We have to construct an appropriate \bool{\Ds}-cover \Kb of $L_0a_1L_1 \cdots a_nL_n$.
	
	The case $n = 0$ is trivial: it suffices to define $\Kb = \Hb_0$. Assume now that $n > 1$. By induction hypothesis, there exists a \bool{\Cs}-cover $\Kb'$ of $L_1a_2L_2 \cdots a_nL_n$ such that for every $K' \in \Kb'$, we have $H_i \in \Hb_i$ for $2 \leq i \leq n$ such that $K' \subseteq H_1a_2H_2 \cdots a_nH_n$. Since $\Ds = \pol{\Cs}$ is closed under marked product, we have $L_1a_2L_2 \cdots a_nL_n \in \Ds$. Hence, since we have a \bool{\Ds}-cover $\Hb_0$ of $L_0 \in \Ds$ and a \bool{\Cs}-cover $\Kb'$ of $L_1a_2L_2 \cdots a_nL_n \in \Ds$, we may use the case $n = 1$ in Proposition~\ref{prop:bconcat} (which we proved above) to get a \bool{\Ds}-cover $\Kb$ of $L_0a_1L_1 \cdots a_nL_n$ such that for every $K \in \Kb$, there exist $H_0 \in \Hb_0$ and $K' \in \Kb'$ which satisfy $K \subseteq H_0a_1K'$. By definition of $\Kb'$, we know that there also exist  $H_i \in \Hb_i$ for $2 \leq i \leq n$ such that $K' \subseteq H_1a_2H_2 \cdots a_nH_n$. Altogether, it follows that $K \subseteq H_0a_1H_1 \cdots a_nH_n$ which completes the proof.
\end{proof}

We now prove Proposition~\ref{prop:pgcov}. Let us first recall the statement.

\pgcov*

\begin{proof}
  Since $L$ is a group language, there exists a morphism $\eta: A^* \to G$ into a finite group~$G$ recognizing $L$. We let $L' = \eta\inv(1_G)$. Clearly, $L'$ is a group language and $\veps \in L'$. Moreover, since $\veps \in L$ and $L$ is recognized by $\eta$, we have $L' \subseteq L$.

  We use $L'$ to define an ordering ``$\preceq$'' on $A^*$. Consider two words $u,v \in A^*$, we write $u \preceq v$ when $v \in \uclos_{L'} u$. By definition of $L'$, it is straightforward to verify that for every $u,v \in A^*$, if $u \preceq v$, then $\eta(u) = \eta(v)$. Since $\veps \in L'$, it~is simple to verify that $\preceq$ is reflexive and antisymmetric. We prove that it is transitive. Let $u,v,w \in A^*$ such that $u \preceq v$ and $v \preceq w$. We show that $u \preceq w$. By definition, we have $v \in \uclos_{L'} u$. Hence, we get $a_1,\dots,a_n \in A$ and $x_0,\dots,x_n \in L'$ such that $u = a_1 \cdots a_n$ and $v = x_0a_1x_1 \cdots a_nx_n$. Since we also have $w \in \uclos_{L'} v$, one may verify that this yields $y_0,\dots,y_n \in A^*$ such that $w = y_0a_1y_1 \cdots a_ny_n$ and $y_i \in \uclos_{L'} x_i$ for every $i \leq n$. The latter property implies that $\eta(x_i) = \eta(y_i)$ for every $i \leq n$. Hence, since $x_0,\dots,x_n\in L' = \eta\inv(1_G)$, we get $y_0,\dots,y_n  \in L'$. We conclude that $w \in \uclos_{L'} u$ which exactly says that $u \preceq w$ as desired. The following lemma states that $\preceq$ is a ``well quasi-order''. A proof is available in~\cite[Proposition 3.10]{cano:hal-01247172}. Here we use a simple generalization of the proof of Higman's lemma.


  \begin{lemma}\label{lem:half:higman}
    Consider an infinite sequence $(u_i)_{i \in \nat}$ of words in $A^*$. There exist $i,j \in \nat$ such that $i<j$ and $u_i \preceq u_j$.
  \end{lemma}

  \begin{proof}
    We say that a sequence $(u_i)_{i \in \nat}$ is \emph{bad} if $u_i \not\preceq u_j$ for every $i < j$. We need to prove that there exists no bad sequence. We proceed by contradiction and assume that there exists a bad sequence. We first use this hypothesis to construct a specific one. Using induction, we build a particular sequence $(u_i)_{i \in \nat}$ such that for every $i \in \nat$, $u_0,\dots,u_i$ can be continued into a bad sequence and then verify that $(u_i)_{i \in \nat}$ is bad itself.

    We let $u_0$ be a word of minimal length such that $u_0$ can be continued into a bad sequence. Such a word must exist by the assumption that there exists a bad sequence. Assume now that $u_0,\dots,u_i$ have been defined up to some $i \in \nat$. By construction, $u_0,\dots,u_i$ can be continued into a bad sequence. We define $u_{i+1}$ as a word of minimal length such that $u_0,\dots,u_i,u_{i+1}$ can be continued into a bad sequence. This defines $(u_i)_{i \in \nat}$. Observe that it is necessarily bad. Indeed, otherwise, we would have $i<j$ such that $u_i \preceq u_j$ which contradicts the hypothesis that $u_0,\dots,u_j$ can be continued into a bad sequence.

    Consider the set $A^{|G|}$ consisting of all words of length $|G|$. Since $A^{|G|}$ is finite, there exists some word $w \in A^{|G|}$ such that $w$ is a prefix of infinitely many words in the sequence $(u_i)_{i \in \nat}$. We write $i_1 < i_2 < \cdots$ the infinitely many indices such that $w$ is a prefix of $u_{i_k}$, \emph{i.e.} $u_{i_k} = wv_{i_k}$ for some word $v_{i_k} \in A^*$. Since $|w| = |G|$, a pumping argument yields $x,z \in A^*$ and $y \in A^+$ such that $w = xyz$ and $\eta(x) = \eta(xy)$. Since $G$ is a group, it follows that $\eta(y) = 1_G$. We prove that the infinite sequence $u_0,u_1,\dots,u_{i_1-1},xzv_{i_1},xzv_{i_2},xzv_{i_3},\dots$ is bad. In particular, this means that $u_0,\dots,u_{i_1-1},xzv_{i_1}$ can be continued into a bad sequence. This is a contradiction: we have $u_{i_1} = xyzv_{i_1}$ and since $y \in A^+$, this implies that $|xzv_{i_1}| < |u_{i_1}|$. This is not possible since $u_{i_1}$ is defined as a word of minimal length such that $u_0,\dots,u_{i_1-1},u_{i_1}$ can be continued into a bad sequence.

    It remains to prove that $u_0,u_1,\dots,u_{i_1-1},xzv_{i_1},xzv_{i_2},xzv_{i_3},\dots$ is bad. Since $(u_i)_{i \in \nat}$ is bad itself, we already know that for $i < j \leq i_1-1$, we have $u_i \not\preceq u_j$. We now prove that $u_i \not\preceq xzv_{i_h}$ for $i \leq i_1-1$ and $h \geq 1$. By contradiction, if $u_i \preceq xzv_{i_h}$, then $u_i \preceq xyzv_{i_h} = u_{i_h}$ since $\eta(y) = 1_G$. This contradicts the hypothesis that $(u_i)_{i \in \nat}$ is bad. Finally, we show that given $h,k$ such that $1 \leq h < k$, we have $xzv_{i_h} \not\preceq xzv_{i_k}$. By contradiction assume that $xzv_{i_h} \preceq xzv_{i_k}$. One may verify from the definition of $\preceq$ that this yields $v,v' \in A^*$ such that $v_{i_k} = vv'$, $xz \preceq xzv$ and $v_{i_h} \preceq v'$. Moreover, $xz \preceq xzv$ implies that $\eta(xz) = \eta(xzv)$ by definition. Hence, $\eta(v) = 1_G$ since $G$ is a group. It follows that $xyz \preceq xyzv$ by definition of $\preceq$. Since we also have $v_{i_h} \preceq v'$, one may verify from the definition of $\preceq$ that this yields $xyzv_{i_h} \preceq xyzvv'$. Since $u_{i_h} = xyzv_{i_h}$ and $u_{i_k} = xyzv_{i_k} = xyzvv'$, this exactly says that  $u_{i_h} \preceq u_{i_k}$, contradicting the hypothesis that $(u_i)_{i \in \nat}$ is bad.
  \end{proof}

  We may now complete the proof and build the desired cover of the language $H \subseteq A^*$.  We say that a word $v \in H$ is \emph{minimal} if there exists no other word $u \in H$ such that $u \preceq v$. Moreover, we define $F \subseteq H$ as the set of all minimal words of $H$. By definition, we have $u \not\preceq u'$ for every $u,u' \in F$ such that $u \neq u'$.  Hence, it is immediate from Lemma~\ref{lem:half:higman} that $F \subseteq H$ is finite. We define $\Kb = \{\uclos_L u \mid u \in F\}$. It remains to prove that \Kb is a cover of $H$. Since \Kb is finite by definition, we have to prove that for every $v \in H$, there exists $u \in F$ such that $v \in \uclos_L u$. We fix $v$ for the proof. If $v$ is minimal, then $v \in F$ and it is clear that $v\in\uclos_L v$ since $\veps \in L$. Assume now that $v$ is \emph{not} minimal. In that case, there exists another word $u \in H$ which is minimal and such that $u \preceq v$. Since $u$ is minimal, we have $u \in F$. Thus, it suffices to prove that $v \in \uclos_L u$. Since $u \preceq v$, we have $v \in \uclos_{L'} u$ by definition. Moreover, since $L' \subseteq L$, it is immediate that $\uclos_{L'} u \subseteq \uclos_{L} u$. Consequently, we obtain that $v \in \uclos_L u$, which completes the proof.
\end{proof}

We turn to Proposition~\ref{prop:pgpcov}. The statement is as follows.

\pgpcov*

\begin{proof}
  We write $\As = (Q,\delta)$ and consider the \emph{transition morphism} of \As. We let $M = 2^{Q^2}$. It is standard that $M$ is a finite monoid for the following multiplication: given $P,P' \in M$ (\emph{i.e.}, $P,P' \subseteq Q^2$), we let $PP' = \{(q,r) \in Q^2 \mid \text{there is $p \in Q$ such that $(q,p) \in P$ and $(p,r) \in P'$}\}$ (the neutral element is $\{(q,q) \mid q \in Q\}$). The transition morphism $\alpha: A^* \to M$ of \As is defined by $\alpha(a) = \{(q,r) \in Q^2 \mid (q,a,r) \in \delta\}$ for every $a \in A$. Recall that an idempotent $e \in M$ is an element such that $ee=e$.

  We fix $k = |M|^2$ for the proof. We define an auxiliary alphabet \frB. Intuitively, we use a word in $\frB^+$ to represent the \As-guarded decompositions of nonempty words in $A^+$ with length greater than $k$. We write $E \subseteq \alpha(A^+)$ for the set of all idempotents in $\alpha(A^+)$. Consider the following sets (note that the bound on $|w|$ in $\frB_r$ differs from the ones in $\frB_\ell$,~$\frB_c$):
  \[
    \begin{array}{lll}
      \frB_\ell & = & \{(w,f) \in A^+ \times E \mid |w| \leq 2k \text{ and } \alpha(w)f = \alpha(w)\}. \\
      \frB_c & = & \{(e,w,f) \in E \times A^+ \times E \mid |w| \leq 2k \text{ and } e\alpha(w)f = \alpha(w)\}. \\
      \frB_r & = & \{(e,w) \in E \times A^+ \mid |w| \leq k \text{ and } e\alpha(w) = \alpha(w)\}.
    \end{array}
  \]
  We define $\frB = \frB_\ell \cup \frB_r \cup \frB_c$. It is clear from the definition that \frB is finite. We use it as an alphabet and define a morphism $\gamma: \frB^* \to A^*$. Let $b \in \frB$. There exists a nonempty word $w \in A^+$ and $e,f \in E$ such that $b = (w,f) \in \frB_\ell$, $b = (e,w) \in \frB_r$ or $b = (e,w,f) \in \frB_c$. We define $\gamma(b) = w$. Moreover, we write $\gamma_c: \frB_c^* \to A^*$ for the restriction of $\gamma$ to $\frB_c^*$. Finally, we say that a word $x \in \frB^*$ is \emph{well-formed} if $x \in \frB^{}_\ell \frB_c^* \frB^{}_r$ (in particular, $|x| \geq 2$) and $x$ is of the form $x = (w_1,f_1)(e_2,w_2,f_2) \cdots (e_n,w_n,f_n)(e_{n+1},w_{n+1})$ where $f_i = e_{i+1}$ for every $i \leq n$. We have the following lemma.

  \begin{lemma} \label{lem:wformg}
    Let $b_1,\dots,b_m \in \frB$ be letters such that the word $x = b_1 \cdots b_m \in \frB^+$ is well-formed. Then, $(\gamma(b_1),\dots,\gamma(b_m))$ is an \As-guarded decomposition of the word $\gamma(x) \in A^+$.
  \end{lemma}

  \begin{proof}
    Since $\gamma$ is a morphism, it is immediate from the definition that $\gamma(x) = \gamma(b_1) \cdots \gamma(b_m)$. Hence, it suffices to verify that for every $i < m$, there exists $z_i \in A^+$ which is a right \As-loop for $\gamma(b_i)$ and a left \As-loop for $\gamma(b_{i+1})$. By definition of well-formed words, there exists an idempotent $e_i \in E \subseteq \alpha(A^+)$ such that $\alpha(\gamma(b_i))e_i = \alpha(\gamma(b_i))$ and $e_i\alpha(\gamma(b_{i+1})) = \alpha(\gamma(b_{i+1}))$. We let $z_i \in A^+$ be an antecedent of $e_i$: we have $\alpha(z_i) = e_i \in E$. It remains to prove that $z_i$ is a right \As-loop for $\gamma(b_i)$ and a left \As-loop for $\gamma(b_{i+1})$. By symmetry, we only prove the former. We fix $q,r \in Q$ such that $\gamma(b_i) \in \alauto{q}{r}$ for the proof. By definition of $\alpha$, it follows that $(q,r) \in \alpha(\gamma(b_i))$. Hence, since  $\alpha(\gamma(b_i))e_i = \alpha(\gamma(b_i))$ and $e_i = \alpha(z_i)$, we get $(q,r) \in \alpha(\gamma(b_i)z_i)$ which means that $\gamma(b_i)z_i \in \alauto{q}{r}$. This yields $s \in Q$ such that $\gamma(b_i) \in \alauto{q}{s}$ and $z_i \in \alauto{s}{r}$. Finally, since $e_i = \alpha(z_i)$ is an idempotent and $z_i \in \alauto{s}{r}$, one may verify using a pumping argument that there exists $t \in Q$ such that $z_i \in \alauto{s}{t} \cap \alauto{t}{t} \cap \alauto{t}{r}$. This completes the proof.
  \end{proof}

  Intuitively, Lemma~\ref{lem:wformg} states that every well-formed word in $x \in \frB^+$ encodes an \As-guarded decomposition of some word in $A^+$. We handle the converse direction in the following lemma: for every long enough word $w \in A^+$, there exists an \As-guarded decomposition of $w$ which is encoded by a word in $\frB^+$.

  \begin{lemma} \label{lem:half:wformed}
    For every $w \in A^+$ such that $|w| > k$, there exists $x \in \frB^+$ which is well-formed and such that $w = \gamma(x)$.
  \end{lemma}

  \begin{proof}
    We proceed by induction on the length of $w$. Since $|w| > k$, there exist $a_0,\dots,a_{k} \in A$ and $w'\in A^*$ such that $w = w'a_0 \cdots a_k$. Since $k = |M|^2$, we may apply the pigeon-hole principle to obtain $i,j$ such that $0 \leq i < j \leq k$, $\alpha(a_0\cdots a_{i}) =  \alpha(a_0\cdots a_{j})$ and $\alpha(a_{i+1} \cdots a_{k}) =  \alpha(a_{j+1}\cdots a_{k})$. Let $u = a_0\cdots a_{i}$ and $v = a_{i+1} \cdots a_{k}$. We have $u,v \in A^+$, $|u|\leq k$ and $|v| \leq k$. Moreover, $w = w'uv$. We consider the idempotent $e = (\alpha(a_{i+1} \cdots a_j))^\omega \in E$. By definition, we have $\alpha(u)e = \alpha(u)$ and $e\alpha(v) = \alpha(v)$. There are now two cases depending on $w'$.

    Assume first that $|w'| \leq k$. In that case $|w'u| \leq 2k$ which implies that $(w'u,e) \in \frB_\ell$ since $\alpha(u)e = \alpha(u)$. Moreover, we have $(e,v) \in \frB_r$ since $|v| \leq k$ and $e\alpha(v) = \alpha(v)$. Consequently, $x = (w'u,e)(e,v) \in \frB^+$ is a well-formed word such that $\gamma(x) = w'uv = w$. Assume now that $|w'| > k$. Since it is clear that $|w'| < |w|$, induction yields a well-formed word $x' \in \frB^+$ such that $\gamma(x') = w'$. By definition $x' = x''(f,v')$ where $x'' \in \frB^+$ and $(f,v') \in \frB_r$. In particular, we have $|v'| \leq k$ and $f\alpha(v') = \alpha(v')$ by definition of $\frB_r$. Hence, $|v'u| \leq 2k$ which implies that $(f,v'u,e) \in \frB_c$ since $\alpha(u)e = \alpha(u)$. Moreover, we have $(e,v) \in \frB_r$ since $|v| \leq k$ and $e\alpha(v) = \alpha(v)$. Let $x = x''(f,v'u,e)(e,v)$. Clearly, $x$ is well-formed since $x' = x''(f,v)$ was well-formed. Moreover, $\gamma(x) = \gamma(x''(f,v'))uv = w'uv = w$. This concludes the proof.
  \end{proof}

  We now prove Proposition~\ref{prop:pgpcov}.  We define $L_c=\gamma_c\inv(L) \subseteq \frB_c^*$. Since $L$ is a group language (over $A$) and $\veps \in L$, one may verify that $L_c$ is also a group language (over $\frB_c$) and $\veps\in L_c$. Let $b_\ell \in \frB_\ell$ and $b_r \in \frB_r$. We define,
  \[
    H_{b_\ell,b_r} = \{x \in \frB_c^* \mid \text{$b_\ell x b_r \in \frB^*$ is well-formed and $\gamma(b_\ell x b_r ) \in H$}\}.
  \]
  Proposition~\ref{prop:pgcov} yields a \emph{finite} set $F_{b_\ell,b_r} \subseteq H_{b_\ell,b_r} \subseteq \frB_c^*$ such that $\{\uclos_{L_c} x \mid x \in F_{b_\ell,b_r}\}$ is a cover of $H_{b_\ell,b_r}$. We are ready to build the desired cover \Kb of $H \subseteq A^*$. For every word $x = b_1 \cdots b_n  \in \frB_c^*$, every  $b_\ell \in \frB_\ell$ and every $b_r \in \frB_r$, we associate the language $[x]_{b_\ell,b_r} = \gamma(b_\ell)L\gamma(b_1)L \cdots \gamma(b_n)L\gamma(b_r) \subseteq A^+$. Finally, we define,
  \[
    \Kb = \{\{w\} \mid w \in H \text{ and } |w| \leq k\} \cup \bigcup_{b_\ell \in \frB_\ell} \bigcup_{b_r \in \frB_r} \{[x]_{b_\ell,b_r} \mid x \in F_{b_\ell,b_r}\}.
  \]
  It remains to prove that \Kb is the desired cover of $H$. First, let us verify that every $K \in \Kb$ is of the form $K = w_1 L \cdots w_nLw_{n+1}$ where $(w_1,\dots, w_{n+1})$ is an \As-guarded decomposition of some word $w \in H$. This immediate if $K = \{w\}$ for some $w \in H$. We have to handle the case when $K = [x]_{b_\ell,b_r}$ for some $x \in F_{b_\ell,b_r}$. By definition, $x \in H_{b_\ell,b_r}$ which means that $b_\ell x b_r \in \frB^*$ is well-formed and $\gamma(b_\ell x b_r ) \in H$. Let $b_1,\dots,b_n \in \frB_c^*$ be the letters such that $x = b_1 \cdots b_n$. Since $b_\ell b_1 \cdots b_n b_r \in \frB^*$ is well-formed, Lemma~\ref{lem:wformg} yields that $(\gamma(b_\ell),\gamma(b_1),\dots,\gamma(b_n),\gamma(b_r))$ is an \As-guarded decomposition of $\gamma(b_\ell x b_r ) \in H$. This concludes the proof since $K=  [x]_{b_\ell,b_r} = \gamma(b_\ell)L\gamma(b_1)L \cdots \gamma(b_n)L\gamma(b_r)$.

  We now prove that \Kb is a cover of $H$. It is immediate by definition that \Kb is finite. Given $w \in H$, we exhibit $K \in \Kb$ such that $w \in K$. This is immediate if $|w|\leq k$: we have $\{w\} \in \Kb$ in that case. We now consider the case $|w|>k$. Lemma~\ref{lem:half:wformed} yields $x \in B^+$ which is well-formed and such that $w = \gamma(x)$. By definition of well-formed words $x = b_\ell y b_r$ where $y \in \frB_c^*$, $b_\ell \in \frB_\ell$ and $b_r \in \frB_r$. Therefore, since $\gamma(x) = w \in H$, we have $y \in H_{b_\ell,b_r}$ by definition. Hence, since $\{\uclos_{L_c} z \mid z \in F_{b_\ell,b_r}\}$ is a cover of $H_{b_\ell,b_r}$, we get $z \in  F_{b_\ell,b_r}$ such that $y \in \uclos_{L_c} z$.  We prove that $w \in [z]_{b_\ell,b_r}$ which concludes the proof since $[z]_{b_\ell,b_r} \in \Kb$ by definition. We have $\uclos_{L_c} z = L_cb_1L_c \cdots b_nL_c$ where $b_1,\dots,b_n \in \frB_c$ are the letters such that $b_1 \cdots b_n = z \in H_{b_\ell,b_r}$. Therefore, since $y \in \uclos_{L_c} z$, this yields $x_0, \dots, x_n \in L_c$ such that $y = x_0b_1x_1 \cdots b_nx_n$. Altogether, it follows that $x = b_\ell x_0b_1x_1 \cdots b_nx_n b_r$. Since $w = \gamma(x)$, we get $w = \gamma(b_\ell) \gamma(x_0) \gamma(b_1) \gamma(x_1) \cdots \gamma(b_n) \gamma(x_n) \gamma(b_r)$. Finally, since $L_c = \gamma_c\inv(L)$ and $x_0,\dots,x_n \in L_c$, we have $\gamma(x_i) \in L$ for every $i \leq n$. Hence, we obtain that $w \in \gamma(b_\ell)L\gamma(b_1)L \cdots \gamma(b_n)L\gamma(b_r)$. This exactly says that $w \in [z]_{b_\ell,b_r}$ since $b_1 \cdots b_n = z$ by definition. This concludes the proof.
\end{proof}

\section{Appendix to Section~\ref{sec:separ}}
\label{app:separ}
We present the missing proofs for the statements in Section~\ref{sec:separ}. We start with those concerning classical separation.

\subsection{Non-separable quadruples}

We first prove Lemma~\ref{lem:sepconcat}. The statement is as follows.

\sepconcat*

\begin{proof}
	Given $K \in \Cs$ such that $L_0H_0 \subseteq K$, we prove that $L_1H_1 \cap K \neq \emptyset$. Consider the two following languages:
	\[
	U = \bigcap_{w \in H_0} Kw\inv \quad \text{and} \quad V = \bigcap_{u \in U} u\inv K.
	\]
	Since $K \in \Cs$ is regular, it has finitely many quotients by the Myhill-Nerode theorem. Therefore, while the above intersections may be infinite, they boil down to finite ones. Since \Cs is a \vari and $K \in \Cs$, it follows that $U,V \in \Cs$. Moreover, $L_0 \subseteq U$. Indeed, if $x \in L_0$, then for every $w \in H_0$ we have $xw \in L_0H_0 \subseteq K$ which yields $x \in U$ by definition of $U$. Therefore, since $L_0$ is not \Cs-separable from $L_1$, we get $L_1 \cap U\neq \emptyset$. We fix $u \in L_1 \cap U$. Additionally, $H_0 \subseteq V$. Indeed, if $y \in H_0$, then for every $x \in U$, we have $xy \in K$ which yields $y \in V$ by definition. Consequently, since $H_0$ is not \Cs-separable from $H_1$, we get $H_1 \cap V \neq \emptyset$. Let $v \in H_1 \cap V$. Altogether, we have $uv \in L_1H_1$, $u \in U$ and $v \in V$. By definition of $V$, we get $uv \in K$. Hence,  $L_1H_1 \cap K \neq \emptyset$ as desired.
\end{proof}

We first consider Proposition~\ref{prop:autosep}. The statement is as follows.

\autosep*

\begin{proof}
	Assume first that $\alauto{I_1}{F_1}$ is \Cs-separable from $\alauto{I_2}{F_2}$. This yields a separator $K\in\Cs$. It is clear that for every $(q_1,r_1,q_2,r_2) \in I_1 \times F_1 \times I_2 \times F_2$, $K \in \Cs$ also separates $\alauto{q_1}{r_1}$ from $\alauto{q_2}{r_2}$ (these two languages are included in $\alauto{I_1}{F_1}$ and $\alauto{I_2}{F_2}$ respectively). Thus, $\left(I_1\times F_1\times I_2\times F_2\right) \cap \nsepca = \emptyset$ by definition of \nsepca.
	
	Conversely, assume that $\left(I_1 \times F_1 \times I_2\times F_2 \right) \cap \nsepca = \emptyset$. By definition, this means that for every $(q_1,r_1,q_2,r_2) \in I_1 \times F_1 \times I_2 \times F_2$, there exists a language $K_{q_1,r_1,q_2,r_2} \in \Cs$ which separates $\alauto{q_1}{r_1}$ from $\alauto{q_2}{r_2}$. Consider the following language:
	\[
	K = \bigcup_{(q_1,r_1) \in I_1 \times F_1} \left(\bigcap_{(q_2,r_2) \in I_2 \times F_2} K_{q_1,r_1,q_2,r_2} \right).
	\]
	Since \Cs is a lattice, we have $K \in \Cs$. Moreover, one may verify that $K$ separates $\alauto{I_1}{F_1}$ from $\alauto{I_2}{F_2}$, concluding the proof.
\end{proof}

We now prove Lemma~\ref{lem:autosep} whose statement is the following.

\lautosep*

\begin{proof}
  By definition of \nsepca, for every quadruple $\bar{q} = (q,r,s,t) \in Q^4 \setminus \nsepca$, there exists $H_{\bar{q}} \in \Cs$ which separates \alauto{q}{r} from \alauto{s}{t}. We use the languages $H_{\bar{q}}$ to define the following equivalence $\sim$ on $A^*$: for $u,v \in A^*$, we let $u \sim v$ if and only if $u \in H_{\bar{q}} \Leftrightarrow v \in H_{\bar{q}}$ for every $\bar{q} \in Q^4 \setminus \nsepca$. Let \Kb be the partition of $A^*$ into $\sim$-classes. By definition, each $K \in \Kb$ is a Boolean combination of languages $H_{\bar{q}} \in \Cs$. Hence, $K \in \Cs$ since \Cs is a Boolean algebra. We conclude that \Kb is a \Cs-cover of $A^*$. Moreover, one may verify from the definition of the languages $H_{\bar{q}}$ that  \Kb is separating for \nsepca.

  For the second assertion, we let $S \subseteq Q^4$ and consider a \Cs-cover \Kb of $A^*$ which is separating for $S$. We show that $Q^4 \setminus S \subseteq Q^4 \setminus \nsepca$. By definition, this boils down to proving that if $(q,r,s,t) \in Q^4 \setminus S$, then \alauto{q}{r} is \Cs-separable from \alauto{s}{t}. We build a separator $H \in \Cs$ from \Kb. Let $H \in \Cs$ be the union of all languages $K \in \Kb$ such that $K \cap \alauto{q}{r} \neq \emptyset$. Clearly, $\alauto{q}{r} \subseteq H$ since \Kb is a cover of $A^*$. It remains to prove that $H \cap \alauto{s}{t} = \emptyset$. By contradiction, assume that $H \cap \alauto{s}{t} \neq \emptyset$. By definition of $H$, this yields $K \in \Kb$ such that $K \cap \alauto{q}{r} \neq \emptyset$ and $K \cap \alauto{s}{t} \neq \emptyset$. Since \Kb is separating for $S$, it follows that $(q,r,s,t) \in S$ which is a contradiction since $(q,r,s,t) \in Q^4 \setminus S$.
\end{proof}

We turn to Proposition~\ref{lem:gautosep}. The statement is as follows.

\gautosep*

\begin{proof}
	For every quadruple $\bar(q) = (q,r,s,t) \in Q^4 \setminus \dnsepca$, there exists $L_{\bar{q}} \in \Ds$ such that $\veps \in \Ds$ and $H_{\bar{q}} \in \Cs$ which separates $\alauto{q}{r} \cap L_{\bar{q}}$ from $\alauto{s}{t}\cap L_{\bar{q}}$. We define $L \in \Ds$ as the intersection of all languages $L_{\bar{q}}$ for $\bar{q} \in Q^4 \setminus \dnsepca$. Clearly, we have $\veps \in L$. Moreover, we define an equivalence $\sim$ on $L$: for $u,v \in L$, we let $u \sim v$ if and only if $u \in H_{\bar{q}} \Leftrightarrow v \in H_{\bar{q}}$ for every $\bar{q} \in Q^4 \setminus \dnsepca$. Finally, we let \Kb be the partition of $L$ into $\sim$-classes. By definition, each $K \in \Kb$ is a Boolean combination involving the languages $H_{\bar{q}} \in \Cs$ and $L \in \Ds \subseteq \Cs$. Hence, $K \in \Cs$ since \Cs is a Boolean algebra. We conclude that \Kb is a \Cs-cover of $L$. Moreover, one may verify from the definition of the languages $H_{\bar{q}}$ that  \Kb is separating for \dnsepca.
	
	For the second assertion, we let $S \subseteq Q^4$. Consider $L \in \Ds$ such that $\veps \in L$ and a \Cs-cover \Kb of $L$ which is separating for $S$. We show that $Q^4 \setminus S \subseteq Q^4 \setminus \dnsepca$. By definition,it suffices to prove that if $(q,r,s,t) \in Q^4 \setminus S$, then $\alauto{q}{r} \cap L$ is \Cs-separable from $\alauto{s}{t} \cap L$. We build a separator $H \in \Cs$ from \Kb. Let $H \in \Cs$ be the union of all languages $K \in \Kb$ such that $K \cap \alauto{q}{r} \cap L \neq \emptyset$. Clearly, $\alauto{q}{r} \cap L \subseteq H$ since \Kb is a cover of $L$. It remains to prove that $H \cap \alauto{s}{t} \cap L = \emptyset$. By contradiction, assume that $H \cap \alauto{s}{t} \cap L \neq \emptyset$. By definition of $H$, this yields $K \in \Kb$ such that $K \cap \alauto{q}{r} \cap L \neq \emptyset$ and $K \cap \alauto{s}{t} \cap L\neq \emptyset$. Since \Kb is separating for $S$, we get $(q,r,s,t) \in S$. This is a contradiction since $(q,r,s,t) \in Q^4 \setminus S$.
\end{proof}

\subsection{Tuple separation}

We now present proofs for the statements concerning tuple separation. We start with Lemma~\ref{lem:tupconcat}.

\tupconcat*

\begin{proof}
  We proceed by induction on $n \geq 1$. If $n = 1$, then $(L_1)$ and $(H_1)$ being not \Cs-separable means that $L_1 \neq \emptyset$ and $H_1 \neq \emptyset$. Hence, $L_1H_1 \neq \emptyset$ which implies that $(L_1H_1)$ is not \Cs-separable. Assume now $n > 1$ and consider $(L_1,\dots,L_n),(H_1,\dots,H_n)$ which are not \Cs-separable. Given $K \in \Cs$ such that $L_1H_1 \subseteq K$, we prove that $(L_2H_2,\dots,L_nH_n) \cap K$ is not \Cs-separable. Consider the two following languages:
  \[
	U = \bigcap_{w \in H_1} Kw\inv \quad \text{and} \quad V = \bigcap_{u \in U} u\inv K.
  \]
  Note that since $K \in \Cs$ is regular, it has finitely many quotients by the Myhill-Nerode theorem. Hence, while the above intersections may be infinite, they boil down to finite ones. Since \Cs is a \vari and $K \in \Cs$, it follows that $U,V \in \Cs$.

  Observe that $L_1 \subseteq U$. Indeed, if $x \in L_1$, then for every $w \in H_1$ we have $xw \in L_1H_1 \subseteq K$ which yields $x \in U$ by definition. Since $(L_1,\dots,L_n)$ is not \Cs-separable, it follows that $(L_2,\dots,L_n) \cap U$ is not \Cs-separable. Moreover, observe that $H_1 \subseteq V$. Indeed, if $y \in H_1$, then for every $x \in U$, we have $xy \in K$ which yields $y \in V$ by definition. Since $(H_1,\dots,H_n)$ is not \Cs-separable, it follows that $(H_2,\dots,H_n) \cap V$ is not \Cs-separable. It now follows from induction on $n$ that $((L_2 \cap U)(H_2 \cap V),\dots,(L_n \cap U)(H_n \cap V))$ is not \Cs-separable. It is clear that $(L_i \cap U)(H_i \cap V) \subseteq (L_iH_i) \cap (UV)$ for every $i \leq n$. Moreover, observe that $UV \subseteq K$. Indeed, if $u \in U$ and $v \in V$, we have $uv \in K$ by definition of $V$. Altogether, it follows from the second assertion in Lemma~\ref{lem:tuptriv} that $(L_2H_2,\dots,L_nH_n) \cap K$ is not \Cs-separable, which completes the proof.
\end{proof}

We turn to Corollary~\ref{cor:covtsep}. As we explained in the main paper, this statement follows from a theorem of~\cite{pzbpol} which we first recall and prove.

\begin{restatable}{theorem}{covtsep}\label{thm:covtsep}
	Let \Cs be a lattice and $L_0,L_1 \subseteq A^*$. The following properties are equivalent:
	\begin{enumerate}
		\item $L_0$ is \bool{\Cs} separable from $L_1$.
		\item There exists $p \geq 1$ and such that $(L_0,L_1)^p$ is \Cs-separable.
	\end{enumerate}
\end{restatable}

\begin{proof}
  We first prove that $2) \Rightarrow 1)$. Let $L_0,L_1 \subseteq A^*$ and assume that there exists $p \geq 1$ such that $(L_0,L_1)^p$ is \Cs-separable. We use induction on $p$ to prove that $L_0$ is \bool{\Cs}-separable from~$L_1$. When $p = 1$, $L_0$ is \Cs-separable from $L_1$ and since $\Cs \subseteq \bool{\Cs}$, the result is trivial. Assume that $p \geq 2$. By hypothesis, we have $K,K' \in \Cs$ such that $L_0 \subseteq K$, $L_1 \cap K \subseteq K'$ and $(L_0,L_1)^{p-1} \cap K \cap K'$ is \Cs-separable. Using induction, we then obtain a language $P \in \bool{\Cs}$ separating $L_0 \cap K \cap K'$ from $L_1 \cap K \cap K'$. Consider the language $H = (K \cap P) \cup (K \setminus K') \in \bool{\Cs}$. We prove that $H$ separates $L_0$ from~$L_1$. We begin with $L_0 \subseteq H$. Let $w \in L_0$, we prove that $w \in H$. Clearly, $w \in K$ since $L_0 \subseteq K$. Moreover, either $w \in K'$ and therefore $w \in K \cap P$ since $L_0 \cap K \cap K' \subseteq P$, or $w \not\in K'$ and therefore $w \in K \setminus K'$. Altogether, we conclude that $w \in H$. It remains to prove that $L_1 \cap H = \emptyset$. Let $w \in L_1$, we prove that $w \not\in H$. There are two cases depending on whether $w \in K$. If $w \not\in K$, then clearly $w \not\in K \cap P$ and $w \not\in  K \setminus K'$, hence $w \not\in H$. Otherwise, $w \in L_1 \cap K \subseteq K'$. Therefore, $w \not\in K \setminus K'$ and $w \not\in K \cap P$ since $L_1 \cap K \cap K' \cap P = \emptyset$ by the choice of~$P$. We get $w \not\in H$, which completes the proof.

  \smallskip

  We turn to the implication $1) \Rightarrow 2)$ in Theorem~\ref{thm:covtsep}. We start with an auxiliary lemma.

  \begin{lemma} \label{lem:sepunion}
    Let $k \geq 1$ and $(L_1,\dots,L_k)$ be a $k$-tuple. Moreover, let $K_1,K_2 \in \Cs$ be such that $(L_1,\dots,L_k) \cap K_1$ and $(L_1,\dots,L_k) \cap K_2$ are both \Cs-separable. Then, $(L_1,\dots,L_k) \cap (K_1 \cup K_2)$ is \Cs-separable as well.
  \end{lemma}

  \begin{proof}
    We proceed by induction on $k$. When $k = 1$, then we have $L_1 \cap K_1 = \emptyset$ and $L_1 \cap K_2 = \emptyset$ by hypothesis. Hence, $L_1 \cap (K_1 \cup K_2) = \emptyset$ and $(L_1) \cap (K_1 \cup K_2)$ is \Cs-separable. When $k \geq 2$, for $i = 1,2$, our hypothesis yields a separator $U_i \in \Cs$ for $(L_1,\dots,L_k) \cap K_i$. We prove that $U = (U_1 \cap K_1) \cup (U_2 \cap K_2) \in \Cs$ is a separator for $(L_1,\dots,L_k) \cap (K_1 \cup K_2)$. It is clear that $L_1 \cap (K_1 \cup K_2) \subseteq U$ since we have $L_1 \cap K_1 \subseteq U_1$ and $L_2 \cap K_2 \subseteq U_2$ by definition of $U_1$ and $U_2$. Moreover, we know that $(L_2,\dots,L_k) \cap K_1 \cap U_1$ and $(L_2,\dots,L_k) \cap K_2 \cap U_2$ are both \Cs-separable. Thus, it is immediate from induction that $(L_2,\dots,L_k) \cap U$ is \Cs-separable, concluding the proof.
  \end{proof}

  We now concentrate on proving the implication $1) \Rightarrow 2)$ in Theorem~\ref{thm:covtsep}. Given $L_0,L_1 \subseteq A^*$ which are \bool{\Cs}-separable, we have to prove that there exists $p \geq 1$ such that $(L_0,L_1)^p$ is \Cs-separable. By hypothesis there exists a language $K \in \bool{\Cs}$ such that $L_0 \subseteq K$ and $L_1 \cap K = \emptyset$. By definition, $K$ is the Boolean combination of languages in \Cs. We put it in disjunctive normal form. Each disjunct is an intersection languages belonging to \Cs, or whose complement belongs to \Cs. Since \Cs is lattice, both \Cs and the complement class \cocl{\Cs} are closed under intersection. Therefore, each disjunct in the disjunctive normal form of $K$ is actually of the form $K' \setminus H'$, where $K',H'$ both belong to \Cs (for the case where $K'$ or $H'$ is empty, recall that both $\emptyset$ and $A^*$ belong to \Cs). In other words, there exist $n \geq 1$ and $K_1,\dots,K_n,H_1,\dots,H_n \in \Cs$ such that $K =  \bigcup_{1 \leq i \leq n} (K_i \setminus H_i)$. We use induction on $n \geq 1$ to prove that $(L_0,L_1)^{n+1}$ is \Cs-separable

  Assume first that $n = 1$. We prove that $(L_0,L_1,L_0,L_1) = (L_0,L_1)^2$ is \Cs-separable. By hypothesis, $K = K_1 \setminus H_1$, $L_0 \subseteq K_1\setminus H_1$ and $L_1 \cap (K_1 \setminus H_1) = \emptyset$. Clearly, $L_0 \subseteq K_1 \in \Cs$. Thus, it remains to prove that $(L_1,L_0,L_1) \cap K_1$ is \Cs-separable. Since $L_1 \cap (K_1 \setminus H_1) = \emptyset$, we have $L_1 \cap K_1 \subseteq H_1$. Thus, it now remains to prove that $(L_0,L_1) \cap K_1 \cap H_1$ is \Cs-separable. Since $L_0 \subseteq K_1 \setminus H_1$, we have $L_0 \cap K_1 \cap H_1 = \emptyset$. Thus, it is immediate that $(L_0,L_1) \cap K_1 \cap H_1$ is \Cs-separable, as desired.

  We now assume that $n > 1$. We prove that $(L_0,L_1)^{n+1}$ is \Cs-separable. In the proof, we write $\bar{L}$ for $(2n+1)$-tuple $(L_1) \cdot (L_0,L_1)^n$. Since $L_0 \subseteq K$ and $K = \bigcup_{1 \leq i \leq n} (K_i \setminus H_i)$. We know that $L_0\subseteq\bigcup_{1\leq i \leq n} K_i \in \Cs$. Therefore, it now remains to prove that,
  \[
    \bar{L} \cap \left(\bigcup_{1\leq i \leq n} K_i\right) \quad \text{is \Cs-separable.}
  \]
  In view of Lemma~\ref{lem:sepunion}, since each language $K_i$ belongs to \Cs by hypothesis, it now suffices to prove that $\bar{L} \cap K_i$ is \Cs-separable for every $i \leq n$. We fix $i \leq n$ for the proof. By hypothesis, $L_1 \cap K = \emptyset$ which implies that $L_1 \cap (K_i \setminus H_i) = \emptyset$. Hence, $L_1 \cap K_i \subseteq H_i \in \Cs$. Hence, by definition of $\bar{L}$, proving that $\bar{L} \cap K_i$ is \Cs-separable boils down to proving that $(L_0,L_1)^n \cap K_i \cap H_i$ is \Cs-separable. We use induction on $n$. Let $K' = \bigcup_{j \neq i} (K_j \setminus H_j) \in \bool{\Cs}$ by definition $K'$ is the union $n-1$ languages $K_j \setminus H_j$. Moreover, since $L_0 \subseteq K$ and $L_1 \cap K = \emptyset$, it is immediate that $L_0 \cap K_i \cap H_i \subseteq K'$ and $L_1 \cap K' = \emptyset$. Hence, it follows by induction on $n$ that $(L_0,L_1)^{n} \cap K_i \cap H_i$ is \Cs-separable which completes the proof.
\end{proof}

We may now prove Corollary~\ref{cor:covtsep} itself. We first recall the statement.

\ccovtsep*

\begin{proof}
	Assume first that $L_0$ is \bool{\Cs}-separable from $L_1$ under \Ds-control. By definition, this yields $H \in \Ds$ such that $\veps \in H$ and $L_0 \cap H$ is \bool{\Cs}-separable from $L_1 \cap H$. Hence, Theorem~\ref{thm:covtsep} yields $p \geq 1$ such that $(L_0,L_1)^p \cap H$ is \Cs-separable. We conclude that $(L_0,L_1)^p$ is \bool{\Cs}-separable under \Ds-control, as desired.
	
	Conversely, assume that there exists $p \geq 1$ such that $(L_0,L_1)^p$ is \bool{\Cs}-separable under \Ds-control. We get $H \in \Ds$ such that $\veps \in H$ and $(L_0,L_1)^p \cap H$ is \Cs-separable. Therefore, Theorem~\ref{thm:covtsep} implies that  $L_0 \cap H$ is \bool{\Cs}-separable from $L_1 \cap H$. By definition, we conclude that  $L_0$ is \bool{\Cs}-separable from $L_1$ under \Ds-control, which completes the proof.
\end{proof}

We turn to Lemma~\ref{lem:pgsound}. The statement is as follows.

\pgsound*

\begin{proof}
	We prove the contrapositive. Assume that $(\{\veps\},L_1,\dots,L_n)$ is \pol{\Ds}-separable: there exists $K \in \pol{\Ds}$ such that $\veps \in K$ and $(L_1,\dots,L_n) \cap K$ is \pol{\Ds}-separable. By definition of \pol{\Ds}, $K$ is a finite union of marked product of languages in \Ds. Hence, since $\veps \in K$, there exists a marked product involving a single language $H \in \Ds$  such that $\veps \in H$ in the union defining $K$. In particular, $H \subseteq K$ and Lemma~\ref{lem:tuptriv} implies that $(L_1,\dots,L_n) \cap H$ is \pol{\Ds}-separable. Since $H \in \Ds$ and $\veps \in H$, it follows that  $(L_1,\dots,L_n)$ is \pol{\Ds}-separable under \Ds-control.
\end{proof}


Finally, we prove Lemma~\ref{lem:pgpsound} whose statement is as follows.

\pgpsound*

\begin{proof}
	We prove the contrapositive. Assume that $(w^+,w^+L_1w^+,\dots,w^+L_nw^+)$ is \pol{\Ds^+}-separable. We show that $(L_1,\dots,L_n)$ is \pol{\Ds^+}-separable under \Ds-control. By hypothesis, there exists $K \in \pol{\Ds^+}$ such that $w^+ \subseteq K$, and  $(w^+L_1w^+,\dots,w^+L_nw^+) \cap K$ is \pol{\Ds^+}-separable. By definition, $K$ is a finite union of languages $K_0a_1K_1 \cdots a_mK_m$ with \mbox{$a_1,\dots, a_m \in A$} and $K_0,\dots,K_m \in \Ds^+$. Let $k \in \nat$ such that $m \leq k$ for every marked product $K_0a_1K_1 \cdots a_mK_m$ in the finite union defining $K$. Consider the word $w^{2(k+1)} \in w^+$. Since $w^+ \subseteq K$, we have $w^{2(k+1)} \in K$. Hence, there exists a marked product $K_0a_1K_1 \cdots a_mK_m$ in the finite union defining $K$ (in particular $m \leq k$) such that,
	\[
	w^{2(k+1)} \in K_0a_1K_1 \cdots a_mK_m \subseteq K.
	\]
	We get a word $u_i \in K_i$ for each $i \leq m$ such that $w^{2(k+1)} = u_0a_1u_1 \cdots a_mu_m$. Since $m \leq k$, there exists $i \leq m$ such that $ww$ is an infix of $u_i$.  Thus, we get $x,y \in A^*$ and $\ell_1, \ell_2 \in \nat$ such that $u_i = xwwy$, $u_0a_1u_1 \cdots a_ix = w^{\ell_1}$, $ya_{i+1}u_{i+1} \cdots a_mu_m = w^{\ell_2}$ and $\ell_1+2+ \ell_2 = 2(k+1)$
	
	By definition $K_i \in \Ds^+$ which means that there exists a language $H \in \Ds$ such that either $K_i = H \cup \{\veps\}$ or $K_i = H \cap A^+$. In particular, since $u_i \in K_i$ and $u_i \in A^+$ (recall that $w \in A^+$), we have $xwwy = u_i \in H$. Let $H' = (xw)\inv H (wy)\inv$. By closure under quotients, we have $H' \in \Ds$ and it is clear that $\veps \in H'$ since $xwwy \in H$. Hence, it now suffices to prove that $(L_1,\dots,L_n) \cap H'$ is \pol{\Ds^+}-separable. This will imply as desired that $(L_1,\dots,L_n)$ is \pol{\Ds^+}-separable under \Ds-control. 
	
	By contradiction, assume that $(L_1,\dots,L_n) \cap H'$ is \emph{not} \pol{\Ds^+}-separable.  by Lemma~\ref{lem:tuptriv}, the $n$-tuples $(\{xw\})^n$ and $(\{yw\})^n$ are not \pol{\Ds^+}-separable as well.  Hence, we obtain from Lemma~\ref{lem:tupconcat} that,
	\[
	(xw(L_1 \cap H')wy,\dots,xw(L_n \cap H')wy) \quad \text{is not \pol{\Ds^+}-separable}.
	\]
	By definition of $H$ and $H'$, we know that $xw(L_j \cap H')wy \subseteq xwL_jwy \cap H \subseteq xwL_jwy \cap K_i$ for every $j \leq n$. Hence, we conclude that $(xwL_1wy,\dots,xwL_nwy) \cap K_i$ is not \pol{\Ds^+}-separable. We may now use Lemma~\ref{lem:tuptriv} again to obtain that the $n$-tuples $(\{u_0a_1u_1 \cdots a_i\})^n$ and $(\{a_{i+1}u_{i+1} \cdots a_mu_m\})^n$ are not \pol{\Ds^+}-separable. Therefore, since $u_0a_1u_1 \cdots a_ix = w^{\ell_1}$, $ya_{i+1}u_{i+1} \cdots a_mu_m = w^{\ell_2}$ and $u_i \in K_i$ for every $i \leq m$, one may use Lemma~\ref{lem:tuptriv} and Lemma~\ref{lem:tupconcat} to obtain that,
	\[
	(w^{\ell_1+1}L_1w^{\ell_2+1},\dots,w^{\ell_1+1}L_nw^{\ell_2+1}) \cap (K_0a_1K_1 \cdots a_mK_m)  \quad \text{is not \pol{\Ds^+}-separable}.
	\]
	Since $K_0a_1K_1 \cdots a_mK_m \subseteq K$ and $w^{\ell_1+1}L_jw^{\ell_2+1} \subseteq w^+L_j w^+$, we may apply Lemma~\ref{lem:tuptriv} one last time to obtain that $(w^+L_1w^+,\dots,w^+L_nw^+) \cap K$ is not \pol{\Ds^+}-separable. This is a contradiction.
\end{proof}

\section{Proof of Theorem~\ref{thm:pgauto}}
\label{app:bpolgp}
We provide the missing proofs in Section~\ref{sec:bpolg}. First, we prove Lemma~\ref{lem:ginc} and Lemma~\ref{lem:pginc} which are fairly simple statements. Then, we concentrate on the proof of Theorem~\ref{thm:pgauto}.

\subsection{Lemma~\ref{lem:ginc} and Lemma~\ref{lem:pginc}}

Let us first recall the statement of Lemma~\ref{lem:ginc}.

\ginc*

\begin{proof}
	We assume that $S \subseteq S'$. Let $(q,r,s,t) \in \tautg(S)$. We prove that $(q,r,s,t) \in \tautg(S')$. Consider the \nfas $\Bs_S = (Q^3,\gamma_{S})$ and $\Bs_{S'} = (Q^3,\gamma_{S'})$. Since $S \subseteq S'$, the definition yields $\gamma_S \subseteq \gamma_{S'}$. Hence,  $\lauto{\Bs_S}{(s,q,s)}{(t,r,t)} \subseteq \lauto{\Bs_{S'}}{(s,q,s)}{(t,r,t)}$ and $\lauto{\Bs_S}{(q,s,q)}{(r,t,r)} \subseteq \lauto{\Bs_{S'}}{(q,s,q)}{(r,t,r)}$. Finally, since $(q,r,s,t) \in \tautg(S)$, we know that~\eqref{eq:gauto} holds: $\{\veps\}$ is not \Gs-separable from \lauto{\Bs_S}{(s,q,s)}{(t,r,t)} and $\{\veps\}$ is not \Gs-separable from \lauto{\Bs_S}{(q,s,q)}{(r,t,r)}. Hence, the above inclusions imply that $\{\veps\}$ is not \Gs-separable from \lauto{\Bs_{S'}}{(s,q,s)}{(t,r,t)} and $\{\veps\}$ is not \Gs-separable from \lauto{\Bs_{S'}}{(q,s,q)}{(r,t,r)}. We obtain $(q,r,s,t) \in \tautg(S')$ as desired.
\end{proof}

We turn to Lemma~\ref{lem:pginc}.

\pginc*

\begin{proof}
	We assume that $S \subseteq S'$. Let $(q,r,s,t) \in \betg(S)$. We prove that $(q,r,s,t) \in \betg(S')$. Consider the \nfas $\Bs^+_S = (Q^3,\gamma^+_{S})$ and $\Bs^+_{S'} = (Q^3,\gamma^+_{S'})$. Since $S \subseteq S'$, the definition yields $\gamma^+_S \subseteq \gamma^+_{S'}$. Hence, $\lauto{\Bs^+_S}{(s,q,s)}{(t,r,t)} \subseteq \lauto{\Bs^+_{S'}}{(s,q,s)}{(t,r,t)}$ and $\lauto{\Bs^+_S}{(q,s,q)}{(r,t,r)} \subseteq \lauto{\Bs^+_{S'}}{(q,s,q)}{(r,t,r)}$. Finally, since $(q,r,s,t) \in \betg(S)$, we know that~\eqref{eq:gpauto} holds: $\{\veps\}$ is not \Gs-separable from \lauto{\Bs^+_S}{(s,q,s)}{(t,r,t)} and $\{\veps\}$ is not \Gs-separable from \lauto{\Bs^+_S}{(q,s,q)}{(r,t,r)}. Hence, the above inclusions imply that $\{\veps\}$ is not \Gs-separable from \lauto{\Bs^+_{S'}}{(s,q,s)}{(t,r,t)} and $\{\veps\}$ is not \Gs-separable from \lauto{\Bs^+_{S'}}{(q,s,q)}{(r,t,r)}. We obtain $(q,r,s,t) \in \betg(S')$ as desired.
\end{proof}

\subsection{Theorem~\ref{thm:pgauto}}

Let us first recall the statement.

\pgauto*

 The proof argument is based on the same outline as the one presented for Theorem~\ref{thm:gauto} in the main paper. We fix a group \vari \Gs and an \nfa $\As = (Q,\delta)$. Let $S \subseteq Q^4$ be the greatest $(\bpoln,+)$-sound subset for \Gs and \As. We prove that $S = \gnsbgp$.

\smallskip
\noindent
{\bf First part: $S\subseteq\gnsbgp$.} We use \emph{tuple separation} and Lemma~\ref{lem:pgpsound}. Let us start with terminology. For every $n \geq 1$ and $(q_1,r_1,q_2,r_2) \in Q^4$, we associate an $n$-tuple $T_n(q_1,r_1,q_2,r_2)$. We use induction on $n$ and tuple concatenation to present the definition. If $n = 1$ then, $T_1(q_1,r_1,q_2,r_2) = (\alauto{q_2}{r_2})$. If $n > 1$, then,
\[
T_n(q_1,r_1,q_2,r_2) = \left\{\begin{array}{ll}
	(\alauto{q_2}{r_2}) \cdot T_{n-1}(q_1,r_1,q_2,r_2) & \text{if $n$ is odd} \\
	(\alauto{q_1}{r_1}) \cdot T_{n-1}(q_1,r_1,q_2,r_2) & \text{if $n$ is even.}
\end{array}\right.
\]
We use induction on $n$ to prove the following proposition.

\begin{restatable}{proposition}{gppsound} \label{prop:gppsound}
	For every $n \geq 1$ and $(q_1,r_1,q_2,r_2) \in S$, the $n$-tuple $T_n(q_1,r_1,q_2,r_2)$ is not \pol{\Gs^+}-separable under \Gs-control. 
\end{restatable}

By definition, Proposition~\ref{prop:gppsound} implies that for every $p \geq 1$ and every $(q_1,r_1,q_2,r_2) \in S$, the $2p$-tuple $(\alauto{q_1}{r_1},\alauto{q_2}{r_2})^p$ is not \pol{\Gs^+}-separable under \Gs-control. By Corollary~\ref{cor:covtsep}, it follows that $\alauto{q_1}{r_1}$ is not \bpol{\Gs^+}-separable from $\alauto{q_2}{r_2}$ under \Gs-control, \emph{i.e.} that $(q_1,r_1,q_2,r_2)  \in \gnsbgp$. We get $S \subseteq \gnsbgp$ as desired.

We prove Proposition~\ref{prop:gppsound} using induction on $n$. We fix $n \geq 1$ for the proof. In order to exploit the fact that $S$ is $(\bpoln,+)$-sound, we need a property of the \nfa $\Bs^+_S = (Q^3,\gamma_S)$ used to define \betg. When $n \geq 2$, this is where we use induction on $n$ and Lemma~\ref{lem:pgpsound}.

\begin{restatable}{lemma}{gpsinduc}\label{lem:gpsinduc}
	Consider $(s_1,s_2,s_3),(t_1,t_2,t_3) \in Q^3$ and a group language $H \subseteq A^*$. Assume that  $H \cap \lauto{\Bs^+_{S}}{(s_1,s_2,s_3)}{(t_1,t_2,t_3)} \neq \emptyset$. Then, $H \cap \alauto{s_1}{t_1} \neq \emptyset$ and, if $n \geq 2$, then the $n$-tuple $(H \cap \alauto{s_1}{t_1})  \cdot T_{n-1}(s_{2},t_{2},s_{3},t_{3})$ is not \pol{\Gs^+}-separable.
\end{restatable}

\begin{proof}
	By hypothesis, there exists $w \in H \cap \lauto{\Bs^+_{S}}{(s_1,s_2,s_3)}{(t_1,t_2,t_3)}$. Hence, the \nfa $\Bs^+_S$ contains some run labeled by $w$ from $(s_1,s_2,s_3)$ to $(t_1,t_2,t_3)$. We use a sub-induction on the number of transitions involved in that run. When no transitions are used: we have $w = \veps$ and $(s_1,s_2,s_3) = (t_1,t_2,t_3)$. It follows that $w = \veps \in H \cap \alauto{s_1}{t_1}$. Moreover, if $n \geq 2$, the $n$-tuple $(H \cap \alauto{s_1}{t_1})  \cdot T_{n-1}(s_{2},s_{2},s_{3},s_{3})$ is not \pol{\Gs^+}-separable by Lemma~\ref{lem:tuptriv} since $\veps \in \alauto{s_2}{s_2} \cap \alauto{s_3}{s_3}$. We now assume that at least one transition is used. We get a triple $(q_1,q_2,q_3) \in Q^3$, a word $w' \in A^*$ and $x \in A\cup \{\veps\}$ such that we have $w = w'x$, $w' \in \lauto{\Bs^+_{S}}{(s_1,s_2,s_3)}{(q_1,q_2,q_3)}$ and $((q_1,q_2,q_3),x,(t_1,t_2,t_3)) \in \gamma^+_S$. Since $H$ is a group language, it is recognized by a morphism $\alpha: A^* \to G$ into a finite group $G$. Let $H'= \alpha\inv(\alpha(w'))$. Clearly, $H'$ is a group language and $w'\in H'\cap\lauto{\Bs^+_{S}}{(s_1,s_2,s_3)}{(q_1,q_2,q_3)}$. Thus, induction yields that $H' \cap \alauto{s_1}{q_1} \neq \emptyset$ and, if $n \geq 2$, the $n$-tuple $(H' \cap \alauto{s_1}{q_1})  \cdot T_{n-1}(s_{2},q_{2},s_{3},q_{3})$ is not \pol{\Gs^+}-separable. We now consider two cases depending on $x \in A \cup \{\veps\}$.
	
	Assume first that $x = a \in A$: we have $((q_1,q_2,q_3),a,(t_1,t_2,t_3)) \in \gamma^+_S$. By definition, it follows that $(q_i,a,t_i) \in \delta$ for $i = \{1,2,3\}$.  Observe that $(H' \cap \alauto{s_1}{q_1})a \subseteq H \cap \alauto{s_1}{t_1}$. Indeed, if $u \in (H' \cap \alauto{s_1}{q_1})a$, then $u = u'a$ where $u' \in H'$ and $u' \in \alauto{s_1}{q_1}$. Since $H' = \alpha\inv(\alpha(w'))$, the hypothesis that $u' \in H'$ yields $\alpha(u) = \alpha(u'a) = \alpha(w'a) = \alpha(w)$ which implies that $u \in H$ since $w \in H$ and $H$ is recognized by $\alpha$. Moreover, since  $u' \in \alauto{s_1}{q_1}$ and $(q_1,a,t_1) \in \delta$, we get $u = u'a \in \alauto{s_1}{t_1}$. Altogether, this yields $u \in  H \cap \alauto{s_1}{t_1}$ as desired. Since we already know that $H' \cap \alauto{s_1}{q_1} \neq \emptyset$, we get $H \cap \alauto{s_1}{t_1} \neq \emptyset$. Moreover, if $n \geq 2$, since $(q_2,a,t_2),(q_3,a,t_3) \in \delta$, Lemma~\ref{lem:tuptriv} yields that $(\{a\})  \cdot T_{n-1}(q_{2},t_{2},q_{3},t_{3})$ is not \pol{\Gs^+}-separable. Hence, since we already know that $(H' \cap \alauto{s_1}{q_1})  \cdot T_{n-1}(s_{2},q_{2},s_{3},q_{3})$ is not \pol{\Gs^+}-separable and $(H' \cap \alauto{s_1}{q_1})a \subseteq H \cap \alauto{s_1}{t_1}$, it follows from  Lemma~\ref{lem:tupconcat} that $(H \cap \alauto{s_1}{t_1})  \cdot T_{n-1}(s_{2},t_{2},s_{3},t_{3})$ is not \pol{\Gs^+}-separable.
	
	Finally, assume that $x = \veps$: we have  $((q_1,q_2,q_3),\veps,(t_1,t_2,t_3)) \in \gamma^+_S$.  By definition, it follows that $q_1 = t_1$, $(q_2,t_2,q_3,t_3) \in S$ and there exists a nonempty word $y \in A^+$  which belongs to  \alauto{q_1}{q_1}, \alauto{q_2}{q_2}, \alauto{q_3}{q_3}, \alauto{t_2}{t_2} and \alauto{t_3}{t_3}. Since $x = \veps$, we have $w = w'$. Hence, since $w \in H$ and $H$ is recognized by $\alpha$, we obtain that $H' = \alpha(\alpha\inv(w')) \subseteq H$. Since $H' \cap \alauto{s_1}{q_1} \neq \emptyset$ and $q_1 = t_1$, we get $H \cap \alauto{s_1}{t_1} \neq \emptyset$. We now assume that $n \geq 2$. Since $G$ is a finite group, there exists $k \geq 1$ such that $\alpha(y^k) = 1_G$. We write $z = y^k$. By hypothesis on $y$, we also have $z \in \alauto{q_1}{q_1}$. It follows that $z^+ \subseteq \alpha\inv(1_G) \cap \alauto{q_1}{q_1}$. Additionally, since $z$ belongs to \alauto{q_2}{q_2}, \alauto{q_3}{q_3}, \alauto{t_2}{t_2} and \alauto{t_3}{t_3}, we know that $z^+\alauto{q_2}{t_2} z^+ \subseteq \alauto{q_2}{t_2}$ and $z^+\alauto{q_3}{t_3} z^+ \subseteq \alauto{q_3}{t_3}$. Since $(q_2,t_2,q_3,t_3) \in S$, it follows from induction on $n$ in Proposition~\ref{prop:gppsound} that the $(n-1)$-tuple $T_{n-1}(q_2,t_2,q_3,t_3)$ is not \pol{\Gs^+}-separable under \Gs-control. Altogether, we obtain from Lemma~\ref{lem:pgpsound} that the $n$-tuple $(\alpha\inv(1_G) \cap \alauto{q_1}{q_1})  \cdot T_{n-1}(q_2,t_2,q_3,t_3)$ is not \pol{\Gs^+}-separable. Finally, since $q_1 = t_1$ and $H' \subseteq H$, one may verify that  $(H' \cap \alauto{s_1}{q_1})(\alpha\inv(1_G) \cap \alauto{q_1}{q_1}) \subseteq (H \cap \alauto{s_1}{t_1})$. Since we already know that $(H' \cap \alauto{s_1}{q_1})  \cdot T_{n-1}(s_{2},q_{2},s_{3},q_{3})$ is not \pol{\Gs^+}-separable, Lemma~\ref{lem:tupconcat} yields that $(H \cap \alauto{s_1}{t_1}) \cdot T_{n-1}(s_{2},t_{2},s_{3},t_{3})$ is not \pol{\Gs^+}-separable.
\end{proof}

We may now complete the proof of Proposition~\ref{prop:gppsound}. By symmetry, we only treat the case when $n$ is odd and leave the even case to the reader. Let $(q_1,r_1,q_2,r_2) \in S$, we have to prove that $T_n(q_1,r_1,q_2,r_2)$ is not \pol{\Gs^+}-separable under \Gs-control. Hence, we fix $H \in \Gs$ such that $\veps \in H$ and prove $H\cap T_n(q_1,r_1,q_2,r_2)$ is not \pol{\Gs^+}-separable. Since $S$ is $(\bpoln,+)$-sound, we have $\betg(S) = S$ which implies that $(q_1,r_1,q_2,r_2) \in \betg(S)$. Hence, it follows from~\eqref{eq:gpauto} that $\{\veps\}$ is not \Gs-separable from $\lauto{\Bs^+_{S}}{(q_2,q_1,q_2)}{(r_2,r_1,r_2)}$. Since $H \in \Gs$ and $\veps \in H$, it follows that $H \cap \lauto{\Bs^+_{S}}{(q_2,q_1,q_2)}{(r_2,r_1,r_2)} \neq \emptyset$. If $n = 1$,  Lemma~\ref{lem:gpsinduc} yields $H \cap \alauto{q_2}{r_2} \neq \emptyset$. Since $T_1(q_1,r_1,q_2,r_2) = (\alauto{q_2}{r_2})$, we get that $H \cap T_1(q_1,r_1,q_2,r_2)$ is not \pol{\Gs^+}-separable as desired. If $n \geq 2$, then Lemma~\ref{lem:gpsinduc} implies that $(H \cap \alauto{s_1}{t_1})  \cdot T_{n-1}(s_{2},t_{2},s_{3},t_{3})$ is not \pol{\Gs^+}-separable. Thus, since $H \in \Gs \subseteq \pol{\Gs^+}$, one may verify that the $n$-tuple $(H \cap \alauto{q_2}{r_2}) \cdot (H \cap T_{n-1}(q_1,r_1,q_2,r_2))$ is not \pol{\Gs^+}-separable. By definition, this exactly says that $H \cap T_n(q_1,r_1,q_2,r_2)$ is not \pol{\Gs^+}-separable, completing the proof.

\medskip
\noindent
{\bf Second part: $\gnsbgp \subseteq S$.} Consider an arbitrary set $R \subseteq Q^4$. We say that $R$ is multiplication-closed to indicate that for every $(q,r,s,t) \in R$ and $(q',r',s',t') \in R$, if $r = q'$ and $t = s'$, then $(q,r',s,t') \in R$. Moreover, we say that an arbitrary set $R \subseteq Q^4$ is \emph{good} if it is multiplication-closed and there are $L \in \Gs$ such $\veps \in L$ and a \bpol{\Gs^+}-cover \Kb of $L$ which is separating for $R$. 

\begin{restatable}{proposition}{pbgcomp}\label{prop:bgpcomp}
	Let $R \subseteq Q^4$. If $R$ is good, then $\betg(R)$ is good as well.
\end{restatable}

We use Proposition~\ref{prop:bgpcomp} to complete the proof. Let $S_0 = Q^4$ and $S_i = \betg(S_{i-1})$ for $i \geq 1$. By Lemma~\ref{lem:pginc}, we have $S_0 \supseteq S_1 \subseteq S_2 \supseteq \cdots$ and the is $n \in \nat$ such that $S_n$ is the greatest $(\bpoln,+)$-sound subset for \Gs and \As, \emph{i.e.} such that $S_n = S$. Since $S_0$ is good (it is clearly multiplication-closed and $\{A^*\}$ is a \bpol{\Gs^+}-cover of $A^* \in \Gs$ which is separating for $S_0 = Q^4$), Proposition~\ref{prop:bgpcomp} implies that $S_i$ is good for all $i \in \nat$.  Hence, $S = S_n$ is good. We get  $L \in \Gs$ such $\veps \in L$ and a \bpol{\Gs^+}-cover \Kb of $L$ which is separating for $S$. By Lemma~\ref{lem:gautosep}, this yields $\gnsbgp \subseteq S$ as desired.

\medskip

We turn to Proposition~\ref{prop:bgcomp}. Let $R \subseteq Q^4$ be a good set. We have to prove that $\betg(R)$ is multiplication-closed and build $L \in \Gs$ such $\veps \in L$ and a \bpol{\Gs^+}-cover \Kb of $L$ which is separating for $\betg(R)$. This proves that $\betg(R)$ is good as desired. Let us first prove that $\betg(R)$ is multiplication-closed (we use the hypothesis that $R$ is good).

\begin{lemma} \label{lem:mclos}
	The set $\betg(R) \subseteq Q^4$ is multiplication-closed.
\end{lemma}

\begin{proof}
	Let $(q,r,s,t) \in \betg(R)$ and $(q',r',s',t') \in \betg(R)$ such that $r = q'$ and $t = s'$. We need to prove that $(q,r',s,t') \in \betg(R)$. By~\eqref{eq:gpauto} in the definition, this boils down to proving that $\{\veps\}$ is \emph{not} \Gs-separable from \lauto{\Bs^+_R}{(s,q,s)}{(t',r',t')} and \lauto{\Bs^+_R}{(q,s,q)}{(r',t',r')}. By symmetry, we only prove the former. By hypothesis on $(q,r,s,t)$ and $(q',r',s',t')$, we get from~\eqref{eq:gpauto} that $\{\veps\}$ is \emph{not} \Gs-separable from both \lauto{\Bs^+_R}{(s,q,s)}{(t,r,t)} and \lauto{\Bs^+_R}{(s',q',s')}{(t',r',t')}. Since \Gs is a \vari it then follows from Lemma~\ref{lem:tupconcat} that $\{\veps\}$ is not \Gs-separable from the concatenation $\lauto{\Bs^+_R}{(s,q,s)}{(t,r,t)}\lauto{\Bs^+_R}{(s',q',s')}{(t',r',t')}$. Finally, since $(t,r,t) = (s',q',s')$, we know that $\lauto{\Bs^+_R}{(s,q,s)}{(t,r,t)}\lauto{\Bs^+_R}{(s',q',s')}{(t',r',t')} \subseteq \lauto{\Bs^+_R}{(s,q,s)}{(t',r',t')}$. We conclude that $\{\veps\}$ is \emph{not} \Gs-separable from both \lauto{\Bs^+_R}{(s,q,s)}{(t',r',t')} as desired.
\end{proof}

We now build $L \in \Gs$ such that $\veps \in L$ (this part is independent from our hypothesis on $R$).

\begin{restatable}{lemma}{pgbl}\label{lem:pgbl}
	There exists $L \in \Gs$ such that $\veps \in L$ and for every $(q,r,s,t) \in Q^4$, if $\lauto{\Bs^+_R}{(q,s,q)}{(r,t,r)} \cap L \neq \emptyset$ and $\lauto{\Bs^+_R}{(s,q,s)}{(t,r,t)} \cap L \neq \emptyset$, then $(q,r,s,t) \in \betg(R)$.
\end{restatable}

\begin{proof}
	Let \Hb be the \emph{finite} set of all languages recognized by $\Bs^+_R$ such that $\{\veps\}$ is \Gs-separable from $H$. For every $H \in \Hb$, there exists $L_H \in \Gs$ such that $\veps \in L_H$ and $L_H \cap H = \emptyset$. We define $L = \bigcap_{H \in \Hb} L_H \in \Gs$. It is clear that $\veps \in L$. Moreover, given $(q,r,s,t) \in Q^4$, if $\lauto{\Bs^+_R}{(q,s,q)}{(r,t,r)} \cap L \neq \emptyset$ and $\lauto{\Bs^+_R}{(s,q,s)}{(t,r,t)} \cap L \neq \emptyset$, it follows from the definition of $L$ that $\{\veps\}$ is not \Gs-separable from both \lauto{\Bs^+_R}{(q,s,q)}{(r,t,r)} and \lauto{\Bs^+_R}{(s,q,s)}{(t,r,t)}. It then follows from~\eqref{eq:gpauto} in the definition of \betg that $(q,r,s,t) \in \betg(R)$.
\end{proof}

We fix $L \in \Gs$ as described in Lemma~\ref{lem:pgbl} for the remainder of the proof. We now build the \bpol{\Gs^+}-cover \Kb of $L$ using the hypothesis that $R$ is good and Proposition~\ref{prop:pgpcov}.

\begin{restatable}{lemma}{pgrun}\label{lem:pgrun}
	For all $(q,r) \in Q^2$, there is $H_{q,r} \in \bpol{\Gs^+}$ such that $\alauto{q}{r}\cap L\subseteq H_{q,r}$ and for all pairs $(s,t) \in Q^2$, if $\alauto{s}{t} \cap H_{q,r} \neq \emptyset$ then $\lauto{\Bs^+_R}{(q,s,q)}{(r,t,r)} \cap L \neq \emptyset$.
\end{restatable}

\begin{proof}
	Since $R$ is good, there are $U \in \Gs$ such that $\veps\in U$ and a \bpol{\Gs^+}-cover \Vb of $U$ which is separating for $R$. We use them to build $H_{q,r}$. Since $U \in \Gs$ and $\veps \in U$ Proposition~\ref{prop:pgpcov} yields a cover \Pb of $\alauto{q}{r} \cap L$ such that for each $P \in \Pb$, there exists a word $w_P \in \alauto{q}{r} \cap L$ and an \As-guarded decomposition $(w_1,\dots, w_{n+1})$ of $w_P$ for some $n \in \nat$ such that  $P = w_1 U \cdots w_nUw_{n+1}$ (if $n=0$, then $P=\{w_1\}$). Now, for every $P \in \Pb$, we build a \bpol{\Gs^+}-cover $\Kb_P$ of $P$ from the cover $\Vb$ of $U$.  Let $(w_1,\dots, w_{n+1})$ be the \As-guarded decomposition of $w_P$ such that $P = w_1 U \cdots w_nUw_{n+1}$ (in particular, this means that $P$ is of the form $U_0a_1U_1 \cdots a_mU_m$ where $a_1 \cdots a_m = w_1 \cdots w_n$ and $U_i = U$ or $U_i = \{\veps\}$ for each $i \leq m$). By definition, \Vb is a \bpol{\Gs^+}-cover of $U \in \Gs \subseteq \pol{\Gs^+}$. Moreover, we have $\{\veps\} \in \Gs^+ \subseteq \pol{\Gs^+}$ by definition of $\Gs^+$ and $\{\{\veps\}\}$ is a \bpol{\Gs^+}-cover of $\{\veps\}$. Hence, Proposition~\ref{prop:bconcat} yields a \bpol{\Gs^+}-cover $\Kb_P$ of $P = w_1 U \cdots w_nUw_{n+1}$ such that for every $K \in \Kb_P$, there exist $V_1,\dots,V_n \in \Vb$ such that $K \subseteq w_1 V_1 \cdots w_nV_nw_{n+1}$. We define $H_{q,r}$ as the union of all languages $K$ such that $K \in \Kb_P$ for some $P \in \Pb$ and $\alauto{q}{r} \cap K \neq \emptyset$. Clearly, $H_{q,r} \in \bpol{\Gs^+}$. Moreover, since \Pb is a cover of $\alauto{q}{r} \cap L$, and $\Kb_P$ is a cover of $P$ for each $P \in \Pb$, it is clear that $\alauto{q}{r}\cap L\subseteq H_{q,r}$. We now fix $(s,t) \in Q^2$ such that $\alauto{s}{t} \cap H_{q,r} \neq \emptyset$ and show that $\lauto{\Bs^+_R}{(q,s,q)}{(r,t,r)} \cap L \neq \emptyset$. By definition of $H_{q,r}$, we get $P \in \Pb$ and $K \in \Kb_P$ such that $\alauto{q}{r} \cap K\neq \emptyset$ and $\alauto{s}{t}\cap K \neq \emptyset$. By definition, $P = w_1 U \cdots w_nUw_{n+1}$ where $(w_1,\dots, w_{n+1})$ is an \As-guarded decomposition of $w_P \in \alauto{q}{r} \cap L$. We use $w_P$ to build a new word $w' \in \lauto{\Bs^+_R}{(q,s,q)}{(r,t,r)} \cap L$. 
	
	We fix $x \in \alauto{s}{t}\cap K$ and $y \in \alauto{q}{r} \cap K$. Since $w_P=w_1 \cdots w_{n+1}$ and $w_P \in \alauto{q}{r}$, we may decompose the corresponding run in \As: we get $p_0,\dots,p_{n+1} \in Q$ such that $p_0 = q$, $p_{n+1} = r$ and $w_i \in \alauto{p_{i-1}}{p_{i}}$ for $1 \leq i \leq n+1$. Moreover, since $K \in \Kb_P$, we have $K \subseteq  w_1 V_1 \cdots w_nV_nw_{n+1}$ for $V_1, \dots,V_n \in \Vb$ (if $n = 0$, then $K \subseteq \{w_1\}$). Since $x,y \in K$, we get $x_i,y_i \in V_i$ for $1 \leq i \leq n$ such that $x = w_1x_1 \cdots w_n x_nw_{n+1}$ and $y = w_1y_1 \cdots w_n y_nw_{n+1}$. Since $x \in \alauto{s}{t}$, we get $s_1,t_1,\dots,s_{n+1},t_{n+1} \in Q$ where $s_1 = s$, $t_{n+1} = t$, $w_i \in \alauto{s_i}{t_i}$ for $1 \leq i \leq n+1$ and $x_i \in \alauto{t_i}{s_{i+1}}$ for $1 \leq i \leq n$. Symmetrically, since $y \in \alauto{q}{r}$, we get $q_1,r_1,\dots,q_{n+1},r_{n+1} \in Q$ with $q_1 = q$, $r_{n+1} = r$, $w_i \in \alauto{q_i}{r_i}$ for $1 \leq i \leq n+1$, and $y_i \in \alauto{r_i}{q_{i+1}}$ for $1 \leq i \leq n$. First, note that when $n = 0$, we have $w_P = w_1$ and the above implies that $w_P \in \alauto{q}{r}$ and $w_P \in \alauto{s}{t}$. Thus, $w_P \in \lauto{\Bs^+_R}{(q,s,q)}{(r,t,r)}$ by definition of the labeled transition in $\Bs^+_R$. This concludes the proof since we also know that $w_P \in L$. We now assume that $n \geq 1$.

	By hypothesis, $(w_1,\dots, w_{n+1})$ is an \As-guarded decomposition. Hence, for $1 \leq i \leq n$, we get $z_i \in A^+$ which is a right \As-loop for $w_i$ and a left \As-loop for $w_{i+1}$. Let $\alpha: A^* \to G$ be a morphism into a finite group $G$ recognizing both $L$ and $U$ (recall that $L$ and $U$ are group languages). Since $g$ is a finite group, there exists $k \geq 1$ such that for each $1 \leq i \leq n$, we have $\alpha(z_i^k) = 1_G$. We let $u_i = z_i^k$ for $1 \leq i \leq n$. One may verify that $u_i$ remains a right \As-loop for $w_i$ and a left \As-loop for $w_{i+1}$. Moreover,  since $\alpha(u_i) = 1_G$, we know that $u_i \in U$ (recall that $\veps \in U$ and $U$ is recognized by $\alpha$). We let $w'_1 = w_1u_1$, $w'_{n+1} =u_nw_{n+1}$ and $w'_i=u_{i-1}w_iu_i$ for $2 \leq i \leq n$. Finally, we let $w' = w'_1 \cdots w'_nw'_{n+1}$ and show that $w' \in L \cap \lauto{\Bs^+_R}{(q,s,q)}{(r,t,r)}$ which completes the proof. First, since $\alpha(u_i) = 1_G$ for $1 \leq i \leq n$, it is immediate that $\alpha(w') = \alpha(w_1 \cdots w_nw_{n+1}) = \alpha(w_P)$. Since $w_P \in L$ which is recognized by $\alpha$, we get $w' \in L$.

	We now concentrate on proving that $w' \in \lauto{\Bs^+_R}{(q,s,q)}{(r,t,r)}$. For $1 \leq i \leq n+1$, we know that $w_i$ belongs to \alauto{p_{i-1}}{p_{i}}, \alauto{s_i}{t_i} and \alauto{q_i}{r_i}. Hence, one may verify from the definition of left/right \As-loops that there are $p'_0,\dots,p'_{n+1} \in Q$, $s'_1,t'_1,\dots,s'_{n+1},t'_{n+1} \in Q$ and $q'_1,r'_1,\dots,q'_{n+1},r'_{n+1} \in Q$ such that,
	\begin{itemize}
		\item $p'_0 = p_0 = q$, $p'_{n+1} = p_{n+1} = r$, $w'_i \in \alauto{p'_{i-1}}{p'_{i}}$ for $1 \leq i \leq n+1$ and $u_i \in \alauto{p'_i}{p'_i}$ for $1 \leq i \leq n$.
		\item $s'_0 = s_0 = s$, $t'_{n+1} = t_{n+1} = t$, $w'_i \in \alauto{s'_{i}}{t'_{i}}$ for $1 \leq i \leq n+1$ and we have $u_i \in \alauto{t'_i}{t'_i} \cap \alauto{t'_i}{t_i} \cap \alauto{s_{i+1}}{s'_{i+1}}\cap \alauto{s'_{i+1}}{s'_{i+1}}$ for $1 \leq i \leq n$.
		\item $q'_0 = q_0 = q$, $r'_{n+1} = r_{n+1} = r$, $w'_i \in \alauto{q'_{i}}{r'_{i}}$ for $1 \leq i \leq n+1$ and we have $u_i \in \alauto{r'_i}{r'_i} \cap \alauto{r'_i}{r_i} \cap \alauto{q_{i+1}}{q'_{i+1}}\cap \alauto{q'_{i+1}}{q'_{i+1}}$ for $1 \leq i \leq n$.
	\end{itemize}
	By definition of the labeled transitions in the \nfa $\Bs^+_R$, it is straightforward to verify that we have $w'_i \in \lauto{\Bs^+_R}{(p'_{i-1},s'_i,q'_i)}{(p'_{i},t'_i,r'_i)}$ for $1 \leq i \leq n+1$. We now prove the following fact. 
	
	\begin{fact} \label{fct:thezi}
		For $1 \leq i \leq n$, we have $((p'_{i},t'_i,r'_i),\veps,(p'_{i},s'_{i+1},q'_{i+1})) \in \gamma^+_R$.
	\end{fact}
	
	\begin{proof}
		We fix $i$ for the proof. Since we know that $u_i \in A^+$ belongs to  \alauto{p'_i}{p'_i},  \alauto{t'_i}{t'_i},  \alauto{r'_i}{r'_i},  \alauto{s'_{i+1}}{s'_{i+1}} and \alauto{q'_{i+1}}{q'_{i+1}}, it suffices to prove that $(t'_i,s'_{i+1},r'_i,q'_{i+1}) \in R$. This will imply that $((p'_{i},t'_i,r'_i),\veps,(p'_{i},s'_{i+1},q'_{i+1})) \in \gamma^+_R$ by definition of $\gamma^+_R$. Recall that $x_i \in \alauto{t_i}{s_{i+1}}$, $y_i \in \alauto{r_i}{q_{i+1}}$ and $x_i,y_i \in V_i$. Since $V_i \in \Vb$ which is \emph{separating} for $R$, it follows that $(t_i,s_{i+1},r_i,q_{i+1}) \in R$. Moreover, $u_i \in U$ which yields $V \in \Vb$ such that $u_i \in V$ since \Vb is a cover of $U$. Hence, since $u_i \in \alauto{t'_i}{t_i}$ and $u_i\in \alauto{r'_i}{r_i}$. The hypothesis that \Vb is separating for $R$ also yields $(t'_i,t_i,r'_i,r_i) \in R$. Symmetrically, one may use the hypotheses that $u_i \in \alauto{s_{i+1}}{s'_{i+1}}$ and $u_i \in \alauto{q_{i+1}}{q'_{i+1}}$ to verify that $(s_{i+1},s'_{i+1},q_{i+1},q'_{i+1}) \in R$. Altogether, since $R$ is multiplication-closed, we get $(t'_i,s'_{i+1},r'_i,q'_{i+1}) \in R$ as desired.
	\end{proof}

	In view of Fact~\ref{fct:thezi}, we obtain $w' = w'_1 \cdots w'_nw'_{n+1} \in \lauto{\Bs^+_R}{(p'_{0},s'_1,q'_1)}{(p'_{n+1},t'_{n+1},r'_{n+1})}$. This exactly says that $w' \in \lauto{\Bs^+_R}{(q,s,q)}{(r,t,r)}$ which completes the proof.
\end{proof}

We may now build \Kb. Let $\Hb = \{H_{q,r} \mid (q,r) \in Q^2\}$. Consider the following equivalence $\sim$ defined on $L$: given $u,v\in L$, we let $u \sim v$ if and only if $u \in H_{q,r} \Leftrightarrow v \in H_{q,r}$ for every $(q,r) \in Q^2$. We let \Kb as the partition of $L$ into $\sim$-classes. Clearly, each $K \in \Kb$ is a Boolean combination involving the languages in \Hb (which belong to \bpol{\Gs^+}) and $L \in \Gs$. Hence, $\Kb$ is a \bpol{\Gs^+}-cover of $L$. It remains to prove that it is separating for $\betg(R)$. Let $q,r,s,t \in Q$ and $K \in \Kb$ such that there are $u \in \alauto{q}{r} \cap K$ and $v \in \alauto{s}{t} \cap K$. By definition of \Kb, we have $u,v \in L$ and $u \sim v$. In particular, we have $u \in \alauto{q}{r} \cap L$ which yields $u \in H_{q,r}$ by definition in Lemma~\ref{lem:pgrun}. Together with $u \sim v$, this yields $v \in H_{q,r}$. Hence, $\alauto{s}{t} \cap H_{q,r} \neq \emptyset$ and Lemma~\ref{lem:pgrun} yields $\lauto{\Bs^+_R}{(q,s,q)}{(r,t,r)} \cap L \neq \emptyset$. One may now use a symmetrical argument to obtain $\lauto{\Bs^+_R}{(s,q,s)}{(t,r,t)} \cap L \neq \emptyset$. By definition of $L$ in Lemma~\ref{lem:pgbl}, this yields $(q,r,s,t) \in \tautg(R)$, completing the proof.

\end{document}